\documentclass[12pt]{article}
\usepackage[utf8]{inputenc}
\usepackage{amsfonts}
\usepackage{amsgen,amsmath,amstext,amsbsy,amsopn,amssymb,subfigure}
\usepackage[dvips]{graphicx}
\usepackage{comment}

\usepackage[shortlabels]{enumitem}
\usepackage{natbib}
\usepackage{booktabs}
\usepackage{amsthm}
\usepackage{hyperref}
\usepackage{blkarray}
\usepackage{algorithm}
\usepackage{algorithmic}
\usepackage{url}
\usepackage{longtable}
\usepackage{stmaryrd}
\usepackage{mwe}
\DeclareMathAlphabet\mathbfcal{OMS}{cmsy}{b}{n}

\newcommand{\diag}{{\rm diag}}
\newcommand{\dist}{{\rm dist}}
\newcommand{\subscript}[2]{$#1 _ #2$}
\makeatletter
\newcommand*{\rom}[1]{\expandafter\@slowromancap\romannumeral #1@}
\makeatother

\usepackage[dvipsnames]{xcolor}

\newcommand{\cG}{\mathcal{G}}
\newcommand{\cX}{\mathcal{X}}

\newcommand{\cL}{\mathcal{L}}

\newcommand{\cT}{\mathcal{T}}
\newcommand{\cH}{\mathcal{H}}
\newcommand{\cR}{\mathcal{R}}
\newcommand{\cF}{\mathcal{F}}
\newcommand{\cE}{\mathcal{E}}
\newcommand{\cM}{\mathcal{M}}

\newcommand{\Var}{\text{Var}}
\newcommand{\bA}{\boldsymbol{A}}
\newcommand{\bB}{\boldsymbol{B}}

\newcommand{\bZ}{\boldsymbol{Z}}

\newcommand{\bY}{\boldsymbol{Y}}

\newcommand{\bK}{\boldsymbol{K}}
\newcommand{\bSigma}{\boldsymbol{\Sigma}}
\newcommand{\bLambda}{\boldsymbol{\Lambda}}

\newcommand{\bU}{\boldsymbol{U}}

\newcommand{\bW}{\boldsymbol{W}}
\newcommand{\bI}{\boldsymbol{I}}

\newcommand{\bcX}{{\mathbfcal{X}}}
\newcommand{\bcA}{{\mathbfcal{A}}}
\newcommand{\bcY}{{\mathbfcal{Y}}}

\newcommand{\bcZ}{{\mathbfcal{Z}}}

\newcommand{\bbK}{{\mathbb{K}}}
\newcommand{\bbE}{{\mathbb{E}}}
\newcommand{\bbP}{{\mathbb{P}}}
\newcommand{\bbR}{{\mathbb{R}}}
\newcommand{\bbS}{{\mathbb{S}}}
\newcommand{\bbG}{{\mathbb{G}}}
\newcommand{\bbO}{{\mathbb{O}}}
\newcommand{\argmin}{\mathop{\rm arg\min}}
\newcommand{\argmax}{\mathop{\rm arg\max}}
\newcommand{\tF}{{\rm F}}
\newcommand{\SVD}{{\rm SVD}}
\newcommand{\tr}{{\rm tr}}

\newtheorem{Theorem}{Theorem}
\newtheorem{Assumption}{Assumption}

\newtheorem{Lemma}{Lemma}
\newtheorem{Remark}{Remark}
\newtheorem{Corollary}{Corollary}

\newtheorem{Proposition}{Proposition}

\title{Guaranteed Functional Tensor Singular Value Decomposition\footnote{Rungang Han is a postdoc fellow, Department of Statistical Science, Duke University, Durham, NC 27710, E-mail: rungang.han@duke.edu; Pixu Shi is Assistant Professor, Department of Biostatistics \& Bioinformatics, Duke University, E-mail: pixu.shi@duke.edu; Anru R. Zhang is Eugene Anson Stead, Jr. M.D. Associate Professor, Department of Biostatistics \& Bioinformatics and Department of Computer Science, Duke University, E-mail: anru.zhang@duke.edu. The research of R. Han and A. R. Zhang was supported in part by NSF Grant CAREER-1944904 and NIH Grant R01 GM131399.}
}
\author{Rungang Han, ~ Pixu Shi, ~ and ~ Anru R. Zhang}
\date{}

\begin{document}

	\maketitle
	
	\begin{abstract}
		This paper introduces the functional tensor singular value decomposition (FTSVD), a novel dimension reduction framework for tensors with one functional mode and several tabular modes. The problem is motivated by high-order longitudinal data analysis. Our model assumes the observed data to be a random realization of an approximate CP low-rank functional tensor measured on a discrete time grid. Incorporating tensor algebra and the theory of Reproducing Kernel Hilbert Space (RKHS), we propose a novel RKHS-based constrained power iteration with spectral initialization. Our method can successfully estimate both singular vectors and functions of the low-rank structure in the observed data. With mild assumptions, we establish the non-asymptotic contractive error bounds for the proposed algorithm. The superiority of the proposed framework is demonstrated via extensive experiments on both simulated and real data.
	\end{abstract}
	
	\begin{sloppypar}
		
		\section{Introduction}\label{sec:intro}

		In recent decades, the analysis of tensor data has become an active research topic. Datasets in the form of high-order tensors or multiway arrays arise from various scientific applications, such as neuroimaging~\citep{zhang2019tensor}, microscopy imaging \citep{zhang2020denoising}, and longitudinal microbiome study \citep{martino2021context}. Such high-order data pose significant challenges in both theoretical analysis and computation implementation due to their complicated structures and a large number of entries involved, making it inappropriate to extend many existing methods for matrices to the analysis of tensor data.
		
		In real applications, the different modes of tensor datasets (or different directions that the tensor arrays align in) can come in various formats. Two prominent formats are (1) {\it tabular modes}, such as subject ID, genomics ID, treatments, where shuffling the indices does not essentially change the data structure; (2) {\it functional modes}, such as time, location, spectrum (in the hyperspectral imaging), where the order of indices exhibits structures such as continuity. Here we provide two scenarios to illustrate tensor data with both types of modes.
		
		\begin{itemize}
			\item \emph{Multivariate functional data analysis.} Data with different functional features appear in various applications and can be formatted into a tensor with two tabular modes representing units and variables, and a functional mode representing time point. One example is the longitudinal microbiome studies, where microbiome samples are taken from multiple subjects (units) at multiple time points to study the abundance of bacteria (variables) over time. Depending on what taxonomic level is being studied, there can be hundreds and thousands of bacterial taxa in the feature mode and many of these taxa have strong correlation in their abundance.
			
			\item \emph{Dynamic networks.} In network analysis, one often observes multiple snapshots of dynamic networks. The adjacency matrix in each snapshot can be stacked together into an adjacency tensor, where two tabular modes correspond to the vertices of the network, and a functional mode corresponds to time. The dimension of the tabular modes, i.e., number of nodes, can be as large as thousands, but the hidden community structure of the network can often be modeled by low-rankness of adjacency tensor.
		\end{itemize}
		Just as in the examples above, the tabular modes may be very high-dimensional. The functional mode is often continuous, which is essentially infinite-dimensional. Therefore, dimension reduction tools are crucial for the visualization and analysis of high-dimensional tensor data with both tabular modes and functional modes. Various lines of methods can be applied, whereas one prominent framework is the {\it multivariate/multilevel functional principal component analysis (MFPCA)} in the functional data analysis (FDA) literature. The classic MFPCA methods often focus on a moderately large number of variables and dimension reduction on the functional mode. The resulting number of estimated eigenfunctions is often in the same order as the number of variables, making the eigenfunctions hard to interpret when the variable mode is high-dimensional. In addition, most existing MFPCA frameworks focus on characterizing the covariance structure among functional variables at the population level and assuming the i.i.d. observations. Such assumption can be violated when dealing with heterogeneous samples. A more detailed discussion on the literature of multivariate functional analysis is given in Section \ref{sec:related-work}.
		Another important class of works utilizes the high-order structure of the tensor data via different types of low-rank tensor decomposition models including (sparse) CP low-rankness \citep{anandkumar2014guaranteed,sun2017provable}, (sparse) Tucker low-rankness \citep{zhang2018tensor,zhang2019optimal-statsvd}. While the low-rank models achieve efficient dimension reduction and are interpretable when the modes are tabular, they do not directly characterize or utilize the information included in the functional modes such as time and location. Therefore, new methods and theory are needed to tackle both the high-dimensional and functional aspects of functional tensor data.
		
		In this paper, we introduce a new framework for the dimension reduction of functional tensor data which we refer to as \emph{Functional Tensor Singular Value Decomposition} (FTSVD). Suppose $\bcY \in \bbR^{p_1 \times p_2 \times \cT}$ is the underlying function of interest, where $\cT \subset \bbR$ is some compact set on the real axis. For any fixed $(i,j)$ pair, $\bcY_{ijt}$ is a function of $t \in \cT$. We aim to decompose $\bcY$ into the following canonical polyadic (CP) format:
		\begin{equation}\label{eq:function-cp-approxi}
			\bcY = \bcX+\bcZ, ~~ \text{where} ~~ \bcX = \sum_{l=1}^r \lambda_l a_l \circ b_l \circ \xi_l ~~ \text{or equivalently}\quad \bcX_{ijt}=\sum_{l=1}^r\lambda_l (a_{l})_i(b_{l})_j\xi_{l}(t).
		\end{equation}
		Here, $a_l\in \mathbb{R}^{p_1}$, $b_l\in \mathbb{R}^{p_2}$ are singular vectors of tabular modes; $\xi_l: \cT\to \mathbb{R}$ is the singular function, which corresponds to the eigenfunction in functional PCA; $\bcZ$ is the remainder term of the rank-$r$ CP decomposition of $\bcY$. We assume the observed data $\widetilde{\bcY}$ to be $\bcY$ measured over a discrete grid $\{s_k\}_{k=1}^n \subset \cT$ with observational noise, i.e.,
		\begin{equation}\label{eq:model-observe-intro}
			\widetilde\bcY \in \bbR^{p_1\times p_2 \times n},\qquad \widetilde\bcY_{ijk} = \bcY_{ijs_k} + \varepsilon_{ijk}.
		\end{equation}
		Without loss of generality, we assume that $\mathcal{T} = [0,1]$ throughout the paper. 
		
		The proposed model setting is flexible and adaptive. For example, in the scenario of multivariate functional data analysis, different from many existing literature on (multivariate) functional PCA~\citep{happ2018multivariate,wang2020low}, which assumed the samples are i.i.d. distributed and try to estimate the covariance function, our model allows significant heterogeneity among samples (which is characterized by the singular vector of the sample mode). In the scenario of dynamic network, to model the temporal connectivity dynamics, one can assume the $r$ latent factors for all vertices $Z = [\sqrt{\lambda_1}a_1,\ldots,\sqrt{\lambda_r}a_r]$ are time-invariant, and the similarity (connectivity probability) matrix at time $s_k$ can be modeled as a weighted inner product: $Z\Lambda_{s_k} Z^\top$, where the weight matrix $\Lambda_{s_k} = \diag\left(\xi_1(s_k),\ldots,\xi_r(s_k)\right)$. The model can be reduced to the regular symmetric CP decomposition with $a_i = b_i$. Additional applications that can be formatted under this framework include (hyper)imaging analysis, online recommendation system, etc. 
		
		Departing from the existing literature, the proposed framework FTSVD can provide important new insights to the analysis of data tensors with both tabular and functional modes. Firstly, through the lens of functional tensor SVD, we can describe the trend of all variables $\bcX_{ij\cdot}$ over time using only a small set of nonparametric singular functions $\xi_l(t)$. Such low-dimensional representation of the functional mode is more interpretable than high-dimensional ones and may cast light on the driving factors behind the trend of variables. Secondly, the singular vectors $a_l$ and $b_l$ provide a low-dimensional depiction of the information in the tabular modes, which can be used to improve the efficiency and interpretability of subsequent analysis such as clustering, classification, and regression.

		\subsection{Contributions}\label{sec:contributions}
		
		The main contributions of this paper are three-fold. First, to the best of our knowledge, we are the first to study multivariate/multilevel functional data analysis through a lens of a low-rank tensor decomposition model. We are also the first to adopt the reproducing kernel Hilbert space (RKHS) framework for the singular function in the low-rank decomposition. Compared to recently emerged state-of-the-art tabular tensor models, the proposed model can better characterize the longitudinal behaviors of the data when there exist temporal modes. Compared to most classic FDA literature, our model does not assume the samples to be i.i.d. Instead, our method may provide new insight on how to analyze heteroskedastic functional effects for mixed-population samples. 
		
		Secondly, we propose a new power iteration algorithm for the estimation of our model. In each iteration step, we use projection to update the estimators for the tabular components and constrained/regularized empirical risk minimization (ERM) to update the estimators for the functional component, respectively. The novel inclusion of multiple-times ERM in our iterative algorithm imposes new challenges on both statistical theory and numeric optimization.
		
		Thirdly, we establish the finite-sample statistical error bounds for the estimation of both singular vectors and functions \eqref{ineq:error-bound-thm1}. Our error bound includes contribution from the remainder term of CP approximation (Eq. \eqref{eq:function-cp-approxi}) and observational noises (Eq. \eqref{eq:model-observe-intro}) relative to the signal level in $\bcX$. It also reflects the effect of time grid density and incoherence condition of the CP decomposition. To the best of our knowledge, this is the first contractive error bound for iterative algorithms in functional data analysis.
		
		\subsection{Related Work}\label{sec:related-work}
		
		This work is related to a series of papers on \emph{low-rank tensor decomposition} from statistics and machine learning, which aims at recovering the low-dimensional linear representations of the noisy tabular tensors~\citep{anandkumar2014tensor}. Commonly-used low-rank models include CP~\citep{anandkumar2014guaranteed} and Tucker~\citep{zhang2018tensor}
		where various power iteration methods were proposed on the singular vector (or principal component) estimation with statistical consistency guarantees. Motivated by real data applications in practice, related statistical models were further studied under non-Gaussian or missing-data scenarios. For example, \cite{hong2020generalized,han2021optimal} considered the likelihood-based generalized low-rank tensor decomposition under general exponential family.
		
		In addition, various dimension-reduced structures are imposed on the singular vectors or subspace of the low-rank tensor decomposition model to better capture other intrinsic properties of the data, such as sparsity~\citep{sun2017provable,zhang2019optimal-statsvd} and  blocking~\citep{han2020lloyd}. The most relevant paper to our work is the spatial/temporal structure, which is incorporated to characterize time/location-varying patterns for one or more of the tensor modes. There are various formulations for modeling such structures based on different real data motivations. For example, \cite{sun2019dynamic} considered a fusion structure and assumed the time-varying trend to be piece-wise constant; \cite{han2019isotonic} considered the multidimensional isotonic regression, where there are monotonous trends in each tensor mode. In most of such literature, the temporal structures are still assumed on tabular tensors and make the corresponding models restrictive. In contrast, the target of our framework is the underlying non-parametric functional structures, which reflects the trend of observed data over time.
		
		Another related topic is the \emph{functional data analysis} (FDA), which is a popular branch of statistical research. The readers are referred to books and surveys in \cite{ramsay2006functional,wang2016functional}. Starting from samples with univariate functions, traditional FDA involves various statistical learning tasks, such as regression~\citep{yao2005functional-regression,cai2012minimax}, covariance function estimation~\citep{rice1991estimating}, principal component analysis (PCA)~\citep{james2000principal}, etc. Under the multivariate regime, \cite{hasenstab2017multi} considered the multilevel/multidimensional functional data PCA. \cite{wang2020low} studied the low-rank covariance estimation for multidimensional functional data. \cite{fan2015functional} considered the functional additive regression method for high-dimensional functional regression. \cite{hu2021dynamic} considered the dimension reduction via dynamic principal subspace for multivariate functional data. \cite{chen2017modelling} studied the functional data analysis for scenarios where the observations are functional at each location. The works on multivariate functional principal component analysis \citep{allen2013multi,happ2018multivariate} also address dimension reduction of multivariate functional data. Among these literature, \cite{happ2018multivariate,wang2020low} are the most relevant to our work. In particular, \cite{wang2020low} focused on i.i.d. $p$-dimensional function objects, and they imposed tensor low-rankness on the corresponding $(2p)$-dimensional covariance functions; \cite{happ2018multivariate} considered the multivariate functional PCA for i.i.d. samples from some multivariate functional population and the data can be formulated as a typical functional tensor stacked by ``sample-feature-function.'' Both papers assumed i.i.d. samples and focused on the statistical inference in the covariance function. Compared to these methods, our model does not assume independence of the signal tensor along any mode. Such flexible modeling is able to encompass a range of problems on heterogeneous or heteroskedastic data. The resulting estimated singular vectors can also be used to perform many important tasks, such as clustering, classification, and regression.
		
	    Our framework is also related to a line of recent papers on high-order tensor factor analysis. For example, \cite{chen2021factor, han2020tensor} considered a factor model for dynamic tensor time series, where they assumed the factor tensors are (weakly) stationary without imposing any deterministic or structured time-varying trends. \cite{chen2020semiparametric} proposed a semi-parametric Tucker decomposition model, where the loadings on one or more modes can be approximated by smooth functions from an H{\"o}lder class. Our framework is in the spirit of CP decomposition, which allows for more general functional classes via RKHS theory. 
		
		\section{Low-rank Functional Tensor Decomposition Model}\label{sec:notation-model}
		
		\subsection{Notation and Preliminaries}\label{sec:notation}
		
		We use the lowercase letters, e.g., $a, b, u, v$, to denote scalars or vectors. For any $a, b \in \bbR$, let $a\wedge b$ and $a\vee b$ be the minimum and maximum of $a$ and $b$, respectively. For a vector $u \in \bbR^n$, $\|u\|_2$ denotes its Euclidean norm. The unit sphere in $\bbR^n$ is denoted as $\bbS^{n-1}$. Matrices are denoted as uppercase letters such as $\bA$ and $\bB$. In addition, we let $\bbO_{p,r}$ be the collection of all $p$-by-$r$ matrices with orthonormal columns: $\bbO_{p,r} := \{\bU \in \bbR^{p\times r}: \bU^\top \bU = \bI\}$, where $\bI$ is the identity matrix. We also denote $\bbO_r := \bbO_{r,r}$ as the set of all $r$-dimensional orthogonal matrices. Let $\lambda_1(\bA) \geq \lambda_2(\bA) \geq \cdots \geq 0$ be the singular values of $\bA$ in descending order and let $\SVD_r(\bA)$ be the matrix comprised of the top $r$ left singular vectors of $\bA$. Let $\|\bA\| = \lambda_1(\bA)$ and $\|\bA\|_\tF = \sqrt{\sum_{i=1}^{p_1}\sum_{j=1}^{p_2}\bA_{ij}^2} = \sqrt{\sum_{i=1}^{p_1\wedge p_2}\lambda_i^2(\bA)}$ be the matrix spectral and Frobenius norms, respectively. For any matrix $\bA=[a_1,\ldots, a_J]\in \mathbb{R}^{I\times J}$ and $\bB\in \mathbb{R}^{K\times L}$, the \emph{Kronecker product} is defined as the $(IK)$-by-$(JL)$ matrix $\bA\otimes\bB = [a_1\otimes \bB \cdots a_J\otimes \bB]$.
		
		Without loss of generality, we assume the function domain of interest is $\cT = [0, 1]$ and denote $\cL^2([0,1])$ as the functional space of all square-integrable functions on $[0,1]$, i.e., 
		$$
		\cL^2\left([0,1]\right) = \left\{f : [0,1]\rightarrow \bbR, \|f\|_{\cL^2}^2 < \infty\right\},\quad \text{where}\quad  \|f\|_{\cL^2}:=\left(\int_0^1 f^2(t)dt\right)^{1/2}.
		$$
		Denote the inner product of any two functions $f,g \in \cL^2([0,1])$ as $\langle f, g \rangle_{\cL^2} = \int_0^1 f(t)g(t)dt$. For a sequence $\{s_k\}_{k=1}^n$ in $[0,1]$, we denote $f_n := (f(s_1),\ldots, f(s_n)) \in \bbR^n$ and $\|f\|_n :=   \sqrt{\sum_{k=1}^n f^2(s_k)/n}$. We use $C, C_0, C_1, \ldots$ and $c, c_0, c_1,\ldots$ to represent generic large and small positive constants, respectively. The actual values of these generic symbols may differ from line to line. We denote $a\lesssim b$ if $a \leq Cb$ for a constant $C>0$ that does not depend on other model parameters; we say $a \asymp b$ if $a\lesssim b$ and $b\lesssim a$ both hold.
		
		Throughout this paper, the tensors are denoted by uppercase calligraphy letters, such as $\bcX,\bcY,\bcZ$. For a tabular tensor $\bcA \in \bbR^{p_1 \times p_2 \times p_3}$, let $\bcA_{ijk}$ be the $(i,j,k)$th entry. For any $u \in \bbR^{p_1}, v \in \bbR^{p_2}, w \in \bbR^{p_3}$, the mode-1 tensor-vector product is defined as:
		$\bcA \times_1 u \in \mathbb R^{p_2 \times p_3}$, $(\bcA \times_1 u)_{jk}:= \sum_{i=1}^{p_1} u_i \bcA_{ijk}.$ The mode-2 and mode-3 tensor-vector products, $\mathcal{A}\times_2 v$ and $\mathcal{A}\times_3 w$ for $v\in \mathbb{R}^{p_2}$ and $w\in\mathbb{R}^{p_3}$, can be defined in the parallel way. The multiplication along different modes is cumulative and commutative. For example, $\bcA \times_1 u \times_2 v \in \bbR^{p_3}$ with $\left(\bcA \times_1 u \times_2 v\right)_k = \sum_{i=1}^{p_1}\sum_{j=1}^{p_2} u_iv_j \bcA_{ijk}$; and $\bcA \times_1 u \times_2 v \times_3 w = \sum_{i=1}^{p_1}\sum_{j=1}^{p_2}\sum_{j=1}^{p_3} u_iv_jw_k \bcA_{ijk} \in \bbR$. We also introduce the matricization operator that transforms tensors to matrices. Particularly, the mode-$1$ matricization of $\bcA \in \bbR^{p_1 \times p_2 \times p_3}$ is defined as $\cM_1 (\bcA) \in \bbR^{p_1 \times (p_2p_3)}$, where $[\cM_1(\bcX)]_{i,j + p_2(k-1)} = \bcA_{ijk}.$ In other words, each row of $\cM_k(\bcA)$ is the vectorization of a mode-$k$ slice. 
		
		Next, we extend the notions of tabular tensors to tensors with a hybrid of tabular and functional modes. We still denote the tensors with hybrid modes by uppercase calligraphy letters. Suppose $\bcA = \{\bcA_{ijt}, i \in [p_1], j \in [p_2], t \in [0,1]\} \in \mathbb R^{p_1 \times p_2 \times [0,1]}$ is an order-3 tensor with two tabular modes and a functional mode. The tensor-vector products on the tabular modes are defined as convention: $\bcA \times_1 u \in \mathbb R^{p_2 \times [0,1]}$, $(\bcA \times_1 u)_{jt}:= \sum_{i=1}^n u_i \bcA_{ijt}.$ Suppose $f \in \cL^2([0,1])$, we define the tensor-functional product as $\bcA \times_3 f \in \mathbb R^{p_1 \times p_2}, (\bcA \times_3 f)_{ij}:= \int_0^1 \bcA_{ijt}f(t)dt.$ The multiple-mode tensor-vector/function multiplication can be similarly defined by combining the operators of single-mode multiplication. For example, $\bcA \times_1 u \times_2 v$ yields a function $\sum_{i=1}^{p_1}\sum_{j=1}^{p_2} u_iv_j \bcA_{ijt}$ and $\bcA \times_1 a \times_2 b \times_3 f$ yields a scalar $\sum_{i=1}^{p_1}\sum_{j=1}^{p_2} \int_0^1 u_iv_j \bcA_{ijt}f(t)dt$.
		
		Next, we provide the preliminaries for the reproducing kernel Hilbert space (RKHS). Consider a Hilbert space $\mathcal H \subset \cL^2([0,1])$ with the associated inner product $\langle \cdot, \cdot \rangle_{\cH}$. We assume there exists a continuous symmetric positive-semidefinite kernel function $\mathbb K:[0,1]\times [0,1] \rightarrow \mathbb R_{+}$ that satisfies the following RKHS conditions: (1) for any $s \in [0,1]$, $\mathbb K(\cdot,s) \in \cH$; (2) for each $g \in \cH$, $g(t) = \langle g, \mathbb K(\cdot, t) \rangle_\cH,~\forall t \in [0,1]$. By Mercer's Theorem~\citep{mercer1909xvi}, there exists an orthonormal basis $\{\phi_k\}_{k=1}^\infty$ of $\cL^2([0,1])$, such that $\mathbb K(\cdot,\cdot)$ admits the following eigendecomposition: $\mathbb K(s,t) = \sum_{k=1}^\infty \mu_k\phi_k(s)\phi_k(t),~~\forall s,t \in [0,1].$
		Here, $\mu_1 \geq \mu_2 \geq \cdots \geq 0$ are the non-negative eigenvalues of $\mathbb K$. We assume $\mathbb K$ has a finite trace norm: $\sum_{k=1}^\infty \mu_k < \infty$.  Then any $f \in \cH$ can be decomposed on $\{\phi_k\}_{k=1}^\infty$ as
		\begin{equation*}
			f(t) = \sum_{k=1}^\infty a_k\phi_k(t), \qquad a_k := \int_{0}^1 f(s)\phi_k(s)ds,  \qquad t \in [0,1].
		\end{equation*}
		The RKHS norm of $f$ can be then represented as $\|f\|_\cH = \sqrt{\langle f,f\rangle_\cH} = \sqrt{\sum_{k=1}^\infty a_k^2/\mu_k}$. Accordingly, for any two functions $f = \sum_{k=1}^\infty a_k \phi_k$ and $g = \sum_{k=1}^\infty b_k \phi_k$ in $\cH$, the inner product with respect to $\cH$ can be calculate as $\langle f, g \rangle_\cH = \sum_{k=1}^\infty a_kb_k/\mu_k$.
		
		Assume the following regularity condition on $\bbK$ throughout the paper. \begin{Assumption}\label{asmp:K-regularity}
			(1) $\bbK$ satisfies $\max\{\mu_1, \sup_{s\in [0,1]}\bbK(s,s)\} \leq 1$;\\ (2) there exists an absolute constant $C_\cH$, which only depends on $\cH$, such that
			\begin{equation}\label{ineq:H-cauchy}
				\|fg\|_\cH \leq C_\cH \|f\|_\cH \|g\|_\cH,\qquad \forall f,g \in \cH.
			\end{equation}
		\end{Assumption}
		Assumption \ref{asmp:K-regularity} (a) can be ensured by rescaling $\bbK$, i.e., multiplying $\bbK$ by some positive constant. This condition also implies
		\begin{equation*}
			\begin{split}
				\|f\|_{\cL^2} \leq & \|f\|_\infty := \sup_{t \in [0,1]}|f(t)| = \sup_{s \in [0,1]}\left|\langle f(\cdot),\bbK(\cdot,s)\rangle_\cH\right| \leq \|f\|_{\cH}\sup_{s \in [0,1]}\|\bbK(\cdot,s)\|_\cH \\
				= & \|f\|_\cH \sup_{s \in [0,1]}\langle \bbK(\cdot, s), \bbK(\cdot,s)\rangle_{\cH}^{1/2}
				= \|f\|_\cH \sup_{s \in [0,1]}\sqrt{\bbK(s,s)} \leq \|f\|_\cH.
			\end{split}
		\end{equation*}
		Assumption \ref{asmp:K-regularity} (b) is introduced for convenience of the later theoretical analysis. This condition appears in the literature on high-dimensional functional linear regression/auto-regression \citep{wang2020functional2,wang2020functional} and also holds for some prevalent functional spaces of interest. For example, \cite{wang2020functional2,wang2020functional} showed that the Sobolev space
		\begin{equation}\label{eq:W^alpha,2}
			W^{\alpha,2} = \left\{f: f^{(r)}, \text{the $r$-th derivative of $f$, is absolutely continuous,  $r=0,\ldots,\alpha$; } f^{(\alpha)} \in \cL^2([0,1])\right\}
		\end{equation}
		satisfies \eqref{ineq:H-cauchy} with some constant $C_\cH$ that only depends on $\alpha$ and showed in particular that $W^{1,2}$ satisfies \eqref{ineq:H-cauchy} with $C_\cH = \sqrt{5}$.
		
		Define the \emph{effective dimension} of $\mathcal H$ as $p_{\mathcal H}:=\sum_{k=1}^\infty \mu_k^2/\mu_1^2.$ Intuitively speaking, $p_\cH$ captures the intrinsic dimension of the Hilbert space $\cH$ in the sense that the Gaussian processes over $\cH$ has the similar incoherence as the Gaussian random vectors in a $p_\cH$-dimensional Euclidean space (see the forthcoming Proposition \ref{prop:sampling-incoherence} for details). As a concrete example, suppose $\mu_k=1$ for $k\leq d$ and $\mu_k=0$ for $k>d$. Then, $p_{\cH}=d$ and $\cH$ is equivalent to a $d$-dimensional Euclidean space. 
		
		\subsection{Functional Tensor Decomposition}\label{sec:model}
		
		Let $\bcY \in \bbR^{p_1 \times p_2 \times [0,1]}$ be the functional tensor of interest. Assume $\bcY$ is approximately CP rank-$r$:
		\begin{equation}\label{eq:functional-CP-model}
			\bcY = \bcX + \bcZ,\qquad \bcX = \sum_{l=1}^r \lambda_l a_l \circ b_l \circ \xi_l \in \mathbb R^{p_1\times p_2\times [0,1]}.
		\end{equation}
		Here, each tuple $(\lambda_l, a_l,b_l,\xi_l)$ corresponds to a {\it singular component} of $\bcX$: $\lambda_l$ is the singular value, $a_l,b_l$ are the singular vectors, and $\xi_l$ is the singular function. We also assume $\|a_l\|_2 = \|b_l\|_2 = \|\xi_l\|_{\cL_2} = 1$ and $\lambda_l>0$ for scaling identifiability. We denote $\lambda_{max} = \max_{l \in [r]} \lambda_l$ and $\lambda_{min} = \min_{l \in [r]}\lambda_l$. The functional tensor $\bcZ \in \bbR^{p_1 \times p_2 \times [0,1]}$ is included to model the unexplained remainder term. We also assume $\xi_l$ satisfies the regularity condition, such that $\|\xi_l\|_\cH \leq C_\xi$ for some given RKHS $\mathcal H$ and some universal constant $C_\xi$. 
		
		Motivated by the previously discussed applications, we assume the functional tensor is measured on a discrete grid of time points $\{s_{k}\}_{k=1}^n \overset{i.i.d.}{\sim} \text{Unif}(0,1)$ and the following tabular tensor $\widetilde\bcY \in \bbR^{p_1 \times p_2 \times n}$ is observed:
		\begin{equation}\label{eq:CP-model-observation}
			\widetilde\bcY_{ijk} := \bcY_{ijs_k} + \varepsilon_{ijk},\quad i\in [p_1], j\in [p_2],k\in [n].
		\end{equation}
		Here, $\varepsilon_{ijk}$'s are independent Gaussian random noises with mean zero and variance $\tau^2$. Our goal is to perform functional tensor singular value decomposition on $\bcY$, that is to estimate $\{(a_l,b_l,\xi_l)\}_{l=1}^r$ based on $\widetilde\bcY$ and $\{s_k\}_{k=1}^n$ from the model \eqref{eq:functional-CP-model} and \eqref{eq:CP-model-observation}. 
		
		\subsection{Model Identifiability}\label{sec:identifiability}

		We first discuss the identifiability conditions for the FTSVD model. We say the parameter tuple $\{\lambda_l, a_l, b_l, \xi_l\}_{l=1}^r$ is {\it identifiable} if and only if for any other parameter tuple $\{\tilde \lambda_l, \tilde a_l, \tilde b_l, \tilde \xi_l\}_{l=1}^r$ satisfying $\sum_{l=1}^r \lambda_l a_l \circ b_l \circ \xi_l = \sum_{l=1}^r \tilde\lambda_l \tilde a_l \circ \tilde b_l \circ \tilde \xi_l$, there exists come permutation $\pi$ on $[r]$ such that
		\begin{equation*}
		\lambda_l = \tilde \lambda_{\pi(l)},\qquad
		    a_l \circ b_l \circ \xi_l = \tilde a_{\pi(l)} \circ \tilde b_{\pi(l)} \circ \tilde c_{\pi(l)}, \quad \forall l \in [r].
		\end{equation*}
		In other words, all rank-1 components can be uniquely determined up to permutation and all factors are identifiable up to sign-flipping and permutation. 
		The following Proposition \ref{prop:identify} shows almost all parameter tuples are identifiable for moderate tensor rank $r$. 
		Proposition \ref{prop:identify} is in parallel with the identifiability condition of tabular tensor CP decomposition~\citep{kruskal1976more}.
		\begin{Proposition}\label{prop:identify}
		    The set of unidentifiable parameter tuples has measure zero with respect to the Lebesgue measure on the parameter space $\Theta = \left\{\{\lambda_l, a_l, b_l, c_l \}_{l=1}^r: \lambda_l > 0, \|a_l\|_2 = \|b_l\|_2 = \|\xi_l\|_{\mathcal{L}_2}, \xi_l \in \cH\right\}$ when either of the following conditions is met:
		    \begin{itemize}
		        \item $\cH$ is a functional space with finite dimension $p_3$, and $2r < p_1+p_2+p_3 - 2$;
		        \item $\cH$ is an infinitely-dimensional space, all the $\xi_l$s are continuous, and $r < p_1 + p_2 - 2$.
		    \end{itemize}
		\end{Proposition}
		
		\section{Methods for Functional Tensor SVD}\label{sec:method}
		
		Next, we discuss the estimation methods for FTSVD. 
		
		\vskip.3cm
		
		\noindent{\bf Power iterations.} As a starting point, we approach the problem from an optimization perspective:
		\begin{equation}\label{eq:ls-est}
			\begin{split}
				\min_{\lambda_l,a_l,b_l,\xi_l}& \quad  \sum_{i,j,k} \left(\widetilde\bcY_{ijk} - \sum_{l=1}^r \lambda_l\cdot (a_l)_i \cdot (b_l)_j \cdot \xi_l(s_k)\right)^2 \\
				\text{subject to } & \quad  \|a_l\|_2 = \|b_l\|_2 = \|\xi_l\|_{\cL^2} = 1, \|\xi_l\|_\cH \leq C_\xi.
			\end{split}
		\end{equation}
		However, the computation of \eqref{eq:ls-est} is highly nontrivial and the exact solution is computationally intractable in general since (a) the problem is highly non-convex or even NP-hard due to the tensor product and the multilinear structure; (b) \eqref{eq:ls-est} is essentially an infinite-dimensional optimization problem due to the functional argument $\xi_l$. Therefore, we introduce the following {\it RKHS-based constrained power iteration method} to overcome these two difficulties. As a popular and powerful method for singular value decomposition, power iteration has been successfully applied to tabular tensor decomposition in the past two decades~\citep{anandkumar2014tensor,anandkumar2014guaranteed}. 
		
		We first discuss the power iteration for the one-component case (i.e., $r=1$), while the multiple-component scenario is postponed later. Given estimates at Step $t$: ($a^{(t)}$, $b^{(t)}$, $\xi^{(t)}$), power iteration updates $a^{(t)}$, $b^{(t)}$, and $\xi^{(t)}$ alternatively. Specifically, recall $\xi_n^{(t)} = (\xi^{(t)}(s_1),\ldots,\xi^{(t)}(s_n))$ is the discretization of the function $\xi^{(t)}$. We update the singular vectors $a^{(t)}, b^{(t)}$ by {\it projection-normalization}:
		\begin{equation}\label{ineq:alg-update-a,b}
			\begin{split}
				\tilde a^{(t+1)} & = \widetilde\bcY \times_2 b^{(t)} \times_3 \xi_n^{(t)},\qquad a^{(t+1)} = \tilde a^{(t+1)} / \|\tilde a^{(t+1)}\|_2, \\
				\tilde b^{(t+1)} & = \widetilde\bcY \times_1 a^{(t)} \times_3 \xi_n^{(t)},\qquad b^{(t+1)} = \tilde b^{(t+1)} / \|\tilde b^{(t+1)}\|_2.
			\end{split}
		\end{equation}
		Since only a discrete subset of observations from the functional tensor are accessible, we consider the following optimization to update the singular function estimator:
		\begin{equation}\label{ineq:alg-update-xi}
			\begin{split}
				& \tilde \xi^{(t+1)} = \argmin_{\|\xi\|_\cH \leq (C_\xi\lambda_{\max})}~ \sum_{i,j,k} \left(\widetilde\bcY_{ijk} - (a^{(t)})_i \cdot (b^{(t)})_j \cdot \xi(s_k)\right)^2, \xi^{(t+1)} = \tilde \xi^{(t+1)} / \|\tilde \xi^{(t+1)}\|_{\cL^2}.
			\end{split}
		\end{equation}
		Note that the first part of \eqref{ineq:alg-update-xi} is essentially a weighted mean functional estimation problem, or a special case of functional linear regression. Since $\tilde \xi^{(t+1)}$ is essentially a minimizer of a regularized empirical risk functional defined over the RKHS $\cH$, it admits a finite-dimensional closed-form solution by the classic Representer Theorem~\citep{kimeldorf1971some} to be discussed in Section \ref{sec:numeric}.
		
		Let $(\hat a, \hat b, \hat \xi)$ be the final estimates after sufficient iterations. The singular value is estimated via $\hat \lambda = \widetilde\bcY \times_1 \hat a \times_2 \hat b \times_3 \hat \xi_n$. The procedure of one-component power iteration is summarized to Algorithm \ref{alg:power-iter}.
		\begin{algorithm}
			\caption{Regularized Power Iteration}
			\label{alg:power-iter}
			\begin{algorithmic}
				\REQUIRE Tensor $\widetilde\bcY \in \bbR^{p_1\times p_2 \times n}$, initialization $(a^{(0)}, b^{(0)})$, iteration number $T$
				\STATE{Calculate $\xi_0$ via:} 
				$$\tilde \xi^{(0)} = \argmin_{\|\xi\|_\cH \leq (C_\xi\lambda_{\max})} \sum_{i,j,k} \left(\widetilde\bcY_{ijk} - (a^{(0)})_i \cdot (b^{(0)})_j \cdot \xi(s_k)\right)^2,\qquad \xi^{(0)} = \tilde \xi^{(0)} / \|\tilde \xi^{(0)}\|_{\cL^2}.$$
				\FORALL{$t = 0,T-1$}
				\STATE{Calculate $ a^{(t+1)}, b^{(t+1)}$ via \eqref{ineq:alg-update-a,b} and calculate $\xi^{(t+1)}$ via \eqref{ineq:alg-update-xi}.}
				\ENDFOR
				\STATE{Calculate $\lambda^{(T)} = \widetilde\bcY \times_1 a^{(T)} \times_2  b^{(T)} \times_3 \xi^{(T)}_n$.}
				\RETURN {$(\hat \lambda, \hat a, \hat b, \hat \xi) = (\hat\lambda, a^{(T)}, b^{(T)}, \xi^{(T)})$}
			\end{algorithmic}
		\end{algorithm}
		
		\vskip.3cm
		
		\noindent{\bf Initialization and Overall Algorithm.} Next, we discuss initialization scheme for $(a^{(0)},b^{(0)})$, which is required for the implementation of Algorithm \ref{alg:power-iter}. If $\bcX$ is a rank-1 tensor, the matricizations of the signal tensor $\bcX$ admit the following rank-1 decompositions: 
		$$\cM_1(\bcX) = \lambda_1 a_1 \left(b_1\otimes (\xi_1)_n\right)^\top,\qquad \cM_2(\bcX) = \lambda_1 b_1 \left((\xi_1)_n\otimes a_1\right)^\top,$$
		i.e., $a_1$ and $b_1$ are the the left singular vectors of $\cM_1(\bcX)$ and $\cM_2(\bcX)$ respectively. Recall ``$\otimes$" is the Kronecker product and $b_1\otimes (\xi_1)_n$ yields a $(p_2n)$-dimensional vector. Given $\widetilde\bcY$ being a noisy substitute of $\bcX$, it is natural to initialize $a_1$ and $b_1$ by the first singular vectors of $\cM_1(\widetilde\bcY)$ and $\cM_2(\widetilde\bcY)$, respectively: 
		\begin{equation}\label{alg:spectral-init}
			a^{(0)} = \SVD_1(\cM_1(\widetilde\bcY)), \quad b^{(0)} = \SVD_1(\cM_2(\widetilde\bcY)).
		\end{equation}
		Our estimation procedure for the rank one scenario is then completed by combining Algorithm \ref{alg:power-iter} with Eq. \eqref{alg:spectral-init}.
		
		When the rank $r>1$, there are multiple singular components to be estimated and the power iteration will no longer be the exact alternating minimization scheme. Nevertheless, when all the singular components $\{(a_l,b_l,\xi_l)\}_{l=1}^r$ satisfy the incoherent condition and $(a^{(0)}, b^{(0)})$ satisfies some initialization conditions, Algorithm \ref{alg:power-iter} still yields estimators with guaranteed local convergence. Therefore, we propose to perform spectral initialization scheme sequentially: we first let $a^{(0)}$ and $b^{(0)}$ be initialized via \eqref{alg:spectral-init} and let $(\hat \lambda_1, \hat a_1, \hat b_1, \hat\xi_1)$ be the first estimated singular component by power iteration (Algorithm \ref{alg:power-iter}). Then we subtract $\hat \lambda_1 \circ \hat a_1 \circ \hat b_1 \circ (\hat\xi_1)_n$ from the original observation $\widetilde\bcY$. Note that this is essentially equivalent to projecting $\widetilde\bcY$ onto the low-dimensional tensor space spanned by $\hat a_1 \circ \hat b_1 \circ (\hat\xi_1)_n$ and replacing the data with the projection residue. We perform this procedure for $r$ times to obtain $r$ estimated singular components. The pseudocode of this procedure is summarized to Algorithm \ref{alg:spectral-r>1}.
		\begin{algorithm}
			\caption{Functional Tensor SVD with Sequential Spectral Initialization}
			\label{alg:spectral-r>1}
			\begin{algorithmic}
				\REQUIRE Tensor $\widetilde\bcY \in \bbR^{p_1\times p_2 \times n}$, rank $r$
				\FOR{$l=1,\ldots,r$}
				\STATE{Calculate $a^{(0)} = \SVD_1\left(\cM_1(\widetilde\bcY)\right)$, $b^{(0)} = \SVD_1\left(\cM_2(\widetilde\bcY)\right).$}
				\STATE{Calculate $(\hat \lambda_l,\hat a_l,\hat b_l,\hat \xi_l)$ by applying Algorithm \ref{alg:power-iter} on $(a^{(0)}, b^{(0)})$.}
				\STATE{Update $\widetilde\bcY = \widetilde\bcY - \hat \lambda_l \hat a_l \circ \hat b_l \circ (\hat{\xi}_l)_n$.}
				\ENDFOR
				\RETURN {$\{(\hat \lambda_l, \hat a_l, \hat b_l, \hat \xi_l)\}_{l=1}^r$}
			\end{algorithmic}
		\end{algorithm}

		\section{Statistical Theory}\label{sec:main-results}
		
		In this section, we present the theoretical results for the proposed estimation algorithms. We particularly aim to study the statistical error bound of $\hat a,\hat b, \hat\xi$. Since the singular vectors and functions are identifiable up to sign flips, we focus on the following \emph{sine values} of the pairs of vectors/functions:
		\begin{equation}\label{eq:def-dist}
			\begin{split}
				\dist(u, v) &= \sqrt{1 - \left(u^\top v/(\|u\|_2\|v\|_2)\right)^2}, \qquad \forall u,v \in \bbR^p;\\
				\dist(f, g) &= \sqrt{1- \left( \int_0^1 f(t) g(t) dt/(\|f\|_{\mathcal{L}^2} \|g\|_{\mathcal{L}^2}) \right)^2}, \qquad \forall f,g \in \cL^2([0,1]).
			\end{split}
		\end{equation}
		We introduce the following quantity associated with the RKHS $\cH$:
		\begin{equation}\label{eq:zeta_n}
			\zeta_n := \inf\left\{\zeta \geq \sqrt{\log n/n}: Q_n(\delta) \leq \zeta \delta + \zeta^2,\quad \forall \delta \in (0,1]\right\},
		\end{equation}
		\begin{equation}\label{eq:Q_n}
			\text{where} \quad  Q_n(\delta) := \left(\sum_{k=1}^\infty \min\{\delta^2, \mu_k\}/n\right)^{1/2},
		\end{equation}
		and $\mu_1\geq \mu_2 \geq \ldots \geq 0$ are the eigenvalues of the reproducing kernel $\mathbb K$ of $\cH$. Essentially, $\zeta_n$ quantifies the information loss from only observing measurements on a discrete grid, rather than the whole function, of $f$ in the RKHS $\mathcal H$. This quantity and its variants are commonly used in the literature on functional data analysis~\citep{koltchinskii2010sparsity,raskutti2012minimax,wang2020functional}. As the grid density $n$ increases, $\{f(s_k)\}_{k=1}^n$ reveals more information on $f$ and the value of $\zeta_n$ decreases. For some specific RKHS $\cH$, the explicit rate of $\zeta_n$ can be obtained. As an example, When $\mathcal{H}$ is the Hilbert-Sobolev space $W^{\alpha,2}([0,1])$ defined in \eqref{eq:W^alpha,2}, we have $\zeta_n \lesssim n^{-\alpha/(2\alpha+1)}$ (see Proposition \ref{prop:zeta-rate}). More technical results on $Q_n(\delta)$ and $\zeta_n$ will be discussed in Appendix \ref{sec:prop-zeta}.
		
		We also introduce the following quantity to measure the scale of the unexplained remainder term $\bcZ$: 
		\begin{equation}\label{eq:def_E}
			\begin{split}
				& \mathcal E := \sup_{s \in [0,1]} \left\|\bcZ_{\cdot\cdot s}\right\|. 
			\end{split}
		\end{equation}
		Here, $\|\cdot\|$ is the largest singular value of the matrix. $\mathcal E$ can be interpreted as an extension of the \emph{tensor spectral norm} from the tabular tensors, which appears as the prominent term in the estimation error bound in various statistical tensor models~\citep{anandkumar2014guaranteed,sun2017provable,han2021optimal}. In particular, the forthcoming Proposition \ref{prop:Gaussian-process-xi} explicitly characterizes $\cE$ for Gaussian processes $\bcZ_{ij\cdot}$ via the dimensions of the tabular modes $p_1, p_2$.
		
		\subsection{Assumptions}\label{sec:assumptions}
		
		In this section, we present the technical assumptions that are used to establish the theoretical guarantees for the proposed algorithm. We first introduce the following incoherence conditions on the true singular components $(a_l,b_l,\xi_l)$ of the signal tensor $\bcX$.
		\begin{Assumption}[Incoherence]\label{asmp:incoherence}
			Consider model \eqref{eq:functional-CP-model}. 
			Assume
			\begin{equation*}
				\mu:=\max\left\{\max_{i \neq j}\left\{|\langle a_i,a_j \rangle|\right\}, ~\max_{i \neq j}\left\{|\langle b_i,b_j \rangle|\right\},~ \max_{i \neq j}\left\{|\langle \xi_i,\xi_j \rangle_{\cL^2}|\right\}\right\} \leq c/(\kappa^2 r^2),
			\end{equation*}
			where $\kappa:= \lambda_{max} / \lambda_{min}$ is the condition number. 
		\end{Assumption}
		Assumption \ref{asmp:incoherence} suggests the singular components of $\bcX$ are pairwise incoherent. Such the condition is widely assumed in the literature on tabular CP low-rank tensor decomposition. By \cite[Lemma 2]{anandkumar2014guaranteed}, one can show that $\max_{i\neq j}\{|\langle a_i, a_j\rangle|\} \asymp 1/\sqrt{p_1}$ and $\max_{i\neq j}\{|\langle b_i, b_j\rangle|\} \asymp 1/\sqrt{p_2}$ with high probability if $\{a_l\}_{l=1}^r$ and $\{b_l\}_{l=1}^r$ are i.i.d. uniformly sampled from $\mathbb S^{p_1-1}$ and $\mathbb S^{p_2-1}$, respectively. We can also establish the incoherence for singular functions under certain sampling distributions descried as follows.
		\begin{Proposition}\label{prop:sampling-incoherence}
			Suppose $r\geq 2$. Let $f_i \sim GP(0, \bbG(\cdot,\cdot)), i=1,\ldots, r$ be i.i.d. mean-zero Gaussian processes with covariance function $\bbG(s,t) = \sum_{k=1}^\infty \mu_k^2 \phi_k(s)\phi_k(t)$. Assume $\log r < c\left(\sum_{k=1}^\infty \mu_k^4/\mu_1^4\right)$ and normalize $\xi_i = f_i / \|f_i\|_{\mathcal{L}_2}$. Recall the effective dimension $p_\cH:=\sum_{k=1}^\infty \mu_k^2/\mu_1^2$. Then, with probability at least $1-Cr^{-8}$, 
			\begin{equation*}
				\begin{split}
					\xi_i \in \cH,\quad 
					|\langle \xi_i, \xi_j \rangle| \leq C\log r/\sqrt{p_{\mathcal H}}, \quad \forall 1\leq i\neq j\leq r. 
				\end{split}
			\end{equation*}
		\end{Proposition}
		Our next assumption is on the grid density and signal-to-noise ratio (SNR).
		\begin{Assumption}[Grid density and SNR]\label{asmp:sample-size}
			The grid density $n$ and the least singular value $\lambda_{min}$ satisfy
			\begin{equation}\label{ineq:local-SNR-condition}
				\frac{n}{\log^2 n} \geq C\kappa(p_1+p_2), \quad \zeta_n \leq c/\kappa^2 r, \quad \lambda_{min} \geq C\kappa r\left\{\cE + \tau\left(\zeta_n + \sqrt{(p_1+p_2)/n}\right)\right\}
			\end{equation}
			for some sufficiently small constant $c>0$ and large constant $C>0$.
		\end{Assumption}
		Recall that $\zeta_n$ characterizes the difficulty of estimating a function in $\cH$ under discretization and Gaussian observational errors. In most interesting scenarios (which will be discussed later), these quantities have polynomial decay with respect to $n$, and the first two conditions in \eqref{ineq:local-SNR-condition} is mild that allows $p, r$ to grow with $n$. The third condition in \eqref{ineq:local-SNR-condition} is similar to the widely-assumed SNR condition in the literature on tabular tensor decomposition~\citep{anandkumar2014guaranteed,han2021optimal}, where the ratios of noise (e.g., $\cE,\tau$) and the least singular value ($\lambda_{min}$) are required to be sufficiently large. This guarantees that each singular component has strong-enough signal and can be estimated consistently.
		
		\begin{Assumption}[Initialization]\label{asmp:initialization}
			There exists some $l \in [r]$ and constant $c_0>0$ such that the initialization error satisfies $\max\left\{\dist(a^{(0)}, a_l),\dist(b^{(0)}, b_l)\right\} \leq c_0/(\kappa r).$
		\end{Assumption}
		Assumption \ref{asmp:initialization} requires that the initial tabular mode estimations are reasonably close to one of the true singular component. Such the condition is widely used in the theoretical analysis for power iteration~\citep{anandkumar2014guaranteed,cai2019nonconvex}. In particular, when $\cX$ is a well-conditioned tensor with a constant rank (i.e., $r,\kappa = O(1)$), this assumption means the initialization error is only smaller than a constant. Later, we will show the initial estimator yielded by the spectral initialization (Eqn. \eqref{alg:spectral-init}) satisfies this assumption. 
		
		\subsection{Local Convergence}\label{sec:local-convergence}
		
		The following theorem gives an estimation error bound for the proposed RKHS-based constrained power iteration (Algorithm \ref{alg:power-iter}). 
		\begin{Theorem}\label{thm:local-convergence-RKHS}
			Suppose Assumptions \ref{asmp:K-regularity} - \ref{asmp:initialization} hold. Let $a^{(t)}, b^{(t)}, \xi^{(t)}$ be estimated singular vectors and function at step $t$ of Algorithm \ref{alg:power-iter}. Recall $\cE$ is defined in \eqref{eq:def_E}, $\zeta_n$ is defined in \eqref{eq:zeta_n}, $\mu$ is defined in Assumption \eqref{asmp:incoherence}. Then with probability at least $1-2n^{-9}-C\log n \cdot e^{-c(p_1\wedge p_2)}$, we have 
			\begin{equation}\label{ineq:error-bound-thm1}
				\begin{split}
					& \max\{\dist(a^{(t)}, a_l), \dist(b^{(t)}, b_l),\dist(\xi^{(t)}, \xi_l)\}
					\\
					& \qquad\leq  2^{-t} + C\left\{\frac{\cE}{\lambda_{min}} + \frac{\tau\left(\zeta_n + \sqrt{(p_1+p_2)/n}\right)}{\lambda_{min}} + \kappa(\zeta_n + (r-1)\mu)\right\},\quad 	\forall t \geq 0.
				\end{split}
			\end{equation}
			Specifically if $r=1$, the estimation error of $a^{(t)}, b^{(t)}$ enjoys the following better rate: 	\begin{equation*}
				\max\{\dist(a^{(t)}, a_l), \dist(b^{(t)}, b_l)\} \leq C\left\{\frac{\cE }{\lambda_{min}} + \frac{\tau(\zeta_n + \sqrt{(p_1+p_2)/n})}{\lambda_{min}}\right\}, \quad \forall t\geq 1. 
			\end{equation*}
		\end{Theorem}
		\begin{Remark}[Interpretation of Estimation Error Bound]
			Theorem \ref{thm:local-convergence-RKHS} shows if the number of iterations $T>\Omega(\log(1/\zeta_n))$, the estimation error of $\xi_s$ can be upper bounded by the sum of the following four terms:  $\cE/\lambda_{min}$, $\tau\left(\zeta_n + \sqrt{(p_1+p_2)/n}\right)/\lambda_{min}$, $\kappa\zeta_n$ and $\kappa(r-1)\mu.$ 
			Here, $\cE/\lambda_{min}$ is induced by the unexplained remainder term $\bcZ$, which can be interpreted as the ``spectral-norm-over-least-singular-value'' term that widely appears in the literature on tabular tensor decomposition~\citep{anandkumar2014guaranteed,han2021optimal}. The second term $\tau\left(\zeta_n + \sqrt{(p_1+p_2)/n}\right)/\lambda_{min}$ accounts for the observational noise $\varepsilon_{ijk}$. On the other hand, the classic analysis of tabular tensor decomposition \citep[Theorem 1]{zhang2018tensor}, \citep[Proposition 1]{han2020lloyd} only yields a much weaker rate of the tabular estimation error, $\tau/\lambda_{min}\cdot \left((p_1p_2/n)^{1/4} + \sqrt{(p_1+p_2)/n}\right)$, which can be much larger than the bound in Theorem \ref{thm:local-convergence-RKHS} in the high-dimensional region where $p_1p_2$ is comparable to $n$. This phenomenon also appears in our simulation studies in Section \ref{sec:simulation}. The third term corresponds to the discretization error on estimating the function $f$ from a discrete grid, which is unavoidable even in the noiseless orthogonal-decomposable case (i.e., $\mathcal E = \mathcal E_\cH = \tau = \mu = 0$). The last term $(r-1)\mu$ corresponds to the coherence of the signal tensor $\bcX$: when different singular components are not exactly orthogonal, each of them is not a stationery point even in the noiseless cases. Nevertheless, this term often has negligible rate compared to the other three terms in many scenarios (for example, in the scenario of Proposition \ref{prop:sampling-incoherence}, $\mu \lesssim 1/\sqrt{\min\{p_1,p_2,p_\cH\}}$; see the discussions therein).
			
			It is also noteworthy that when the tensor $\bcX$ is rank-$1$, the estimation error for singular functions does not contain the discretization term $\zeta_n$ while this is not true for $r>1$. The reason is that when there is only one singular component, the ``inaccurate'' estimation $\hat\xi$ will not introduce additional bias from other singular components when projection is performed in power iteration. On the other hand, the estimation error of $\hat\xi$ will be accumulated to the tabular modes when there are more than one singular components in the signal tensor $\bcX$. 
		\end{Remark}
		
		Note that the result of Theorem \ref{thm:local-convergence-RKHS} is conditioning on fixed values of the functional remainder term $\mathcal E$. When $\bcZ$ is random, we can develop finer upper bounds on $\mathcal E$ as well as the estimation error. For example, if $\bcZ$ is a collection of i.i.d. Gaussian processes, we have the following result.
		\begin{Assumption}\label{asmp:Z-Gaussian-H}
			Suppose $\bcZ_{ij\cdot}$ are i.i.d. mean-zero Gaussian processes with almost sure continuous path and
			\begin{equation}\label{eq:bcZ-assumption}
				\bbE \sup_{s} |\bcZ_{ijs}| \leq m,\qquad \sup_{s} {\rm Var}(\bcZ_{ijs}) \leq \sigma^2.
			\end{equation}
		\end{Assumption}
		Assumption \ref{asmp:Z-Gaussian-H} essentially means each functional $\bcZ_{ij\cdot}$ is perturbed within a finite variation/scale that does not grow with tensor dimension $p_1,p_2$. Assumption \ref{asmp:Z-Gaussian-H} also implies the following upper bound on $\mathcal E$.
		\begin{Proposition}\label{prop:Gaussian-process-xi}
			Suppose Assumption \ref{asmp:Z-Gaussian-H} holds. Then with probability at least $1-e^{-c(p_1+p_2)}$,
			\begin{equation*}
				\begin{split}
					\mathcal E \leq  m + \sigma\sqrt{p_1 + p_2}.
				\end{split}
			\end{equation*}
		\end{Proposition}
		
		To illustrate our theoretical results, we specifically consider two RKHSs of major interests next. First, if the RKHS kernel $\bbK$ has $d$ non-zero eigenvalues, it follows that $\zeta_n \lesssim \sqrt{d/n}$ (Proposition \ref{prop:zeta-rate}) and we can obtain the following estimation error upper bound. 
		\begin{Corollary}[Local Convergence: Finite-dimensional RKHS]\label{coro:finite-eig}
			Suppose the reproducing kernel $\bbK$ of $\cH$ has finitely many, say $d$, non-zero eigenvalues. 	Suppose $\log n \leq d\leq n$, Assumptions \ref{asmp:K-regularity} - \ref{asmp:Z-Gaussian-H} hold. Then, with probability at least $1-Cn^{-9}-C\log n \cdot e^{-c(p_1\wedge p_2)}$, 
			\begin{equation*}
				\begin{split}
					&\max\{\dist(\hat a, a_l), \dist(\hat b, b_l), \dist(\hat \xi, \xi_l)\} \\
					& \qquad \lesssim \sqrt{\frac{d}{n}} + \frac{m + \sigma\sqrt{p_1+p_2}}{\lambda_{min}} + \sqrt{\frac{p_1+p_2+d}{n}}\cdot\frac{\tau}{\lambda_{min}} + (r-1)\mu.
				\end{split}
			\end{equation*}
		\end{Corollary}

		Another interesting example is the Sobolev-Hilbert space $W^{\alpha, 2}$ defined in \eqref{eq:W^alpha,2}. If $\cH = W^{\alpha,2}$ with constant $\alpha > 1/2$, then $\cH$ is an RKHS with eigenvalues $\mu_k \asymp k^{-2\alpha}$ (see, e.g., \cite{micchelli1981design}) and $\zeta_n \lesssim n^{-\alpha/(2\alpha+1)}$ (Proposition \ref{prop:zeta-rate}). The following estimation error bound holds accordingly:
		\begin{Corollary}[Local Convergence: Sobolev Space]\label{coro:sobolev}
			Suppose $\cH = W^{\alpha,2}$ with $\alpha > 1/2$ and Assumptions \ref{asmp:K-regularity} -  \ref{asmp:Z-Gaussian-H} hold. 
			Then, with probability at least $1-Cn^{-9}-C\log n \cdot e^{-c(p_1\wedge p_2)}$, 
			\begin{equation*}
				\begin{split}
					&\max\{\dist(\hat a, a_l), \dist(\hat b, b_l), \dist(\hat \xi, \xi_l)\} \\
					& \qquad \lesssim n^{-\frac{\alpha}{2\alpha+1}} + \frac{m + \sigma\sqrt{p_1+p_2}}{\lambda_{min}} + \left(n^{-\frac{\alpha}{2\alpha+1}} + \sqrt{\frac{p_1+p_2}{n}}\right)\frac{\tau}{\lambda_{min}}  + (r-1)\mu.
				\end{split}
			\end{equation*}
		\end{Corollary}
		\begin{Remark}[Consistency Conditions for FTSVD in Sobolev Space and Its Phase Transition]
			When $m \ll \sigma\sqrt{p_1+p_2}$, $\bcX$ is orthogonally decomposable (i.e., $\mu=0$) and let $n\rightarrow \infty$, Corollary \ref{coro:sobolev} implies that $\hat{a}, \hat{b}, \hat{\xi}$ are consistent estimators if
			\begin{equation*}
				\lambda_{min}/\sigma \gg \sqrt{p_1\vee p_2},\qquad \lambda_{min}/\tau \gg n^{-\frac{\alpha}{2\alpha+1}} \vee \sqrt{(p_1 \vee p_2)/n}.
			\end{equation*}
			Here, $\lambda_{min}/\sigma \gg \sqrt{p_1 \vee p_2}$ is required for consistent estimation of the singular vectors $a_l,b_l$~(also see \cite{cai2019nonconvex}); while $\lambda_{min}/\tau \gg n^{-\alpha/(2\alpha+1)} \vee \sqrt{(p_1 \vee p_2)/n}$ is required for consistent estimation/approximation of the singular functions $\xi_l$ and it yields an interesting phase transition: when $p_1 \vee p_2 \ll n^{1/(2\alpha +1)}$, the bound is dominated by the non-parametric rate $n^{-\alpha/(2\alpha+1)}$; when $p_1 \vee p_2 \gg n^{1/(2\alpha +1)}$, the bound is dominated by the parametric rate $\sqrt{(p_1\vee p_2)/n}$. We note that the phase transition between parametric and non-parametric rates commonly appear in the study of high-dimensional functional regression with sparsity, e.g., \cite{wang2020functional}. To the best of our knowledge, we are the first to establish such the phenomenon in the low-rank-based functional data analyses.
		\end{Remark}

		\subsection{Initialization}\label{sec:theory-init}
		
		Next, we provide the theoretical guarantee for the initialization procedure with rank $r=1$ described in \eqref{alg:spectral-init}. The theoretical analysis for initialization of the rank $r>1$ case is technical challenging, which we leave as future research. To simplify the notation, we omit the subscripts and denote $\lambda = \lambda_1 = \lambda_{min}$, $a = a_1$, etc.
		\begin{Theorem}\label{thm:initialization}
			Suppose $r=1$, Assumption \ref{asmp:K-regularity} and \ref{asmp:Z-Gaussian-H} are satisfied and $\zeta_n\leq c_0$. Recall $\sigma$ and $\tau$ are the noise level of $\bcZ$ (defined in \eqref{eq:bcZ-assumption}) and $\varepsilon_{ijk}$ (defined in \eqref{eq:CP-model-observation}), respectively. Suppose $n \geq C_0$ for some sufficiently large constant $C_0$. Then, for any $\delta \in (0,1)$, there exits some constants $C_\delta$ and $C_\delta'$ which depend on $\delta$ such that as long as
			\begin{equation}\label{ineq:init-SNR}
				\begin{split}
					\frac{\lambda}{\sigma} & \geq C_\delta\left(\sqrt{p_1 \vee p_2} \right), \quad 
					\frac{\lambda}{\tau} \geq C_\delta'\left(\sqrt{\frac{p_1 \vee p_2}{n}} + \left(\frac{p_1p_2}{n} \right)^{1/4}\right),
				\end{split}
			\end{equation}
			we have with probability at least $1-n^{-9}-e^{-\sqrt{p_1 \wedge p_2}}$ that
			\begin{equation*}
				\begin{split}
					\max\left\{\dist(a^{(0)}, a), \dist(b^{(0}, b)\right\} \leq \delta.
				\end{split}
			\end{equation*}
			Moreover, if Assumptions \ref{asmp:incoherence}, \ref{asmp:sample-size} and \eqref{ineq:init-SNR} hold, we have the same conclusion as Theorem \ref{thm:local-convergence-RKHS}.
		\end{Theorem}
		
		Theorem \ref{thm:initialization} implies that the initialization condition~(Assumption \ref{asmp:initialization}) for local convergence (Theorem \ref{thm:local-convergence-RKHS}) holds as long as the signal-to-noise ratio (SNR) condition \eqref{ineq:init-SNR} holds.

		\section{Numeric Experiments}\label{sec:numeric}
		
		In this section, we investigate the performance of the proposed Functional Tensor SVD estimators, which we refer as FTSVD for short later, on both synthetic and real datasets. We use the rescaled Bernoulli polynomial as the reproducing kernel throughout the section:
		\begin{equation*}
			\bbK(x,y) = 1 + k_1(x) k_1(y) + k_2(x) k_2(y) - k_4(|x-y|),
		\end{equation*}
		where $k_1(x) = x-.5$, $k_2(x) = (k_1^2(x)-1/12)/2$, and $k_4(x) = (k_1^4(x) - k_1^2(x)/2 + 7/240)/24$ for any $x \in [0,1]$. Note that $\bbK$ is the reproducing kernel for the Hilbert space $W^{2,2}$~\citep[Chapter 2.3.3]{gu2013smoothing}. 
		
		\subsection{Computation of \texorpdfstring{\eqref{ineq:alg-update-xi}}{Lg}}
		We first provide the implementation details for the constraint optimization \eqref{ineq:alg-update-xi} -- a crucial step in updating the singular functions in the proposed algorithm. Denote $k(s) = \left(\bbK(s,s_1),\ldots,\bbK(s,s_n)\right)^\top \in \bbR^{n}$ for any $s \in [0,1]$ and $\bK = \left[k(s_1),\cdots,k(s_n)\right] \in \bbR^{n\times n}$ as the discrete kernel matrix. Then, we can obtain a solution to \eqref{ineq:alg-update-xi} based on the convex duality and Representor Theorem~\citep{kimeldorf1971some}.
		\begin{Proposition}\label{prop:optimization-solution}
		    Let $\tilde y^{(t)} := \cM_3\left(\widetilde \bcY\right)\left(a^{(t)} \otimes b^{(t)}\right) \in \bbR^n$ and let
		    \begin{equation*}
		        C_\lambda = \inf \left\{\mu: \mu \geq 0,~ y^{(t)\top} \left(\bK + \mu \bI\right)^{-1}\bK(\bK + \mu \bI)^{-1} y^{(t)} \leq \left(C_\xi \lambda_{max}\right)^2\right\}.
		    \end{equation*}
		    Then, $\tilde \xi^{(t+1)}(s) = k(s)^\top \beta$ for all $s\in [0,1]$ with $\beta := \left(\bK + C_\lambda \bI\right)^{\dagger} \tilde y^{(t)} \in \bbR^n$.
		\end{Proposition}
		Although $C_\lambda$ does not have a closed form solution, it can be efficiently computed via binary search since the function $h(\mu) := y^{(t)\top} \left(\bK + \mu \bI\right)^{-1}\bK(\bK + \mu \bI)^{-1} y^{(t)}$ is monotonically decreasing. The threshold $(C_\xi\lambda_{max})^2$ corresponds to the constraint value in \eqref{ineq:alg-update-xi}. Recall $C_\xi$ is the upper bound of the RKHS-norm of singular functions and we use a fixed value $C_\xi = 200$ under the Bernoulli polynomial kernel for all the upcoming experiments. $\lambda_{max}$ is the largest singular value in the proposed model that is also unknown. Therefore, we apply High-order Orthogonal Iteration (HOOI) method~\citep{de2000best} to estimate $\lambda_{max}$:
		\begin{equation*}
		    (\hat a, \hat b, \hat c) = \text{HOOI}(\widetilde \bcY), \qquad \hat \lambda_{max} = n^{-1/2} \cdot \widetilde\bcY \times_1 \hat a^\top \times_2 \hat b^\top \times_3 \hat c^\top.
		\end{equation*}
		Here $(\hat a, \hat b, \hat c)$ are the estimates of leading singular vectors of the tabular tensor $\widetilde \bcY$. Note that in addition to projecting $\widetilde\bcY$ onto these three directions, we need to divide it by $n^{1/2}$ as $\lambda_{max}$ is singular value for the functional tensor, not the discretized tabular tensor.

	\subsection{Simulation Studies}\label{sec:simulation}
	We simulate the $p_1\times p_2 \times [0,1]$-dimensional rank-$r$ functional tensor datasets as follows. For each $l \in [r]$, we sample $a_l,b_l$ uniformly from the unit spheres $\mathbb{S}^{p_1-1}, \mathbb{S}^{p_2-1}$ and generate $\xi_l$ from orthonormal basis functions $\{u_i(s)\}_{i=1}^{10} \subset \cL^2([0,1])$. Following \cite{yuan2010reproducing,wang2020functional}, we set $u_1(s) = 1$ and $u_i(s) = \sqrt{2}\cos\left((i-1)\pi s\right)$ for $i=2,\ldots,10$. We generate $x_{li} \sim \text{Unif}[-1/i, 1/i]$ independently, $\xi'_l(\cdot) = \sum_{i=1}^{10} x_{li}u_i(\cdot)$, singular functions $\xi_l = \xi'_l / \|\xi'_l\|_{\cL^2}$, and the signal tensor $\bcX = \sum_{l=1}^r \lambda_l a_l \circ b_l \circ \xi_l \in \bbR^{p_1 \times p_2 \times [0,1]}$, where $\lambda_l := \lambda_{min}\cdot (r-l+1)$ and  $\lambda_{min}>0$ is some pre-specified least singular value. We set the remainder functions $\bcZ_{ij\cdot} = 0$ for the upcoming two experiments and their effects are studied in Appendix \ref{sec:additional_simus}. Finally, we generate a discrete grid $\{s_k\}_{k=1}^n \sim \text{Unif}(0,1)$ and obtain the observation $\widetilde\bcY \in \bbR^{p_1\times p_2 \times n}$ such that $\widetilde\bcY_{ijk} = \bcY_{ijs_k} + \varepsilon_{ijk}$, where $\varepsilon_{ijk}\overset{i.i.d.}{\sim} N(0,\tau^2)$. For each simulation setting, the estimation errors are calculated by \eqref{eq:def-dist} and reported over 100-times repeated experiments. 
	
	\begin{figure}[h!]
		\centering
		\includegraphics[width=6in]{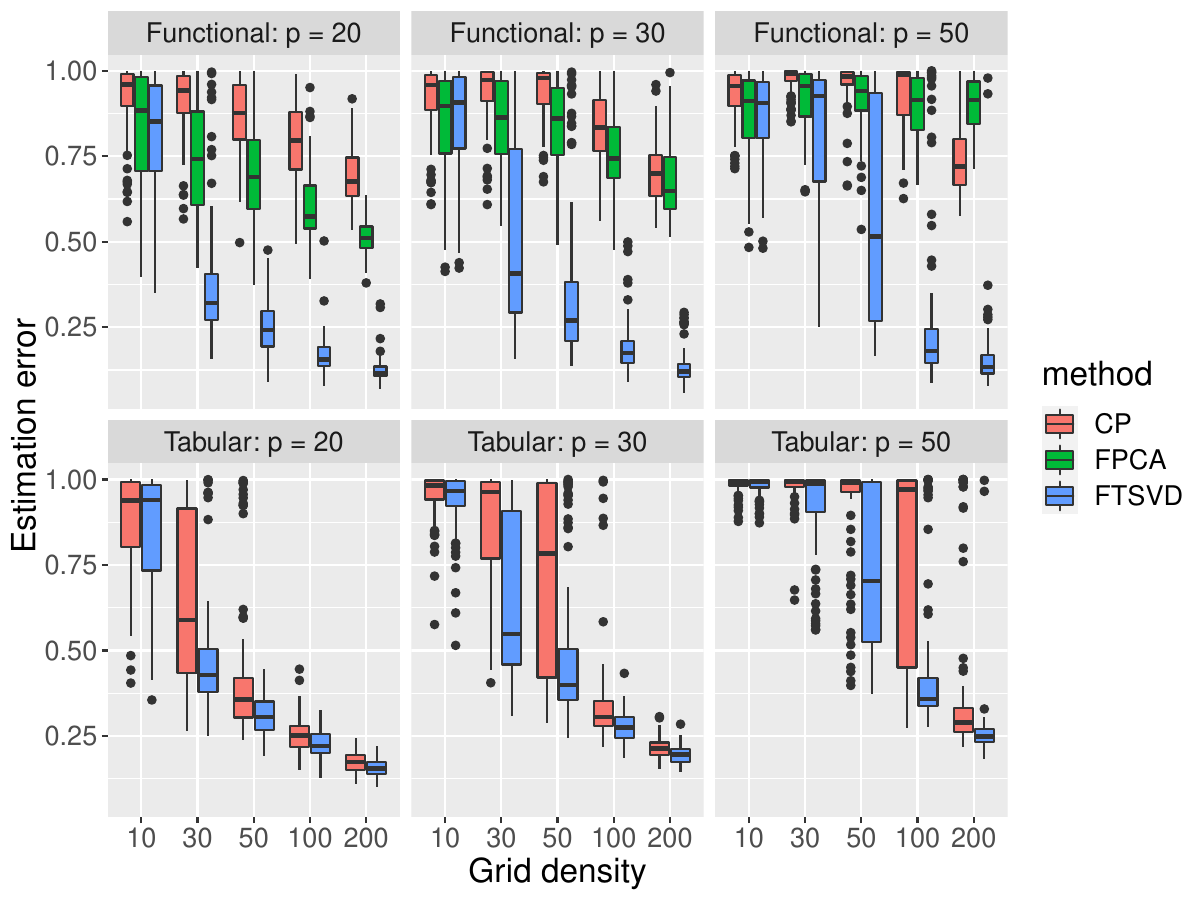}
		\caption{Comparisons among FTSVD, CP and FPCA under different dimension $p=p_1=p_2$ and grid density $n$ for rank-$1$ models. Upper and lower panels plot the estimation errors for singular functions and vectors, respectively. Since FPCA does not yeild a tabular loading estimation, we only report its estimation performance on functional mode.
		}\label{fig:cmp-rank-1}
	\end{figure}
	
	We start by studying the performance of FTSVD on rank-1 models and set $\lambda_{min} = 2$, $\tau = 1$, and $\sigma = 0$. We compare our method with the classic method for rank-1 tabular tensor decomposition (which we refer to as CP later) and the functional prnciple component analysis (FPCA) proposed by \cite{yao2005functional-regression}. Since CP only yields the estimation on the discrete grid for functional mode, we obtain the whole function estimator by interpolation with minimal $\cH$-norm. To apply FPCA, we first unfold the original observations tensor $\widetilde \bcY$ to a matrix $\widetilde \bY = \cM_3(\widetilde \bcY) \in \bbR^{n \times (p_1p_2)}$ and treat each columns of $\widetilde \bcY$ as a functional sample evaluated at an $n$-grid. The comparison between CP and the proposed FTSVD for dimension $p=p_1=p_2 \in \{20,30,50\}$ and various grid density $n$ is presented in Figure~\ref{fig:cmp-rank-1}. We can see the proposed algorithm has more accurate estimations than the classic CP-decomposition and FPCA in all scenarios on both functional and tabular modes. In particular, on the functional mode, our method is significantly better as we fully utilize both the functional smoothness and low-dimensional structures for estimation at each step of power iteration.

	We then explore the high-dimensional settings with imbalanced tabular dimensions ($p_1 \neq p_2$) and/or multiple components ($r>1$). We fix $\lambda_{min} = 8$, $\sigma = \tau = 1$, $n=30$, and consider the following four specific scenarios: \textbf{\rom 1:} $p_1 = 20, p_2 = 500, r = 1$; \textbf{\rom 2:} $p_1 = 20, p_2 = 500, r = 2$; \textbf{\rom 3:} $p_1 = 100, p_2 = 500, r = 1$; \textbf{\rom 4:} $p_1 = 100, p_2 = 100, r = 3$. FPCA and CP, the two baseline methods, are implemented as follows in the cases of $r>1$. For FPCA, we unfold the tensor data to matrix similarly as the previous experiment and apply FPCA to obtain the first $r$ eigenfunctions. For CP, we apply the classic power iteration with random initialization proposed by \cite{anandkumar2014guaranteed}  (the number of initializers is set to 20), and the same interpolation procedure is applied to obtain a functional observation as we did in rank-$1$ situations. The estimation errors of all singular vectors/functions and the corresponding execution time are reported in Table \ref{tab:tab1}. For multiple PCs scenarios, the average estimation errors of all singular vectors/functions (after optimal perturbation) are presented. NAs appear since FPCA does not yield estimates of $a$ or $b$. As we can see, FTSVD achieves better performances on all the four scenarios than CP and FPCA, particularly in singular function estimation. We also observe that our method has significantly better estimation performance for tabular loadings when rank is greater than one. In addition, our method has nearly the same execution time as the classic CP. 
	
	\begin{table}
	\footnotesize
		\centering
		\begin{tabular}{lc|cccc}
			\hline\hline
			Scenarios & Method & $\dist(a, \hat a)$ & $\dist(b, \hat b)$ & $\dist(\xi, \hat \xi)$ & execution time (s) \\ \hline
			\textbf{\rom 1} & FTSVD & 0.113 (0.023) & 0.463 (0.040) & \textbf{0.150} (0.087) & 1.12 \\
			& CP & 0.114 (0.024) & 0.465 (0.040) & 0.516 (0.220) & 1.08 \\
			& FPCA & NA & NA & 0.370 (0.094) & 13.84 \\ \hline
			\textbf{\rom 2} & FTSVD & \textbf{0.112} (0.051) & \textbf{0.376} (0.032)  & \textbf{0.308} (0.110)  & 2.28 \\
			& CP & 0.368 (0.197) & 0.538 (0.125) & 0.694 (0.179) &  37.98 \\
			& FPCA & NA & NA & 0.578 (0.118) & 17.53 \\ \hline
			\textbf{\rom 3} & FTSVD & 0.263 (0.066) & 0.475 (0.062) & \textbf{0.164} (0.091) & 2.61\\
			& CP & 0.285 (0.150) & 0.491 (0.113) & 0.522 (0.238) & 2.35\\
			& FPCA & NA & NA & 0.631 (0.152) & 1.35 \\ \hline
			\textbf{\rom 4} & FTSVD & \textbf{0.172} (0.059) & \textbf{0.173} (0.059) & \textbf{0.379} (0.114) & 2.72\\
			& CP & 0.511 (0.143) & 0.513 (0.140) & 0.735 (0.137) & 13.19\\
			& FPCA & NA & NA & 0.692 (0.076) & 19.77\\ \hline \hline
		\end{tabular}
		\caption{Estimation errors and execution time in high-dimensional settings with imbalanced tabular dimension and multiple ranks. The standard errors are shown in parentheses. }
		\label{tab:tab1}
	\end{table}
		
    We also include additional simulation studies to evaluate the effect of functional perturbation $\bcZ$ (i.e., $\sigma > 0$) and to compare with state-of-the-art methods in multivariate functional data analysis.
    See Appendix \ref{sec:additional_simus} for details.

		\subsection{Real Data Analysis}\label{sec:real-data}
		In this section, we apply the proposed method to the longitudinal microbiome study. An additional example on world-wide crop production analysis is postponed to Appendix~\ref{sec:crop-data} in the supplementary materials.
		
		To investigate the change of fecal microbial composition for new-born infants we apply our methods to the Early Childhood Antibiotics and the Microbiome (ECAM) dataset published by~\cite{bokulich2016antibiotics} (Qiita ID 10249). We consider the 42 infants with multiple fecal microbiome measurements from birth over the first 2 years of life. The infants have fecal microbiome sampled monthly in the first year and bi-monthly in the second year. Among the 42 infants, 24 are vaginally delivered and 18 are Cesarean delivered. A natural question is whether the delivery method affects the composition and development of microbiome in the infants' gut environment.  
		
		We focus on the 50 bacterial genera with non-zero read counts for more than 10\% of all the samples. The data can be organized as an order-3 count tensor $\bar\bcY \in \mathbb N^{42 \times 50 \times 19}$, where the three modes represent different subjects (i.e., infants), bacterial genus and sampling time respectively. To account for the variation in sequencing depth, we transform the count data to the log-composition after $.5$ is added to every count:
		\begin{equation*}
			\widetilde\bcY \in \bbR^{42 \times 50 \times 19},\qquad \widetilde\bcY_{ijk} = \log\left(\frac{\bar\bcY_{ijk}+.5}{\sum_{j'=1}^{50}\left(\bar\bcY_{ij'k}+.5\right)}\right).
		\end{equation*}
		We apply the proposed method with sequential spectral initialization (Algorithm \ref{alg:spectral-r>1}) on the centralized data to estimate the leading singular components for each mode. For the purpose of illustration, we focus on the first three components (i.e., $r=3$). They collectively explain the $31.0\%$ total variations of the data. The singular vectors on the subject mode (i.e., $\{\hat a_s\}_{s=1}^3$) are visualized using bi-plots in Figure \ref{fig:ECAM-subj-PC}, where a well separation of infants with different delivery methods can be seen in the first three components, particularly the second and the third component. We next present the estimated singular functions on the time mode (i.e., $\{\hat\xi_s\}_{s=1}^3$) in the left panel of Figure \ref{fig:ECAM-func-PC}. One can see: 1) the first singular function increases slowly and plateaus after 12-15 months after birth, which matches the recent findings that the gut microbiome of infants is in the developmental phase during 3-15 months of age~\citep{stewart2018temporal}; the second singular function is nearly monotone after the first month and has significant time-variation, which represents a monotone trend of the abundance of certain bacterial; the third singular function increases in the first three months and decreases in the second year, suggesting the difference of bacteria abundance in the first and second year after birth. 
		The two panels on the right of Figure \ref{fig:ECAM-func-PC} demonstrate the mean and error bands of the observed bacteria trajectories for all subjects grouped by delivery method, where the trajectory of each subject is obtained by $\sum_{j=1}^{50} \hat b_{lj}\widetilde\bcY_{ijk}$, i.e., the weighted average of the observed trajectories of all bacteria genera using the feature singular vector of Component 2 or 3 as weights. The trajectories of infants from different delivery methods are well separated in Component 3 and in the latter part of Component 2. Note that all the singular vectors/functions are learned without the information of delivery methods. Our result indicates that the 50 bacteria genera can be reduced to two aggregated bacteria genera using the singular vector on the bacterial mode and still achieve high predictive ability of the delivery methods.

		\begin{figure}[htbp]
			\centering
			\includegraphics[width=.8\textwidth]{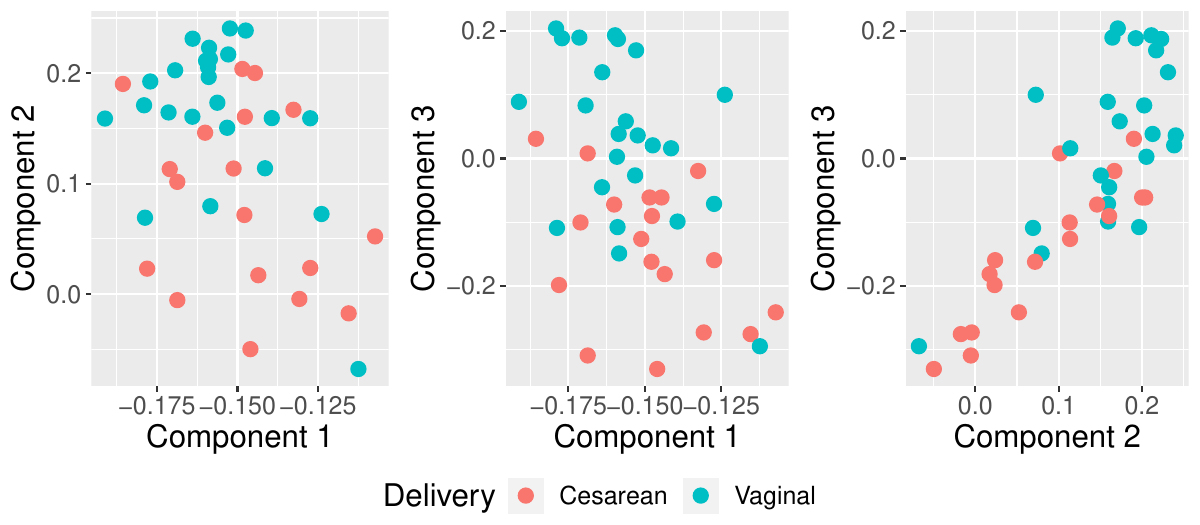}
			\caption{Biplot of three singular vectors on subject mode from ECAM data. Note that each point represents an infant with color indicating the delivery method. }\label{fig:ECAM-subj-PC}
		\end{figure}
		
		\begin{figure}[h!]
			\centering
			\subfigure{\includegraphics[height = 3.5in]{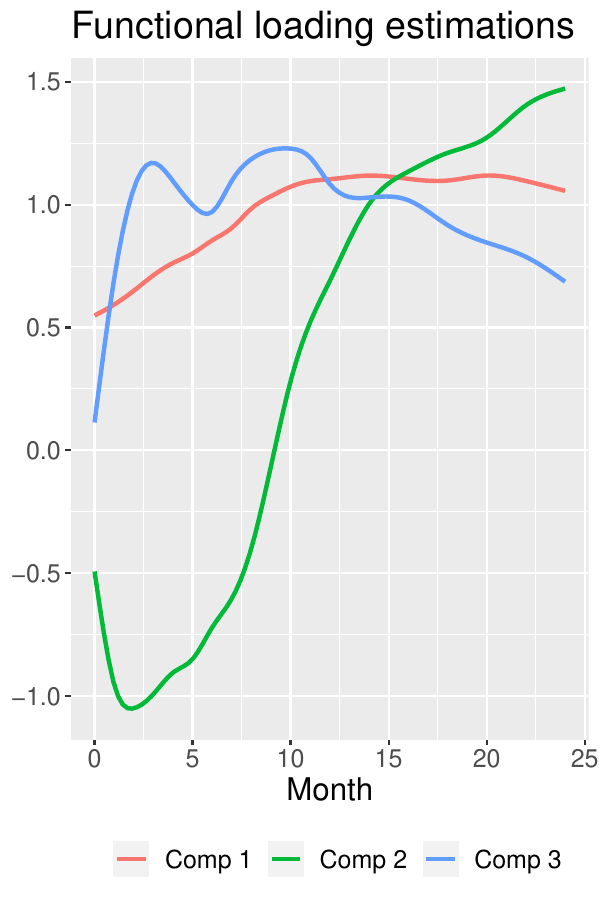}}
			\hskip1cm
			\subfigure{\includegraphics[height = 3.5in]{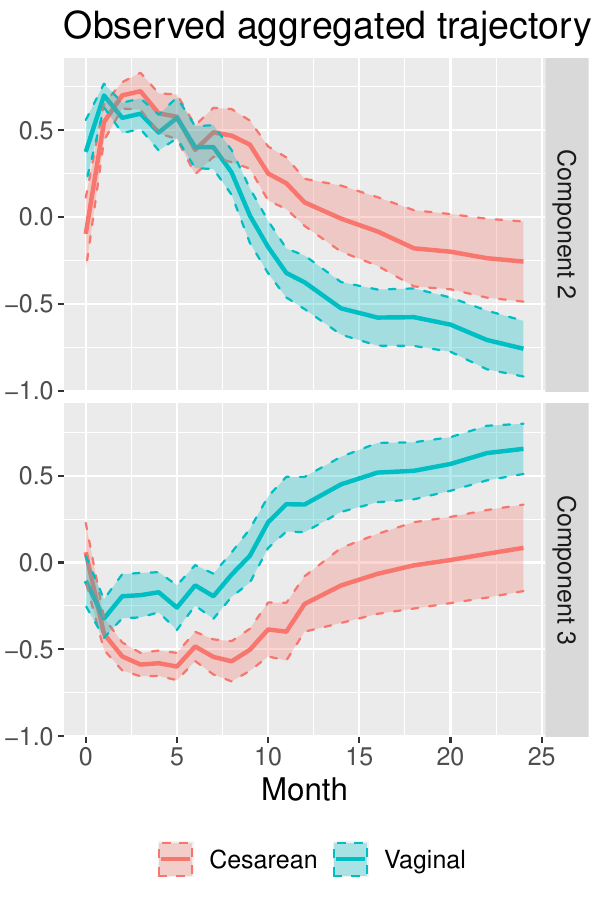}}
			\caption{Three singular functions from ECAM data and the aggregated observed trajectory with respect to the second and third singular vector on bacteria mode. Error bands of trajectories are obtained using mean $\pm 1.64\times$ (standard error of the mean).} \label{fig: phase transition}\label{fig:ECAM-func-PC}
		\end{figure}

		\section{Discussion}\label{sec:discuss}

		Although the presentation of this paper focuses on order-3 tensor with two tabular modes and one functional mode, the algorithms and all the theoretical results can be generalized to arbitrary order-$d$ tensors with $(d-1)$ tabular modes and one functional modes. It is an interesting future direction to further study functional tensor SVD for tensors with multiple functional modes, which may widely appear in spatial-temporal data analysis and imaging processing. In addition, this paper focuses on the functional tensor SVD when all the observations are from a grid of time points. It is interesting to further study the same topic under irregular sampling schemes where the sampling time points $t_{ijk}$ may be different across different units and variables.
		
		This work focuses on unsupervised dimension reduction for high-order functional data. The methods can be extended to the supervised methods that allows for classification or predictions. Specifically, one can first apply tensor SVD on the training high-order functional data, then train a supervised model with the estimated low-dimensional subjects loadings $\hat a$ as the predictors; when the observations/measurements from out-sample subjects are available, one can project the raw high-dimensional functional data on the estimated feature and functional loadings (i.e., $\hat b$ and $\hat{\xi}$) to obtain the corresponding subject scores, then apply the pre-fitted model for prediction. When the label/response is available for each subject, a data-driven cross-validation scheme can be applied to select the hyperparameters $r$ and $C_\xi$.

		It is also worth mentioning that there are several recent results that directly solve the tabular tensor decomposition by (accelerated) gradient method with convergence guarantee~\citep{cai2019nonconvex,han2021optimal,tong2021scaling}. Such methods usually have an advantage over power iteration as they directly aim at the optimization of likelihood, and can successfully remove the effect of incoherence from the final estimation error bound. However, it is unclear that whether and how these gradient-based algorithms and analysis techniques can be applied to the functional setting, which might be another interesting research direction in the future.

	\end{sloppypar}

	\bibliography{reference}

\begin{thebibliography}{}

\bibitem[Adamczak, 2008]{adamczak2008tail}
Adamczak, R. (2008).
\newblock A tail inequality for suprema of unbounded empirical processes with
  applications to markov chains.
\newblock {\em Electronic Journal of Probability}, 13:1000--1034.

\bibitem[Adler and Taylor, 2009]{adler2009random}
Adler, R.~J. and Taylor, J.~E. (2009).
\newblock {\em Random fields and geometry}.
\newblock Springer Science \& Business Media.

\bibitem[Allen, 2013]{allen2013multi}
Allen, G.~I. (2013).
\newblock Multi-way functional principal components analysis.
\newblock In {\em 2013 5th IEEE International Workshop on Computational
  Advances in Multi-Sensor Adaptive Processing (CAMSAP)}, pages 220--223. IEEE.

\bibitem[Anandkumar et~al., 2014a]{anandkumar2014tensor}
Anandkumar, A., Ge, R., Hsu, D., Kakade, S.~M., and Telgarsky, M. (2014a).
\newblock Tensor decompositions for learning latent variable models.
\newblock {\em The Journal of Machine Learning Research}, 15(1):2773--2832.

\bibitem[Anandkumar et~al., 2014b]{anandkumar2014guaranteed}
Anandkumar, A., Ge, R., and Janzamin, M. (2014b).
\newblock Guaranteed non-orthogonal tensor decomposition via alternating
  rank-$1 $ updates.
\newblock {\em arXiv preprint arXiv:1402.5180}.

\bibitem[Bartlett et~al., 2005]{bartlett2005local}
Bartlett, P.~L., Bousquet, O., and Mendelson, S. (2005).
\newblock Local rademacher complexities.
\newblock {\em The Annals of Statistics}, 33(4):1497--1537.

\bibitem[Bartlett and Mendelson, 2002]{bartlett2002rademacher}
Bartlett, P.~L. and Mendelson, S. (2002).
\newblock Rademacher and gaussian complexities: Risk bounds and structural
  results.
\newblock {\em Journal of Machine Learning Research}, 3(Nov):463--482.

\bibitem[Bokulich et~al., 2016]{bokulich2016antibiotics}
Bokulich, N.~A., Chung, J., Battaglia, T., Henderson, N., Jay, M., Li, H.,
  Lieber, A.~D., Wu, F., Perez-Perez, G.~I., and Chen, Y. (2016).
\newblock Antibiotics, birth mode, and diet shape microbiome maturation during
  early life.
\newblock {\em Science translational medicine}, 8(343):343ra82--343ra82.

\bibitem[Cai et~al., 2019]{cai2019nonconvex}
Cai, C., Li, G., Poor, H.~V., and Chen, Y. (2019).
\newblock Nonconvex low-rank symmetric tensor completion from noisy data.
\newblock {\em arXiv preprint arXiv:1911.04436}.

\bibitem[Cai et~al., 2020]{cai2020non-asymptotic}
Cai, T.~T., Han, R., and Zhang, A.~R. (2020).
\newblock On the non-asymptotic concentration of heteroskedastic wishart-type
  matrix.
\newblock {\em arXiv preprint arXiv:2008.12434}.

\bibitem[Cai and Yuan, 2012]{cai2012minimax}
Cai, T.~T. and Yuan, M. (2012).
\newblock Minimax and adaptive prediction for functional linear regression.
\newblock {\em Journal of the American Statistical Association},
  107(499):1201--1216.

\bibitem[Chen et~al., 2020]{chen2020semiparametric}
Chen, E.~Y., Xia, D., Cai, C., and Fan, J. (2020).
\newblock Semiparametric tensor factor analysis by iteratively projected svd.
\newblock {\em arXiv preprint arXiv:2007.02404}.

\bibitem[Chen et~al., 2017]{chen2017modelling}
Chen, K., Delicado~Useros, P.~F., and M{\"u}ller, H.-G. (2017).
\newblock Modelling function-valued stochastic processes, with applications to
  fertility dynamics.
\newblock {\em Journal of the Royal Statistical Society. Series B, Statistical
  Methodology}, 79(1):177--196.

\bibitem[Chen et~al., 2021]{chen2021factor}
Chen, R., Yang, D., and Zhang, C.-H. (2021).
\newblock Factor models for high-dimensional tensor time series.
\newblock {\em Journal of the American Statistical Association}, pages 1--23.

\bibitem[De~Lathauwer et~al., 2000]{de2000best}
De~Lathauwer, L., De~Moor, B., and Vandewalle, J. (2000).
\newblock On the best rank-1 and rank-(r 1, r 2,..., rn) approximation of
  higher-order tensors.
\newblock {\em SIAM journal on Matrix Analysis and Applications},
  21(4):1324--1342.

\bibitem[Fan et~al., 2015]{fan2015functional}
Fan, Y., James, G.~M., and Radchenko, P. (2015).
\newblock Functional additive regression.
\newblock {\em The Annals of Statistics}, 43(5):2296--2325.

\bibitem[Gu, 2013]{gu2013smoothing}
Gu, C. (2013).
\newblock {\em Smoothing spline ANOVA models}, volume 297.
\newblock Springer Science \& Business Media.

\bibitem[Han et~al., 2019]{han2019isotonic}
Han, Q., Wang, T., Chatterjee, S., and Samworth, R.~J. (2019).
\newblock Isotonic regression in general dimensions.
\newblock {\em Annals of Statistics}, 47(5):2440--2471.

\bibitem[Han et~al., 2020a]{han2020lloyd}
Han, R., Luo, Y., Wang, M., and Zhang, A.~R. (2020a).
\newblock Exact clustering in tensor block model: Statistical optimality and
  computational limit.
\newblock {\em arXiv preprint arXiv:2012.09996}.

\bibitem[Han et~al., 2021]{han2021optimal}
Han, R., Willett, R., and Zhang, A.~R. (2021).
\newblock An optimal statistical and computational framework for generalized
  tensor estimation.
\newblock {\em The Annals of Statistics}, to appear.

\bibitem[Han et~al., 2020b]{han2020tensor}
Han, Y., Chen, R., Yang, D., and Zhang, C.-H. (2020b).
\newblock Tensor factor model estimation by iterative projection.
\newblock {\em arXiv preprint arXiv:2006.02611}.

\bibitem[Happ and Greven, 2018]{happ2018multivariate}
Happ, C. and Greven, S. (2018).
\newblock Multivariate functional principal component analysis for data
  observed on different (dimensional) domains.
\newblock {\em Journal of the American Statistical Association},
  113(522):649--659.

\bibitem[Happ-Kurz, 2020]{happ2020object}
Happ-Kurz, C. (2020).
\newblock Object-oriented software for functional data.
\newblock {\em Journal of Statistical Software}, 93(5):1--38.

\bibitem[Hasenstab et~al., 2017]{hasenstab2017multi}
Hasenstab, K., Scheffler, A., Telesca, D., Sugar, C.~A., Jeste, S., DiStefano,
  C., and {\c{S}}ent{\"u}rk, D. (2017).
\newblock A multi-dimensional functional principal components analysis of eeg
  data.
\newblock {\em Biometrics}, 73(3):999--1009.

\bibitem[Hong et~al., 2020]{hong2020generalized}
Hong, D., Kolda, T.~G., and Duersch, J.~A. (2020).
\newblock Generalized canonical polyadic tensor decomposition.
\newblock {\em SIAM Review}, 62(1):133--163.

\bibitem[Hu and Yao, 2021]{hu2021dynamic}
Hu, X. and Yao, F. (2021).
\newblock Dynamic principal subspaces with sparsity in high dimensions.
\newblock {\em arXiv preprint arXiv:2104.03087}.

\bibitem[James et~al., 2000]{james2000principal}
James, G.~M., Hastie, T.~J., and Sugar, C.~A. (2000).
\newblock Principal component models for sparse functional data.
\newblock {\em Biometrika}, 87(3):587--602.

\bibitem[Kimeldorf and Wahba, 1971]{kimeldorf1971some}
Kimeldorf, G. and Wahba, G. (1971).
\newblock Some results on tchebycheffian spline functions.
\newblock {\em Journal of mathematical analysis and applications},
  33(1):82--95.

\bibitem[Koltchinskii and Yuan, 2010]{koltchinskii2010sparsity}
Koltchinskii, V. and Yuan, M. (2010).
\newblock Sparsity in multiple kernel learning.
\newblock {\em The Annals of Statistics}, 38(6):3660--3695.

\bibitem[Kruskal, 1976]{kruskal1976more}
Kruskal, J.~B. (1976).
\newblock More factors than subjects, tests and treatments: an indeterminacy
  theorem for canonical decomposition and individual differences scaling.
\newblock {\em Psychometrika}, 41(3):281--293.

\bibitem[Martino et~al., 2021]{martino2021context}
Martino, C., Shenhav, L., Marotz, C.~A., Armstrong, G., McDonald, D.,
  V{\'a}zquez-Baeza, Y., Morton, J.~T., Jiang, L., Dominguez-Bello, M.~G., and
  Swafford, A.~D. (2021).
\newblock Context-aware dimensionality reduction deconvolutes gut microbial
  community dynamics.
\newblock {\em Nature biotechnology}, 39(2):165--168.

\bibitem[Mendelson, 2002]{mendelson2002geometric}
Mendelson, S. (2002).
\newblock Geometric parameters of kernel machines.
\newblock In {\em International Conference on Computational Learning Theory},
  pages 29--43. Springer.

\bibitem[Mercer, 1909]{mercer1909xvi}
Mercer, J. (1909).
\newblock Xvi. functions of positive and negative type, and their connection
  the theory of integral equations.
\newblock {\em Philosophical transactions of the royal society of London.
  Series A, containing papers of a mathematical or physical character},
  209(441-458):415--446.

\bibitem[Micchelli and Wahba, 1981]{micchelli1981design}
Micchelli, C.~A. and Wahba, G. (1981).
\newblock Design problems for optimal surface interpolation.
\newblock In Ziegler, Z., editor, {\em Approximation Theory and Applications},
  pages 329--347. Academic Press, New York.

\bibitem[Pisier, 1983]{pisier1983some}
Pisier, G. (1983).
\newblock Some applications of the metric entropy condition to harmonic
  analysis.
\newblock In {\em Banach Spaces, Harmonic Analysis, and Probability Theory},
  pages 123--154. Springer.

\bibitem[Ramsay and Silverman, 2006]{ramsay2006functional}
Ramsay, J. and Silverman, B. (2006).
\newblock {\em Functional Data Analysis}.
\newblock Springer Science \& Business Media.

\bibitem[Raskutti et~al., 2012]{raskutti2012minimax}
Raskutti, G., J~Wainwright, M., and Yu, B. (2012).
\newblock Minimax-optimal rates for sparse additive models over kernel classes
  via convex programming.
\newblock {\em Journal of Machine Learning Research}, 13(2).

\bibitem[Rice and Silverman, 1991]{rice1991estimating}
Rice, J.~A. and Silverman, B.~W. (1991).
\newblock Estimating the mean and covariance structure nonparametrically when
  the data are curves.
\newblock {\em Journal of the Royal Statistical Society: Series B
  (Methodological)}, 53(1):233--243.

\bibitem[Rudelson and Vershynin, 2013]{rudelson2013hanson}
Rudelson, M. and Vershynin, R. (2013).
\newblock Hanson-wright inequality and sub-gaussian concentration.
\newblock {\em Electronic Communications in Probability}, 18:1--9.

\bibitem[Stewart et~al., 2018]{stewart2018temporal}
Stewart, C.~J., Ajami, N.~J., O’Brien, J.~L., Hutchinson, D.~S., Smith,
  D.~P., Wong, M.~C., Ross, M.~C., Lloyd, R.~E., Doddapaneni, H., Metcalf,
  G.~A., et~al. (2018).
\newblock Temporal development of the gut microbiome in early childhood from
  the teddy study.
\newblock {\em Nature}, 562(7728):583--588.

\bibitem[Sun and Li, 2019]{sun2019dynamic}
Sun, W.~W. and Li, L. (2019).
\newblock Dynamic tensor clustering.
\newblock {\em Journal of the American Statistical Association},
  114(528):1894--1907.

\bibitem[Sun et~al., 2017]{sun2017provable}
Sun, W.~W., Lu, J., Liu, H., and Cheng, G. (2017).
\newblock Provable sparse tensor decomposition.
\newblock {\em Journal of the Royal Statistical Society: Series B (Statistical
  Methodology)}, 79(3):899--916.

\bibitem[Tong et~al., 2021]{tong2021scaling}
Tong, T., Ma, C., Prater-Bennette, A., Tripp, E., and Chi, Y. (2021).
\newblock Scaling and scalability: Provable nonconvex low-rank tensor
  estimation from incomplete measurements.
\newblock {\em arXiv preprint arXiv:2104.14526}.

\bibitem[Vershynin, 2018]{vershynin2018high}
Vershynin, R. (2018).
\newblock {\em High-dimensional probability: An introduction with applications
  in data science}, volume~47.
\newblock Cambridge university press.

\bibitem[Wang et~al., 2020a]{wang2020functional2}
Wang, D., Zhao, Z., Willett, R., and Yau, C.~Y. (2020a).
\newblock Functional autoregressive processes in reproducing kernel hilbert
  spaces.
\newblock {\em arXiv preprint arXiv:2011.13993}.

\bibitem[Wang et~al., 2020b]{wang2020functional}
Wang, D., Zhao, Z., Yu, Y., and Willett, R. (2020b).
\newblock Functional linear regression with mixed predictors.
\newblock {\em arXiv preprint arXiv:2012.00460}.

\bibitem[Wang et~al., 2020c]{wang2020low}
Wang, J., Wong, R.~K., and Zhang, X. (2020c).
\newblock Low-rank covariance function estimation for multidimensional
  functional data.
\newblock {\em Journal of the American Statistical Association}, pages 1--14.

\bibitem[Wang et~al., 2016]{wang2016functional}
Wang, J.-L., Chiou, J.-M., and M{\"u}ller, H.-G. (2016).
\newblock Functional data analysis.
\newblock {\em Annual Review of Statistics and Its Application}, 3:257--295.

\bibitem[Wedin, 1972]{wedin1972perturbation}
Wedin, P.-{\AA}. (1972).
\newblock Perturbation bounds in connection with singular value decomposition.
\newblock {\em BIT Numerical Mathematics}, 12(1):99--111.

\bibitem[Yao et~al., 2005]{yao2005functional-regression}
Yao, F., M{\"u}ller, H.-G., and Wang, J.-L. (2005).
\newblock Functional linear regression analysis for longitudinal data.
\newblock {\em The Annals of Statistics}, pages 2873--2903.

\bibitem[Yuan and Cai, 2010]{yuan2010reproducing}
Yuan, M. and Cai, T.~T. (2010).
\newblock A reproducing kernel hilbert space approach to functional linear
  regression.
\newblock {\em The Annals of Statistics}, 38(6):3412--3444.

\bibitem[Zhang and Han, 2019]{zhang2019optimal-statsvd}
Zhang, A. and Han, R. (2019).
\newblock Optimal sparse singular value decomposition for high-dimensional
  high-order data.
\newblock {\em Journal of the American Statistical Association}, pages
  1708--1725.

\bibitem[Zhang and Xia, 2018]{zhang2018tensor}
Zhang, A. and Xia, D. (2018).
\newblock Tensor {SVD}: Statistical and computational limits.
\newblock {\em IEEE Transactions on Information Theory}, 64(11):7311--7338.

\bibitem[Zhang and Zhou, 2020]{zhang2018non}
Zhang, A.~R. and Zhou, Y. (2020).
\newblock On the non-asymptotic and sharp lower tail bounds of random
  variables.
\newblock {\em Stat}, 9(1):e314.

\bibitem[Zhang et~al., 2020]{zhang2020denoising}
Zhang, C., Han, R., Zhang, A.~R., and Voyles, P.~M. (2020).
\newblock Denoising atomic resolution 4d scanning transmission electron
  microscopy data with tensor singular value decomposition.
\newblock {\em Ultramicroscopy}, 219:113123.

\bibitem[Zhang et~al., 2019]{zhang2019tensor}
Zhang, Z., Allen, G.~I., Zhu, H., and Dunson, D. (2019).
\newblock Tensor network factorizations: Relationships between brain structural
  connectomes and traits.
\newblock {\em Neuroimage}, 197:330--343.

\end{thebibliography}
	\bibliographystyle{apalike}

	\appendix
	\newpage
	\setcounter{page}{1}

	\begin{center}
		{\bf\LARGE Supplement to ``Guaranteed Functional } 
		
		\medskip
		
		{\bf\LARGE Tensor Singular Value Decomposition"}
	\end{center}
	\smallskip
	\begin{center}
		Rungang Han, ~ Pixu Shi, ~ and ~ Anru R. Zhang\\
		~
	\end{center}
	
	

	\section{Additional Numeric Results}\label{sec:additional_nums}

	\subsection{Additional Simulations}\label{sec:additional_simus}
	We collect additional simulation experiments in this section. We first compare the performance of FTSVD with the existing FDA algorithms when the functional remainder term $\bcZ$ exists. To this end, we first generate the signal tensor $\bcX$ similarly as we did in Section \ref{sec:simulation} and we then generate a random perturbation tensor $\bcZ$ such that each $\bcZ_{ij\cdot}$ are drawn independently in the same way as $\xi_l$. We calculate $\bcY = \bcX + \sigma \cdot \bcZ$. Then, $\sigma$ essentially controls the amplitude of remainder functions. Here we only compare the estimation accuracy for the singular function as the FDA-based method does not directly provide the tabular loading estimates. We take $\tau = 0$ this time for simplicity, $p=20$ and $50$, and vary the amplitude of the functional remainder $\sigma$. The grid density is fixed to be $n=50$ and the least singular value $\lambda_{min} = 2$. The result is presented in Figure \ref{fig:cmp-rank-2}. We note that FPCA has uniformly higher estimation errors than the other two tensor-based methods, as FPCA is designed for i.i.d. samples but may not be suitable for the heterogeneous data considered in this paper. On the other hand, one can see that CP and RKHS has almost the same performance. This is because in this particular setting, the statistical rate of FTSVD is dominated by $\cE/\lambda_{min}$ as suggested by Theorem \ref{thm:local-convergence-RKHS}, which is nearly the same as the one for CP~\citep{anandkumar2014guaranteed}. However, it should be noted that in the general settings where both functional remainder term and observational noises exist, CP may not be as accurate as our method according to the results of the first simulation setting in Section \ref{sec:simulation}.
	\begin{figure}[htbp]
		\centering
		\includegraphics[width=1\textwidth]{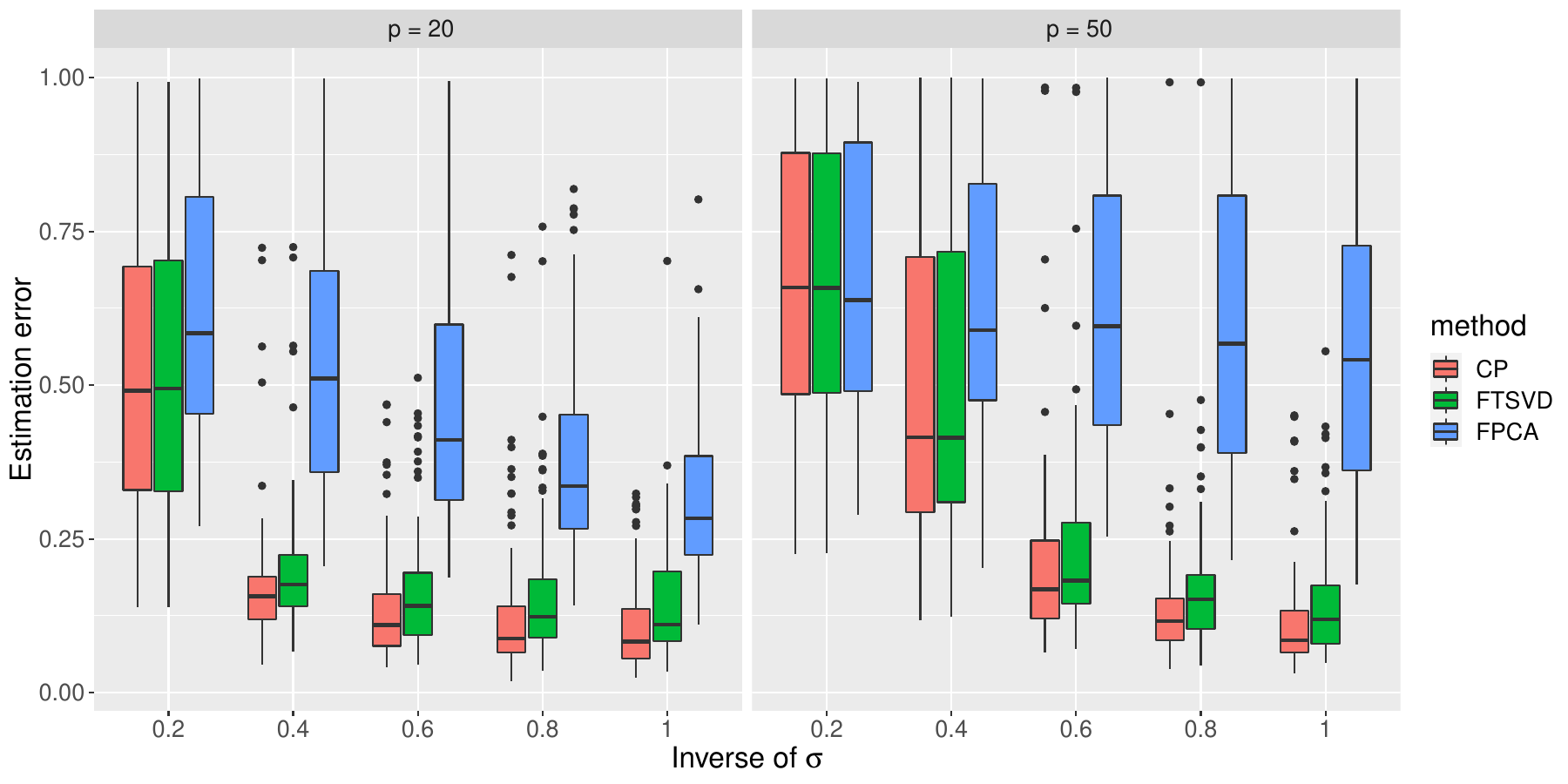}
		\caption{Estimation error of singular functions under different dimension $p$ and amplitude of remainder $\sigma$ for rank-$1$ models.}\label{fig:cmp-rank-2}
	\end{figure}

	We also compare the proposed FTSVD method with multivariate functional principle component (MFPCA) \citep{happ2018multivariate}, a classic method for multivariate functional data analysis. Different from FTSVD, MFPCA yields different functional principle components for different features. To make a direct and fair comparison of FTSVD and MFPCA, we focus on the scenario when there is one singular/principle component. Specifically, we apply MFPCA\footnote{See \cite{happ2020object} for discussions of the package.} on the simulated data to estimate the functional principal components $\{\hat\psi_1(t),\ldots,\hat\psi_{p_2}(t)\}$, where $p_2$ is the number of features. Then, we apply the univariate FPCA on the estimated functional principal components to obtain one final functional estimate $\hat \xi(t)$ and compare it with the output of FTSVD. We set $\lambda = 2$, $n = 30$, and $p_1 = 100$ (note $p_1$ here is the number of subjects), we then choose different numbers of functional features $p_2 \in \{2,5,20,50\}$ with varying noise level $\tau$. The estimation errors are reported in Table \ref{tab:MFPCA}. Table \ref{tab:MFPCA} shows FTSVD is more robust in high-dimensional (i.e., large $p_2$) and noisy (i.e., large $\tau$) regimes and outperforms MFPCA in all scenarios. 
	
	\begin{table}
	\footnotesize
		\centering
		\begin{tabular}{c|ccccccc}
			\hline\hline
			$p_2$ &  & $\tau = 0$ & $\tau = 0.05$ & $\tau = 0.1$ & $\tau = 0.2$ & $\tau = 0.5$ & $\tau = 1$ \\ \hline
			2 & FTSVD & 0.050 & 0.079 & 0.090 & 0.111 & 0.214 & 0.405 \\
			& MFPCA & 0.182 & 0.211 & 0.210 & 0.196 & 0.279 & 0.425 \\ \hline
			5 & FTSVD & 0.042 & 0.074  & 0.087  & 0.121 & 0.214 & 0.440 \\
			& MFPCA & 0.175 & 0.200 & 0.202 &  0.212 (1) & 0.305 & 0.523 \\\hline
			20 & FTSVD & 0.044 & 0.073 & 0.091 & 0.122 & 0.233 & 0.652 \\
			& MFPCA & 0.190 & 0.202  & 0.199 (1) & 0.226 (1) & 0.358  & 0.731 (4)\\ \hline
			50 & FTSVD & 0.056 & 0.074 & 0.089 & 0.131 & 0.245 & 0.845 \\
			& MFPCA & 0.220 & 0.221 (1) & 0.230 (1) & 0.266 (1) & 0.474 (1) & 0.863 (7)\\ \hline \hline
		\end{tabular}
		\caption{Averaged estimation error of FTSVD and MFPCA over 100 experiments. Numbers of failure times of executions (incurred by ill-conditioned matrix calculations) are shown in parentheses.}
		\label{tab:MFPCA}
	\end{table}

	\subsection{Crop Production Data Analysis}\label{sec:crop-data}
    We use the world-wide crops production data (available on \url{http://www.fao.org/faostat/en/#data/QC}) as the second illustration. The dataset contains the annual production of 173 products between 1961 and 2019 from different countries and areas in the world. For simplicity, we only consider 19 continent-based areas including Eastern Asian, Northern America, West Europe, Central Africa, etc. We also select the most widely-planted 26 crops for our analysis. Therefore, the data we consider can be organized as a $19$-by-$26$-by-$59$ tensor $\widetilde\bcY$, with $\widetilde\bcY_{ijk}$ representing the production of the $j$th crop item on the $i$th area in the $k$th year. 
	
	We apply the proposed algorithm with rank $r=2$ to $\widetilde\bcY$ after centralization and it explains $78.7\%$ variations of the data. The estimated singular vectors and functions for each component are presented in Figure \ref{fig:crops}. The first component captures the overall production of areas and items through the years.  The magnitudes of the first singular vectors on area and item modes (i.e., x-coordinate of the bi-plot) coincide with the overall production of the areas and items, respectively, with Northern America and Eastern Asia being the top two crop production areas, and maze, rice paddy, wheat being the three major food crops. The increasing first singular function on time mode coupled with positive singular vectors on area and item modes reflects the increasing trend in overall production for most areas and crops. The second component characterizes the variation in production across different area-item pairs. Such variation is quantified by the second singular function on time mode coupled with the second item singular vectors on item and area modes (i.e., y-coordinate of the bi-plots). For example, Northern America has a large negative value in Component 2, with the same sign as maize but different from rice paddy. This implies maze takes a larger share in Northern America compared to its share in other areas, while rice paddy takes a smaller share in Northern America. An opposite conclusion can also be made on the three Asian areas, whose large positive values in the second component have the same sign as rice paddy, indicating the larger share of rice paddy in Asia compared to other areas. The small difference between the two singular functions on time mode indicate that the variation across area-item pairs are constant across time.

	\begin{figure}[htbp]
		\centering
		\includegraphics[width=.8\textwidth]{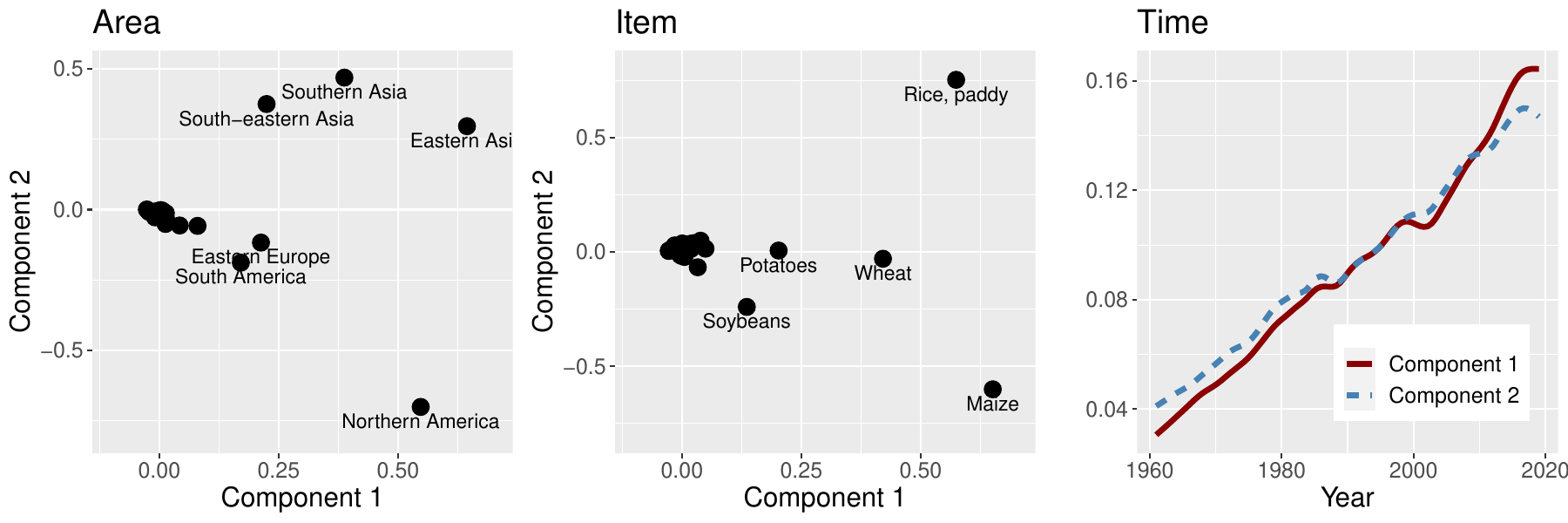}
		\caption{Estimated singular vectors and functions from crop production data. The left and central panels are biplots for the two estimated singular vectors on area and item modes respectively. Points with significant scorings are labeled. The right panel plots the estimated singular function on the time mode.}\label{fig:crops}
	\end{figure}

	\section{Proofs}\label{sec:proofs}
	
	\subsection{Proof of Proposition \ref{prop:identify}}
	We utilize the Indeterminacy Theorem developed by \cite{kruskal1976more} for finite-dimensional matrix and prove the two scenarios separately. We first introduce the matrix representation $A = [a_1,\ldots, a_r] \in \bbR^{p_1 \times r}$ and $B = [b_1,\ldots, b_r] \in \bbR^{p_2 \times r}$. Note that except on a set of measure zero, $A$ and $B$ are full rank. We also denote $\widetilde A = [\tilde a_1,\ldots, \tilde a_r] \in \bbR^{p_1 \times r}$, $\widetilde B = [\tilde b_1,\ldots, \tilde b_r] \in \bbR^{p_2 \times r}$ and denote $\widetilde \bcX = \sum_{l=1}^r \tilde \lambda_l \tilde a_l \circ \tilde b_l \circ \tilde c_l$.
	\begin{itemize}
	    \item $\cH$ is a finite $p_3$-dimensional functional space. We assume the functional basis for $\cH$ are $\{\phi_1,\ldots, \phi_{p_3}\}$. Let $C, \widetilde C \in \bbR^{p_3 \times r}$ such that
	    \begin{equation*}
	        \begin{split}
	            \xi_k &= \sum_{l=1}^{p_3} C_{lk}\phi_{l}, \qquad \tilde{\xi}_k = \sum_{l=1}^{p_3} \widetilde C_{lk}\phi_{l}.
	        \end{split}
	    \end{equation*}
	    Again, $C$ is of full rank except on a set of measure zero. We introduce the tabular tensors $\cX, \widetilde\cX \in \bbR^{p_1\times p_2 \times p_3}$ such that
	    \begin{equation*}
	        \cX_{ijl} = \sum_{k=1}^r \lambda_r (a_k)_i(b_k)_j C_{lk},\qquad \widetilde\cX_{ijl} = \sum_{k=1}^r \lambda_r  (a_k)_i(b_k)_j \widetilde C_{lk}
	    \end{equation*}
	    Then each function $\bcX_{ij\cdot}$ can be represented as
	    \begin{equation*}
	        \bcX_{ij\cdot} = \sum_{k=1}^r \lambda_r (a_k)_i (b_k)_j \sum_{l=1}^{p_3} C_{lk}\phi_l = \sum_{l=1}^{p_3} \left(\sum_{k=1}^r \lambda_r (a_k)_i (b_k)_j C_{lk}\right)\phi_l = \sum_{l=1}^{p_3} \cX_{ijl} \phi_l.
	    \end{equation*}
	    Since $\{\phi_l\}_{l=1}^{p_3}$ are is a functional basis and $\bcX = \widetilde\bcX$, we must have $\cX = \widetilde \cX$. When $A,B,C$ are all full rank, the Indeterminacy Theorem in \cite{kruskal1976more} implies that when $2r < p_1+p_2+p_3+2$, $\{A,B,C\}$ and $\{\widetilde A, \widetilde B, \widetilde C\}$ are identical up to permutation and sign-flipping, and the model identifiability naturally follows.
	    
	    \item $\cH$ has infinite dimensions and all $\xi_l$s are continuous. We discretize the functional tensor on a grid $\left[\frac{1}{n},\ldots, \frac{n-1}{n},1\right]$ for some positive integer $n$. Denote $\bcX^{(n)}, \widetilde \bcX^{(n)} \in \bbR^{p_1\times p_2 \times n}$ with
	    \begin{equation*}
	        \cX^{(n)}_{ijl} = \bcX_{ij\frac{l}{n}}, \qquad \widetilde\cX^{(n)}_{ijl} = \widetilde\bcX_{ij\frac{l}{n}},\qquad l=1,\ldots,n. 
	    \end{equation*}
	    Since $\bcX = \widetilde \bcX$, we must have $\bcX^{(n)} = \widetilde \bcX^{(n)}$. In other words,
	    \begin{equation*}
	        \sum_{k=1}^r \lambda_r A_{ik}B_{jk}\Xi^{(n)}_{lk} = \sum_{k=1}^r \lambda_r \widetilde A_{ik} \widetilde B_{jk}\widetilde \Xi^{(n)}_{lk},
	    \end{equation*}
	    where $\Xi^{(n)}_lk = \xi_l(k/n)$ and $\widetilde\Xi^{(n)}_{lk} = \xi_l(k/n)$. Since $\{\xi_l\}_{l=1}^n$ are orthogonal with each other, when $n$ is sufficiently large, it is guaranteed that $\Xi$ has full rank $r$. Then the Indeterminacy Theorem in \cite{kruskal1976more} implies that $\{A,B,\Xi^{(n)}\}$ and $\widetilde A, \widetilde B, \widetilde \Xi^{(n)}$ are identical up to permutation and sign-flipping when $r < p_1+p_2-2$. We assume $\Xi^{(n)} = \widetilde \Xi^{(n)}$ without loss of generality. With $n$ tends to infinity and the fact that each $\xi_l$ is continuous, we obtain that $\xi_l = \tilde \xi_l$ and the identifiability is proved.
	    
	\end{itemize}

	\subsection{Proof of Theorem \ref{thm:local-convergence-RKHS}} In order to prove Theorem \ref{thm:local-convergence-RKHS}, we introduce the following events:
	\begin{enumerate}[label=(\subscript{A}{\arabic*})]
		\item 
		\begin{equation}\label{ineq:A-1}
			\left|\int_0^1 f(s)ds - \frac{1}{n}\sum_{k=1}^n f(s_k)\right| \leq C_1\left(\zeta_n \|f\|_{\cL^2} + \zeta_n^2 \|f\|_{\cH}\right),\qquad \forall f\in \cH.
		\end{equation}
		\item
		\begin{equation}\label{ineq:A-2}
			\begin{split}
				\left|\left\|f\right\|_{\cL^2}^2 - \|f\|_n^2\right| \leq C_2 \|f\|_\cH\left(\zeta_n \|f\|_{\cL^2} + \zeta_n^2\|f\|_{\cH} \right), \qquad \forall f \in \cH.
			\end{split}
		\end{equation}
		\item
		\begin{equation}\label{ineq:A-3}
			\begin{split}
				& \sup_{a \in \bbS^{p_1-1}, b \in \bbS^{p_2-1}} \frac{1}{n}\sum_{i=1}^{p_1}\sum_{j=1}^{p_2}\sum_{k=1}^n a_ib_jf(s_k)\varepsilon_{ijk} \\
				& \leq C_3\tau\left\{\left(\zeta_n^2 + \frac{(p_1+p_2)\log n}{n}\right)\|f\|_\cH + \left(\zeta_n + \sqrt{\frac{p_1+p_2}{n}}\right)\|f\|_{\cL^2}\right\},\qquad \forall f\in \cH.
			\end{split}
		\end{equation}
	\end{enumerate}
	Here recall that $\|f\|_n := \left(n^{-1}\sum_{k=1}^n f^2(s_k)\right)^{1/2}$, $C_1,C_2,C_3$ are some universal constants, and $\tau$ is the standard deviation of the Gaussian random noises $\varepsilon_{ijk}$s. By Lemmas \ref{lm:MC-int}, \ref{lm:MC-square-int}, and \ref{lm:Gaussian-complexity}, we have
	\begin{equation*}
		\bbP\left(A_1\right) \geq 1-Cn^{-9},\qquad \bbP\left(A_2\right) \geq 1-Cn^{-9},\qquad \bbP\left(A_3\right) \geq 1-Cn^{-9}-C\log n\cdot\exp(-c(p_1+p_2)).
	\end{equation*}
	The following analyses will be conducted under the event $A:=A_1\cap A_2 \cap A_3$, which holds with probability at least $1-Cn^{-9}-C\log n\cdot\exp(-c(p_1+p_2))$.
	
	We next introduce the following two Lemmas to establish the functional and tabular error contractions for the proposed RKHS-constraint power iteration (Algorithm \ref{alg:power-iter}). 
	\begin{Lemma}\label{lm:one-step-constraction-functional}
		Suppose the conditions in Theorem \ref{thm:local-convergence-RKHS} and $(A_1)-(A_3)$ hold. For any $t=0,\ldots,T-1$, if $\dist( a^{(t)}, a_l), \dist( b^{(t)}, b_l) \leq \delta^{(t)}$ for some $\delta^{(t)} \leq c_0/(\kappa r)$, then  $\|\xi^{(t+1)}\|_\cH \leq 2\kappa C_\xi$ and
		\begin{equation*}
			\begin{split}
				\max\{\dist(\xi^{(t+1)}, \xi_l), \dist(\xi^{(t+1)}_n, (\xi_l)_n)\} & \leq \frac{1}{2}\delta^{(t)} + C\kappa(\mu(r-1) + \zeta_n ) \\
				& \quad + C\left\{\frac{\cE}{\lambda_{min}} + \frac{\tau}{\lambda_{min}}\left(\zeta_n + \sqrt{\frac{p_1+p_2}{n}}\right)\right\}.
			\end{split} 
		\end{equation*}
	\end{Lemma}
	\begin{proof}
		See Section \ref{sec:pr-main-lm}.
	\end{proof}
	
	\begin{Lemma}\label{lm:one-step-constraction-tabular}
		Suppose the conditions in Theorem \ref{thm:local-convergence-RKHS} and $(A_1)-(A_3)$ hold. For any $t=0,\ldots,T-1$, if $\max\{\dist( b^{(t)}, b_l), \dist((\xi^{(t)})_n, (\xi_l)_n), \dist(\xi^{(t)}, \xi_l)\} \leq \delta^{(t)}$ for some $\delta^{(t)} \leq c_0/(\kappa r)$ and $\|\xi^{(t)}\|_\cH \leq 2\kappa C_{\xi}$. Then, 
		\begin{equation*}
			\begin{split}
				\dist(a^{(t+1)}, a_l) \leq C\sqrt{r-1}\kappa(\mu+\delta^{(t)})(\mu+\delta^{(t)}+\zeta_n) + C\left\{\frac{\cE}{\lambda_{min}}  + \frac{\tau}{\lambda_{min}}\left(\zeta_n + \sqrt{\frac{p_1+p_2}{n}}\right)\right\}.
			\end{split} 
		\end{equation*}
	\end{Lemma}
	\begin{proof}
		See Section \ref{sec:pr-main-lm}.
	\end{proof}
	Note that we can obtain almost the same error contraction for $b^{(t+1)}$ as an analogue of Lemma \ref{lm:one-step-constraction-tabular}. Now we are ready to prove Theorem \ref{thm:local-convergence-RKHS}.
	\begin{proof}[Proof of Theorem \ref{thm:local-convergence-RKHS}]
		Without loss of generality, we assume $l=1$. We first prove by induction that for any $t=0,\ldots,T$,
		\begin{equation}\label{ineq:tabular-et<delta}
			\max \left\{\dist(a^{(t)}, a_1), \dist(b^{(t)}, b_1)\right\} \leq \frac{c_0}{\kappa r}.
		\end{equation}
		\begin{equation}\label{ineq:functional-et<delta}
			\dist((\xi^{(t)})_n,(\xi_1)_n) \leq \frac{c_0}{\kappa r}.
		\end{equation}
		\begin{equation}\label{ineq:functional-et-bound}
			\|\xi^{(t)}\|_\cH \leq 2\kappa C_\xi.
		\end{equation}
		We claim that \eqref{ineq:tabular-et<delta} implies \eqref{ineq:functional-et<delta} and \eqref{ineq:functional-et-bound} for any $t = 0,\ldots,T$. By Assumptions \ref{asmp:incoherence} and \ref{asmp:sample-size}, we have
		\begin{equation*}
			\mu \leq \frac{c}{\kappa^2r^2},\qquad \zeta_n \leq \frac{c}{\kappa^2 r},\quad 
			\text{and}\quad 
			\frac{\cE}{\lambda_{min}}  + \frac{\tau}{\lambda_{min}}\left(\zeta_n + \sqrt{\frac{p_1+p_2}{n}}\right) \leq \frac{c_0}{4\kappa r}.
		\end{equation*}
		Assuming \eqref{ineq:tabular-et<delta} and applying Lemma \ref{lm:one-step-constraction-functional}, one immediately obtains
		\begin{equation*}
			\dist(\xi_n^{(t)},(\xi_1)_n) \leq \frac{2c_0}{\kappa r} + C\kappa \cdot \frac{3c}{\kappa^2 r}  + \frac{c_0}{4\kappa r} \leq \frac{c_0}{\kappa r}
		\end{equation*}
		and $\|\xi^{(t)}\|_\cH \leq 2\kappa C_\xi$ for any $t = 0,\ldots,T$. So \eqref{ineq:functional-et<delta} and \eqref{ineq:functional-et-bound} are proved.
		Now assuming \eqref{ineq:tabular-et<delta}, \eqref{ineq:functional-et<delta} and \eqref{ineq:functional-et-bound} hold for some $t = t_0<T$, we show that \eqref{ineq:tabular-et<delta} also holds for $t = t_0 + 1$. By Lemma \ref{lm:one-step-constraction-tabular}, we have
		\begin{equation*}
			\begin{split}
				\dist(a^{(t_0+1)},a_1) \leq C\sqrt{r-1}\kappa \left(\frac{c}{\kappa^2r^2} + \frac{c_0}{\kappa r}\right)\left(\frac{c}{\kappa^2r^2} + \frac{c_0}{\kappa r} + \frac{c}{\kappa^2 r}\right) + \frac{c_0}{4\kappa r} \leq \frac{c_0}{\kappa r}.
			\end{split}
		\end{equation*}
		Similarly one can also show that $\dist(b^{(t_0+1)},b_1) \leq \frac{c_0}{\kappa r}$ and thus \eqref{ineq:functional-et<delta} holds for $t = t_0+1$. Finally, by Assumption \ref{asmp:initialization}, \eqref{ineq:tabular-et<delta} holds at $t=0$. Therefore, we have proved that \eqref{ineq:tabular-et<delta} and \eqref{ineq:functional-et<delta} hold for all $t=0,\ldots,T$ by induction. 
		
		In addition, note that by the upper bounds of $\mu,\zeta_n$ and $\delta^{(t)}$, we also have
		\begin{equation*}
			\sqrt{r-1}\kappa(\mu+\delta^{(t)})(\mu+\delta^{(t)}+\zeta_n) \leq \frac{1}{2C} \delta^{(t)} + \kappa(\mu (r-1)+\zeta_n).
		\end{equation*}
		Therefore, under the same condition of Lemma \ref{lm:one-step-constraction-tabular}, we further have
		\begin{equation}\label{ineq:contract-a-looser}
			\begin{split}
				\dist(a^{(t+1)}, a_1) & \leq \frac{1}{2}\delta^{(t)} + C\kappa(\mu(r-1) + \zeta_n) \\
				& \quad + C\left\{\frac{\cE}{\lambda_{min}} + \frac{\tau}{\lambda_{min}}\left(\zeta_n + \sqrt{\frac{p_1+p_2}{n}}\right)\right\}.
			\end{split} 
		\end{equation}
		Now we define 
		\begin{equation*}
			e^{(t)}:=\max\left\{\dist(a^{(t)}, a_1), \dist(b^{(t)}, b_1), \dist(\xi^{(t)}, \xi_1)\right\}.
		\end{equation*}
		Combining \eqref{ineq:contract-a-looser} with Lemma \ref{lm:one-step-constraction-functional}, we obtain that
		\begin{equation*}
			e^{(t+1)} \leq \frac{1}{2}e^{(t)} + C\left\{\kappa\left(\mu(r-1) + \zeta_n\right) + \frac{\cE}{\lambda_{min}} + \frac{\tau}{\lambda_{min}}\left(\zeta_n + \sqrt{\frac{p_1+p_2}{n}}\right)\right\},\qquad \forall t\geq 0.
		\end{equation*}
		Then it follows by induction that
		\begin{equation*}
			e^{(t)} \leq 2^{-t}e^{(0)} + 2C\left\{\kappa(\mu(r-1) + \zeta_n) + \frac{\cE}{\lambda_{min}} + \frac{\tau}{\lambda_{min}}\left(\zeta_n + \sqrt{\frac{p_1+p_2}{n}}\right)\right\},\qquad \forall t \geq 0.
		\end{equation*}
		
		In particular, when $r=1$, Lemma \ref{lm:one-step-constraction-tabular} directly implies that 
		\begin{equation*}
			\dist (a^{(t)}, a_1) \leq C\left\{\frac{\cE}{\lambda_{min}}  + \frac{\tau}{\lambda_{min}}\left(\zeta_n + \sqrt{\frac{p_1+p_2}{n}}\right)\right\}
		\end{equation*}
		for any $t \geq 1$. Now the proof of this theorem is completed.
	\end{proof}

	\subsection{Proofs of Lemmas in Theorem \ref{thm:local-convergence-RKHS} }\label{sec:pr-main-lm}
	\begin{proof}[Proof of Lemma \ref{lm:one-step-constraction-functional}]
		Without loss of generality, we assume $l=1$. We focus on a particular step $t$ and assume $\langle a^{(t)}, a_1\rangle, \langle b^{(t)}, b_1 \rangle \geq 0$ as the signs of $a^{(t)}$ and $b^{(t)}$ have no essential effect on the algorithm and analysis. 
		
		Denote $\Delta(\cdot) = \tilde \xi^{(t+1)}(\cdot) - \lambda_1 \xi_1(\cdot)$. We first present several important inequalities that will be used throughout the analysis. Firstly, since $\|\tilde \xi^{(t+1)}\|_\cH \leq \lambda_{max}C_\xi$ and $\|\xi_1\|_\cH \leq C_\xi$, we have
		\begin{equation}\label{ineq:Delta-H-norm-bound}
			\left\|\Delta\right\|_\cH \leq \|\tilde\xi^{(t+1)}\|_\cH + \lambda_1\|\xi_1\|_\cH \leq 2\lambda_{max}C_\xi.
		\end{equation}
		In addition, Condition $(A_2)$ implies that for any $m \in [r]$,
		\begin{equation}\label{ineq:delta-xi-n-bound}
			\begin{split}
				\|\Delta\|_n & \overset{\eqref{ineq:A-2}}{\leq} \|\Delta\|_{\cL^2} + \sqrt{C_2\zeta_n \|\Delta\|_\cH\cdot\|\Delta\|_{\cL^2}} + \sqrt{C_2}\zeta_n \|\Delta\|_\cH \\
				& \overset{(a)}{\leq} 2\|\Delta\|_{\cL^2} + C\zeta_n\|\Delta\|_\cH;\\
				\|\xi_m\|_n & \overset{\eqref{ineq:A-2}}{\leq}\|\xi_m\|_{\cL^2} + \sqrt{C_2\zeta_n \|\xi_m\|_\cH\cdot\|\xi_m\|_{\cL^2}} + \sqrt{C_2}\zeta_n \|\xi_m\|_\cH \\
				& \overset{(b)}{\leq}  1 + \sqrt{C_2C_\xi\zeta_n } + \sqrt{C_2}C_\xi \zeta_n  \overset{(c)}{\leq}  3/2.
			\end{split}
		\end{equation}
		Here $(a)$ comes from the Arithmetic-Geometric mean inequality; $(b)$ is due to the assumption that $\|\xi_m\|_{\cL^2} = 1$ and $\|\xi_m\|_\cH \leq C_\xi$; and $(c)$ comes from the assumption that $\zeta_n < c$ for sufficiently small constant.
		
		Now we are ready for the proof. Since $\tilde \xi^{(t+1)}$ is the optimal solution of \eqref{ineq:alg-update-xi} and $\lambda_1 \xi_1$ is in the feasible set, we have
		\begin{equation*}
			\sum_{i=1}^{p_1}\sum_{j=1}^{p_2}\frac{1}{n}\sum_{k=1}^{n} \left(\widetilde\bcY_{ijk} - a_i^{(t+1)} b_j^{(t+1)} \tilde\xi^{(t+1)}(s_k)\right)^2 \leq \sum_{i=1}^{p_1}\sum_{j=1}^{p_2} \frac{1}{n} \sum_{k=1}^{n} \left(\widetilde\bcY_{ijk} - \lambda_1 a_i^{(t+1)} b_j^{(t+1)}\xi_1(s_k)\right)^2.
		\end{equation*}
		Since
		$$
		\widetilde \bcY_{ijk} = \sum_{m=1}^r \lambda_m (a_m)_i(b_m)_j\xi_m(s_k) + \bcZ_{ijs_k} + \varepsilon_{ijk},
		$$
		the above inequality is equivalent to:
		\begin{equation}\label{ineq:ls-reduced-form}
			\underbrace{\sum_{i=1}^{p_1}\sum_{j=1}^{p_2} \frac{1}{n} \sum_{k=1}^{n} (a_i^{(t+1)})^2 (b_j^{(t+1)})^2\Delta^2(s_k)}_{(i)} \leq \underbrace{2\sum_{i=1}^{p_1}\sum_{j=1}^{p_2} \frac{1}{n} \sum_{k=1}^{n} a_i^{(t+1)} b_j^{(t+1)}\Delta(s_k)\tilde\varepsilon_{ij}(s_k)}_{(ii)},
		\end{equation}
		where 
		\begin{equation}\label{eq:tilde-eps-decomp}
			\tilde\varepsilon_{ij}(s_k) := \bcZ_{ijs_k} + \varepsilon_{ijk} + \lambda_1 \left((a_1)_i(b_1)_j-a_i^{(t+1)}b_j^{(t+1)}\right)\xi_1(s_k) + \sum_{m=2}^r \lambda_m (a_m)_i (b_m)_j \xi_m(s_k).
		\end{equation}

		We start by providing a lower bound for $(i)$ in \eqref{ineq:ls-reduced-form}. By Condition $(A_2)$,
		\begin{equation*}
			\begin{split}
				\left\|\Delta\right\|_{\cL^2}^2 - \|\Delta\|_n^2 & \overset{\eqref{ineq:A-2}}{\leq} C_2\left(\zeta_n \left\|\Delta\right\|_\cH \cdot \left\|\Delta\right\|_{\cL^2} + \zeta_n^2 \|\Delta\|_\cH^2\right) \\
				& \overset{(a)}{\leq} C_2\left(\frac{1}{4C_2}\left\|\Delta\right\|_{\cL^2}^2 +  (C_2+1)\zeta_n^2 \|\Delta\|_\cH^2\right) \\
				& \overset{\eqref{ineq:Delta-H-norm-bound}}{\leq} \frac{1}{4}\|\Delta\|_{\cL^2}^2 + C\zeta_n^2\lambda_{max}^2.
			\end{split}
		\end{equation*}
		Here $(a)$ comes from the Arithmetic-Geometric mean inequality. Therefore, we have
		\begin{equation}\label{ineq:bouna-a}
			(i) = \left\|\Delta\right\|_n^2  \geq \frac{3}{4}\left\|\Delta\right\|_{\cL^2}^2 - C\zeta_n^2\lambda_{\max}^2.
		\end{equation}
		
		Now we give an upper bound for $(ii)$. By \eqref{eq:tilde-eps-decomp}, we can further decompose $(ii) = (ii_1) + (ii_2) + (ii_3) + (ii_4)$, where
		\begin{equation*}
			\begin{split}
				(ii_1) &= 2\sum_{i=1}^{p_1}\sum_{j=1}^{p_2} \frac{1}{n} \sum_{k=1}^{n} a_i^{(t+1)} b_j^{(t+1)}\Delta(s_k)\bcZ_{ijs_k}, \\
				(ii_2) &= 2\sum_{i=1}^{p_1}\sum_{j=1}^{p_2} \frac{1}{n} \sum_{k=1}^{n} a_i^{(t+1)} b_j^{(t+1)}\Delta(s_k)\varepsilon_{ijk}, \\
				(ii_3) &= 2\sum_{i=1}^{p_1}\sum_{j=1}^{p_2} \frac{1}{n} \sum_{k=1}^{n} \lambda_1  a_i^{(t+1)} b_j^{(t+1)}\left((a_1)_i(b_1)_j - a_i^{(t+1)} b_j^{(t+1)}\right) \Delta(s_k)\xi_1(s_k), \\
				(ii_4) &= 2\sum_{i=1}^{p_1}\sum_{j=1}^{p_2} \frac{1}{n} \sum_{k=1}^{n} \sum_{m=2}^r \lambda_l a_i^{(t+1)} b_j^{(t+1)}(a_m)_i (b_m)_j \Delta(s_k)\xi_m(s_k). 
			\end{split}
		\end{equation*}
		We bound these four terms separately.
		\begin{itemize}
			\item $(ii_1)$: First, recall the definitions of $\cE$ and one immediately has
			\begin{equation*}
				\left\|\sum_{i=1}^{p_1}\sum_{j=1}^{p_2} a_i^{(t+1)} b_j^{(t+1)} \bcZ_{ij\cdot}\right\|_{\infty} \leq \mathcal E.
			\end{equation*} 
			Then it follows that
			\begin{equation}\label{ineq:bound-b1}
				\begin{split}
					\frac{1}{2} \cdot (ii_1) & \leq \left\|\sum_{i=1}^{p_1}\sum_{j=1}^{p_2} a_i^{(t+1)} b_j^{(t+1)} \bcZ_{ij\cdot}\right\|_\infty \cdot \left(\frac{1}{n}\sum_{k=1}^n \left|\Delta(s_k)\right|\right)  \\
					& \overset{(a)}{\leq} \cE \cdot \|\Delta\|_n \\
					& \overset{\eqref{ineq:A-2}}{\leq} \cE \cdot \sqrt{\|\Delta\|_{\cL^2}^2 + C_2\zeta_n \|\Delta\|_{\cL^2}\cdot\|\Delta\|_\cH + C_2\zeta_n^2 \|\Delta\|_\cH^2} \\
					& \overset{(b)}{\leq} \cE \cdot \sqrt{2\|\Delta\|_{\cL^2}^2 + (C_2^2/4 + C_2)\zeta_n^2 \|\Delta\|_\cH^2} \\
					& \overset{(c)}{\leq} \frac{1}{8} \|\Delta\|_{\cL}^2 + (C_2^2/64+C_2/16)\zeta_n^2\|\Delta\|_\cH^2 + 4\cE^2\\
					& \overset{\eqref{ineq:Delta-H-norm-bound}}{\leq} \frac{1}{8}\|\Delta\|_{\cL^2}^2 + C\left(\lambda_{max}^2 \zeta_n^2 + \cE^2\right).
				\end{split}
			\end{equation}
			Here (a) comes from the inequality between $l_1$ and $l_2$ norm, and (b), (c) come from the Arithmetic-Geometric mean inequality.
			\item $(ii_2)$: This term can be bounded by investigating the concentration of the functional tensor spectral norm under the RKHS constraint given by Condition $(A_3)$:
			\begin{equation}\label{ineq:bound-b2}
				\begin{split}
					\frac{1}{2} \cdot (ii_2) & \overset{\eqref{ineq:A-3}}{\leq} C\tau \left(\zeta_n^2 + \frac{(p_1+p_2)\log n}{n}\right) \|\Delta\|_\cH + C\tau\left(\zeta_n + \sqrt{\frac{p_1+p_2}{n}}\right) \|\Delta\|_{\cL^2}\\
					& \overset{\eqref{ineq:Delta-H-norm-bound}}{\leq} C\lambda_{max}\tau \left( \zeta_n^2 + \frac{(p_1+p_2)\log n}{n} \right) + C\tau\left(\zeta_n + \sqrt{\frac{p_1+p_2}{n}}\right) \|\Delta\|_{\cL^2} \\
					& \leq \frac{1}{8}\|\Delta\|_{\cL^2}^2 + C\tau^2\left(\zeta_n+\sqrt{\frac{p_1+p_2}{n}}\right)^2 + C\lambda_{max}^2\left(\zeta_n^2  + \frac{(p_1+p_2)\log n}{n}\right).
				\end{split}
			\end{equation}
			Here, we use the assumption $\tau\leq \lambda_{max}$ and the Arithmetic-Geometric mean inequality to obtain the final inequality.
			\item $(ii_3)$: Note that
			\begin{equation}\label{ineq:b3-decompose}
				\frac{1}{2}\cdot (ii_3) \leq \lambda_{max}\left|\sum_{i=1}^{p_1} \sum_{j=1}^{p_2} a_i^{(t+1)} b_j^{(t+1)} \left((a_1)_i(b_1)_j - a_i^{(t+1)} b_j^{(t+1)}\right)\right| \cdot \left|\frac{1}{n}\sum_{k=1}^n \Delta(s_k)\xi_1(s_k)\right|.
			\end{equation}
			On one hand, since $\langle a_1, a^{(t+1)}\rangle, \langle b_1,b^{(t+1)}\rangle \geq 0$, we can bound
			\begin{equation}\label{ineq:b3-1st-term}
				\begin{split}
					&\left|\sum_{i=1}^{p_1} \sum_{j=1}^{p_2} a_i^{(t+1)} b_j^{(t+1)} \left((a_1)_i(b_1)_j - a_i^{(t+1)} b_j^{(t+1)}\right)\right| \\
					& = \left|\left(\sum_{i=1}^n (a_1)_ia_i^{(t+1)}\right)\left(\sum_{j=1}^n (b_1)_jb_j^{(t+1)}\right) - \left(\sum_{i=1}^n (a_i^{(t+1)})^2\right)\left(\sum_{j=1}^n (b_j^{(t+1)})^2\right)\right| \\
					& = 1 - \langle a_1, a^{(t+1)} \rangle \langle b_1, b^{(t+1)} \rangle  \leq 1 - \left(1-(\delta^{(t+1)})^2\right) = (\delta^{(t+1)})^2.
				\end{split}
			\end{equation}
			On the other hand,
			\begin{equation}\label{ineq:b3-2nd-term}
				\begin{split}
					\left|\frac{1}{n}\sum_{k=1}^n\Delta(s_k)\xi_1(s_k)\right| & \leq \|\Delta\|_n \cdot \|\xi_1\|_n \\
					& \overset{\eqref{ineq:delta-xi-n-bound}}{\leq} \frac{3}{2}\left(2\|\Delta\|_{\cL^2} + C\zeta_n\|\Delta\|_\cH\right) \\
					& \overset{\eqref{ineq:Delta-H-norm-bound}}{\leq} 3\|\Delta\|_{\cL^2} + C\lambda_{max}\zeta_n.
				\end{split}
			\end{equation}
			Combining \eqref{ineq:b3-decompose}, \eqref{ineq:b3-1st-term} and \eqref{ineq:b3-2nd-term}, we obtain
			\begin{equation}\label{ineq:bound-b3}
				\begin{split}
					\frac{1}{2} \cdot (ii_3) & \leq (\delta^{(t+1)})^2\left(3\lambda_{max}\|\Delta\|_{\cL^2} + C\lambda_{max}^2\zeta_n\right) \\
					& \leq \frac{1}{8}\left\|\Delta\right\|_{\cL^2}^2 + C\left((\delta^{(t+1)})^4 \lambda_{max}^2 + (\delta^{(t+1)})^2\lambda_{max}^2\zeta_n\right). \\
					& \leq \frac{1}{8}\left\|\Delta\right\|_{\cL^2}^2 + \frac{c_1}{\kappa^2}(\delta^{(t+1)})^2\lambda_{max}^2.
				\end{split}
			\end{equation}
			Here the last inequality holds for some sufficiently small constant $c_1>0$ by the assumptions on $\delta^{(t)}$ and $\zeta_n$: 
			\begin{equation}\label{ineq:delta-zeta-upper}
				\delta^{(t+1)} \leq c/r\kappa, \qquad\zeta_n \leq c/(\kappa^2r).
			\end{equation}
			
			\item $(ii_4)$: For any $m = 2,\ldots,r$,
			\begin{equation*}
				\begin{split}
					\sum_{i=1}^{p_1}\sum_{j=1}^{p_2} a_i^{(t+1)} b_j^{(t+1)} (a_m)_i (b_m)_j & = \langle a^{(t+1)}, a_m \rangle \cdot \langle b^{(t+1)}, b_m \rangle \\
					& \leq \left(|\langle a,a_m \rangle| + \langle a^{(t+1)}-a,a_m \rangle\right) \cdot \left(|\langle b,b_m \rangle| + \langle b^{(t+1)}-b,b_m \rangle\right) \\
					& \leq \left(|\langle a,a_m \rangle| + \|a^{(t+1)} - a\|_2\right) \cdot \left(|\langle b,b_m \rangle| + \|b^{(t+1)} - b\|_2\right) \\
					& \leq C(\mu + \delta^{(t+1)})^2.
				\end{split}
			\end{equation*}
			Here the last inequality comes from the incoherence assumption and the fact that
			\begin{equation*}
				\left\|u-v\right\|_2 \leq 2\sqrt{1 - \langle u, v \rangle^2} = 2\dist(u,v)
			\end{equation*}
			for any two unit vectors $u,v$ with $\langle u, v\rangle \geq 0$. Then it follows by the similar argument as \eqref{ineq:bound-b3} that
			\begin{equation}\label{ineq:bound-b4}
				\begin{split}
					\frac{1}{2} \cdot (ii_4) &\leq C(\mu+\delta^{(t+1)})^2 (r-1)\left(\lambda_{max}\|\Delta\|_{\cL^2} + \lambda_{max}^2\zeta_n\right) \\
					& \leq \frac{1}{8}\left\|\Delta\right\|_{\cL^2}^2 + C\left((\mu+\delta^{(t+1)})^4 (r-1)^2\lambda_{\max}^2 + (\mu+\delta^{(t+1)})^2(r-1)\lambda_{max}^2\zeta_n\right) \\
    & \leq  \frac{1}{8}\left\|\Delta\right\|_{\cL^2}^2 + C\left(\mu^4 + (\delta^{(t+1)})^4\right)(r-1)^2\lambda_{\max}^2 + C\left(\mu^2 + (\delta^{(t+1)})^2\right)(r-1)\lambda_{\max}^2 \zeta_n \\
    & \overset{\eqref{ineq:delta-zeta-upper}}{\leq} \frac{1}{8}\left\|\Delta\right\|_{\cL^2}^2 + C\mu^4(r-1)^2\lambda_{\max}^2 + C\mu^2(r-1)\lambda_{\max}^2 \frac{c}{\kappa^2 r} \\
    & \qquad + C(\delta^{(t+1)})^2(r-1)^2\lambda_{\max}^2\left(\frac{c}{r\kappa}\right)^2 + C(\delta^{(t+1)})^2(r-1)\lambda_{\max}^2 \frac{c}{r\kappa^2}\\
					& \leq \frac{1}{8}\left\|\Delta\right\|_{\cL^2}^2 + \frac{c_1}{\kappa^2}(\delta^{(t+1)})^2\lambda_{max}^2 + C\mu^2(r-1)^2\lambda_{max}^2.
				\end{split}
			\end{equation}
   In the last inequality, we used the property that $\mu\leq 1, \kappa\geq 1, r\geq 1$ by its definition.
		\end{itemize}
		Combining \eqref{ineq:bouna-a} with \eqref{ineq:bound-b1}, \eqref{ineq:bound-b2}, \eqref{ineq:bound-b3} and \eqref{ineq:bound-b4}, we obtain
		\begin{equation*}
			\begin{split}
				\|\Delta\|_{\cL^2}^2 & \leq C\left(\lambda_{max}^2\left(\zeta_n^2 + \frac{(p_1+p_2)\log n}{n}\right) + \cE^2 + \tau^2\left(\zeta_n + \sqrt{\frac{p_1+p_2}{n}}\right)^2\right) \\
				& \quad + \frac{2c_1}{\kappa^2}(\delta^{(t+1)})^2\lambda_{max}^2 + C\mu^2(r-1)^2\lambda_{max}^2.
			\end{split}
		\end{equation*}
		Taking the square root, and further applying \eqref{ineq:delta-xi-n-bound}, we obtain that
		\begin{equation}\label{ineq:Delta-L2-upper-bound}
			\begin{split}
				\|\Delta\|_{\cL^2} & \leq C\left(\lambda_{max}\sqrt{\zeta_n^2 + \frac{(p_1+p_2)\log n}{n}} +  \cE + \tau\left(\zeta_n + \sqrt{\frac{p_1+p_2}{n}}\right)\right) \\
				& \quad + \frac{c}{\kappa}\delta^{(t+1)}\lambda_{max} + C\mu(r-1)\lambda_{max},
			\end{split}
		\end{equation}
		\begin{equation}\label{ineq:Delta-LN-upper-bound}
			\begin{split}
				\|\Delta\|_{n} & \leq 2\|\Delta\|_{\cL^2} + C\zeta_n \|\Delta\|_\cH \\
				& \leq C\left( \lambda_{max} \sqrt{\zeta_n^2 + \frac{(p_1+p_2)\log n}{n}} + \cE + \tau\left(\zeta_n + \sqrt{\frac{p_1+p_2}{n}}\right)\right) \\
				& \quad + \frac{c}{\kappa}\delta^{(t+1)}\lambda_{max} + C\mu(r-1)\lambda_{max}.
			\end{split}
		\end{equation}
		Combining \eqref{ineq:delta-zeta-upper} and the other assumptions on SNR and incoherence, i.e.,
		\begin{equation*}
			\frac{n}{\log^2 n} \geq \kappa(p_1+p_2),\qquad
			\mu \leq \frac{c}{\kappa^2 r^2}, \qquad \lambda_{min} \geq C\kappa r\left\{\cE + \tau\left(\zeta_n + \sqrt{(p_1+p_2)/n}\right)\right\},
		\end{equation*}
		we further conclude that
		\begin{equation*}
			\max\left\{\left\|\Delta\right\|_{\cL^2}, \left\|\Delta\right\|_n\right\} \leq \frac{1}{2}\lambda_{min}. 
		\end{equation*}
		Then it follows that
		\begin{equation*}
			\begin{split}
				\left\|\tilde\xi^{(t+1)}\right\|_{\cL^2} & \geq \|\lambda_1\xi_1\|_{\cL^2} - \|\Delta\|_{\cL^2} \geq \lambda_{min} - \frac{1}{2}\lambda_{min} = \frac{1}{2}\lambda_{min}. \\ 
				\left\|\tilde\xi^{(t+1)}\right\|_{n} & \geq \|\lambda_1\xi_1\|_{n} - \|\Delta\|_{n} \\
				& \overset{\eqref{ineq:A-2}}{\geq} \lambda_{min}\left(\|\xi_1\|_{\cL^2} - \sqrt{C_2\zeta_n\|\xi_1\|_{\cL^2}\|\xi_1\|_\cH} - \zeta_n\|\xi_1\|_\cH\right) - \frac{1}{2}\lambda_{min} \\
				& \geq \frac{3}{4}\lambda_{min} - \frac{1}{2}\lambda_{min} = \frac{1}{4}\lambda_{min}.
			\end{split}
		\end{equation*}
		
		Since $\xi^{(t+1)} = \tilde \xi^{(t+1)} / \|\tilde \xi^{(t+1)}\|_{\cL^2}$ and $\xi_n^{(t+1)} = \tilde\xi_n^{(t+1)} / \|\tilde \xi_n^{(t+1)}\|$, we obtain
		\begin{equation}\label{ineq:dist-xit-bound}
			\begin{split}
				\dist(\xi^{(t+1)},\xi_1) \overset{\text{Lemma \ref{lm:projection}}}{\leq} \frac{\|\tilde \xi^{(t+1)} - \lambda_1\xi_1\|_{\mathcal{L}^2}}{\|\tilde\xi^{(t+1)}\|_{\cL^2}} \leq \frac{\|\Delta\|_{\cL^2}}{\lambda_{min}}, \\
				\dist(\xi^{(t+1)}_n,(\xi_1)_n) \overset{\text{Lemma \ref{lm:projection}}}{\leq} \frac{\|\tilde \xi_n^{(t+1)} - \lambda_1(\xi_1)_n\|_n}{\|\tilde\xi^{(t+1)}\|_n} \leq \frac{\|\Delta\|_n}{\lambda_{min}}.
			\end{split}
		\end{equation}
  
  Now the one-step error contraction is proved by combining \eqref{ineq:dist-xit-bound} with \eqref{ineq:Delta-L2-upper-bound} and \eqref{ineq:Delta-LN-upper-bound}.

		In addition, we also get that
		\begin{equation*}
			\left\|\xi^{(t+1)}\right\|_\cH = \frac{\left\|\tilde\xi^{(t+1)}\right\|_\cH}{\left\|\tilde\xi^{(t+1)}\right\|_{\cL^2}} \leq \frac{C_\xi\lambda_{max}}{\frac{1}{2}\lambda_{min}} = 2\kappa C_\xi.
		\end{equation*}
	\end{proof}
	\bigskip
	\begin{proof}[Proof of Lemma \ref{lm:one-step-constraction-tabular}]
		We follow the convention in the proof of Lemma \ref{lm:one-step-constraction-functional} by assuming $l=1$ and $\langle b_1, b^{(t)} \rangle, \langle \xi_1, \xi^{(t)} \rangle_{\cL^2} \geq 0$. By definition, $a^{(t+1)} = \tilde a^{(t+1)} /\|\tilde a^{(t+1)}\|_2$, where 
		\begin{equation*}
			\begin{split}
				(\tilde{a}^{(t+1)})_i & = \sum_{j=1}^{p_2}\sum_{k=1}^n (b^{(t)})_j \xi^{(t)}(s_k)\widetilde\bcY_{ijk} \\
				& = \sum_{j=1}^{p_2}\sum_{k=1}^n (b^{(t)})_j \xi^{(t)}(s_k)\left(\lambda_1 (a_1)_i(b_1)_j\xi_1(s_k) + \sum_{m=2}^r \lambda_m(a_m)_i(b_m)_j\xi_m(s_k) + \bcZ_{ijs_k} + \varepsilon_{ijk} \right).
			\end{split}
		\end{equation*}			
		Denote 
		\begin{equation*}
			\bar a^{(t+1)} = \lambda_1\left(\sum_{j=1}^{p_2} \sum_{k=1}^n  (b^{(t)})_j (b_1)_j \xi^{(t)}(s_k) \xi_1(s_k)\right) a_1 = \left(\lambda_1\langle b_1, b^{(t)}\rangle \langle  (\xi_1)_n,  (\xi^{(t)})_n\rangle\right) a_1.
		\end{equation*}
		Then it follows that
		\begin{equation}\label{ineq:a-bar-tilde}
			\begin{split}
				\left\|\tilde a^{(t+1)} - \bar a^{(t+1)}\right\|_2 &\leq  \left\|\sum_{m=2}^r\lambda_m \langle b_m, b^{(t)} \rangle \left\langle (\xi_m)_n, (\xi^{(t)})_n \right\rangle_{\cL^2}  a_m\right\|_2 + \left\|\sum_{j=1}^{p_2}\sum_{k=1}^n (b^{(t)})_j \xi^{(t)}(s_k)\bcZ_{\cdot js_k}\right\|_2 \\
				& \qquad + \left\|\sum_{j=1}^{p_2}\sum_{k=1}^n (b^{(t)})_j \xi^{(t)}(s_k) \varepsilon_{\cdot jk}\right\|_2.
			\end{split}
		\end{equation}
		We bound the three terms in the right-hand-side of \eqref{ineq:a-bar-tilde} separately.
		\begin{itemize}
			\item $\left\|\sum_{m=2}^r\lambda_m \langle b_m, b^{(t)} \rangle \left\langle (\xi_m)_n, (\xi^{(t)})_n \right\rangle_{\cL^2}  a_m\right\|_2$. By condition \eqref{ineq:A-1}, for any $m, m' = 1,\ldots,r$, $m\neq m'$, we have
			\begin{equation*}
				\begin{split}
					\left|\frac{1}{n}\langle (\xi_m)_n, (\xi^{(t)})_n \rangle - \langle \xi_m, \xi^{(t)} \rangle_{\cL^2}\right| & \leq  C\left(\zeta_n \left\|\xi_m\xi^{(t)}\right\|_{\cL^2} + \zeta_n^2 \left\|\xi_m \xi^{(t)}\right\|_\cH\right) \leq C\zeta_n, \\
					\left|\frac{1}{n}\langle (\xi_m)_n, (\xi_{m'})_n \rangle - \langle \xi_m, \xi_{m'} \rangle_{\cL^2}\right| & \leq  C\left(\zeta_n \left\|\xi_m \xi_{m'}\right\|_{\cL^2} + \zeta_n^2 \left\|\xi_m \xi_{m'}\right\|_\cH\right) \leq C\zeta_n.
				\end{split}
			\end{equation*}
			Here we use the assumptions that $\|\xi_m\|_{\cL^2} = \|\xi^{(t)}\|_{\cL^2} = 1$, $\|\xi_m\|_{\cH} \leq C_\xi, \|\xi^{(t)}\|_\cH \leq 2\kappa C_{\xi}'$ and the condition $\zeta_n \leq c/\kappa$. Then it follows that
			\begin{equation}\label{ineq:xi-N-incoherence}
				\begin{split}
					\max_{m\neq m'} \frac{1}{n}\left|\langle (\xi_m)_n, (\xi_{m'})_n \rangle \right|& \leq \max_{l\neq l'} \left|\langle \xi_l, \xi_{l'} \rangle_{\cL^2}\right| + C\zeta_n \leq \mu + C\zeta_n, \\
					\frac{1}{n}\langle (\xi_1)_n,  (\xi^{(t)})_n \rangle & \geq \langle\xi_1, \xi^{(t)}\rangle_{\cL^2} - C\zeta_n \geq \sqrt{1-(\delta^{(t)})^2} - C\zeta_n \geq 1/2,
				\end{split}
			\end{equation}
			where the last inequality comes from the assumptions on $\delta^{(t)}$ and $\zeta_n$.
			Therefore,
			\begin{equation*}
				\begin{split}
					\frac{1}{n}\left|\langle (\xi_m)_n, (\xi^{(t)})_n \rangle\right| & \leq \langle \xi_m, \xi^{(t)}\rangle_{\cL^2} + C\zeta_n  \leq \left|\langle \xi_m , \xi_1 \rangle_{\cL^2}\right| + \left\|\xi_1 - \xi^{(t)}\right\|_{\cL^2} + C\zeta_n \\
					& \overset{\eqref{ineq:xi-N-incoherence}}{\leq} \mu + 2\delta^{(t)} + C\zeta_n;
				\end{split}
			\end{equation*}
			similarly, one can prove 
			\begin{equation*}
				\langle b_m , b^{(t)}\rangle \leq \mu + 2\delta^{(t)}.
			\end{equation*}
			Note that here we do not have the term $\zeta_n$ since $b^{(t)}$ corresponds to a tabular mode and has error term $\zeta_n$ from discretization.
			
			Combining the above results, we have
			\begin{equation}\label{ineq:a-update-bound-1}
				\begin{split}
					& \left\|\sum_{m=2}^r\lambda_m \langle b_m, b^{(t)} \rangle \cdot \left\langle (\xi_m)_n, (\xi^{(t)})_n \right\rangle  a_m\right\|_2^2 \\
					\leq & \sum_{m=2}^r \lambda_m^2 \langle b_m, b^{(t)} \rangle^2 \left\langle (\xi_m)_n, (\xi^{(t)})_n \right\rangle^2 \\
					& \qquad + \sum_{\substack{s,s' \in [r]/\{1\} \\ m \neq m'}}\lambda_m\lambda_{m'} \langle b_m, b^{(t)}\rangle \langle b_{m'}, b^{(t)} \rangle \left\langle (\xi_m)_n , (\xi^{(t)})_n\right\rangle \left\langle (\xi_{m'})_n, (\xi^{(t)})_n\right\rangle \langle a_m, a_{m'}\rangle \\
					& \leq Cn^2(r-1)\lambda_{max}^2 (\mu+\delta)^2(\mu+\delta+\zeta_n)^2 \\
					& \qquad + Cn^2(r-1)^2\lambda_{max}^2\mu^2(\mu+\delta)^2(\mu+\delta+\zeta_n)^2 \\
					& \leq Cn^2(r-1)\lambda_{max}^2 (\mu+\delta)^2(\mu+\delta+\zeta_n)^2.
				\end{split}
			\end{equation}
			Here we use the assumption $\mu \leq r^{-1/2}$ to obtain the final inequality.

			\item $\left\|\sum_{j=1}^{p_2}\sum_{k=1}^n (b^{(t)})_j \xi^{(t)}(s_k)\bcZ_{\cdot js_k}\right\|_2$. First, let $ a^* \in \bbS^{p_1-1}$ such that
			\begin{equation*}
				\frac{1}{n}\sum_{i=1}^{p_1}\sum_{j=1}^{p_2}\sum_{k=1}^n a^*_i (b^{(t)})_j \xi^{(t)}(s_k)\bcZ_{ijs_k} \leq  \left\|\frac{1}{n}\sum_{j=1}^{p_2}\sum_{k=1}^n (b_1^{(t)})_j \xi^{(t)}_1(s_k)\bcZ_{\cdot js_k}\right\|_2.
			\end{equation*}	 
			Then following the same argument as \eqref{ineq:bound-b1}, we have
			\begin{equation*}
				\begin{split}
					& \frac{1}{n}\sum_{i=1}^{p_1}\sum_{j=1}^{p_2}\sum_{k=1}^n a^*_i (b^{(t)})_j \xi^{(t)}(s_k)\bcZ_{ijs_k} \\
					\leq & \left\|\sum_{i=1}^{p_1}\sum_{j=1}^{p_2} a^*_i (b^{(t)})_j\bcZ_{ijs_k} \right\|_\infty \cdot \left(\frac{1}{n}\sum_{k=1}^n \left|\xi^{(t)}(s_k)\right|\right) \\
					\leq & \cE \cdot \left\|\xi^{(t)}\right\|_n \\
					\overset{\eqref{ineq:A-2}}{\leq} &  \cE \cdot \sqrt{\|\xi^{(t)}\|_{\cL^2}^2 + C_2\zeta_n \|\xi^{(t)}\|_{\cL^2}\cdot\|\xi^{(t)}\|_\cH + C_2\zeta_n^2 \|\xi^{(t)}\|_\cH^2} \\
					\leq & \cE \cdot \sqrt{1+C_2\zeta_n \|\xi^{(t)}\|_\cH + C_2 \zeta_n^2 \|\xi^{(t)}\|_\cH^2} \leq 2\cE.
				\end{split}
			\end{equation*}
			Note that the last inequality comes from the assumption that $\|\xi_1^{(t)}\|_\cH \leq 2\kappa C_{\xi}$ and $\zeta_n \leq c/\kappa$ for a sufficiently small constant $c$.
			
			Therefore,
			\begin{equation}\label{ineq:a-update-bound-2}
				\left\|\sum_{j=1}^{p_2}\sum_{k=1}^n (b^{(t)})_j \xi^{(t)}(s_k)\bcZ_{\cdot js_k}\right\|_2 \leq 2n\cE.
			\end{equation}
			\item $\left\|\sum_{j=1}^{p_2}\sum_{k=1}^n (b^{(t)})_j \xi^{(t)}(s_k) \varepsilon_{\cdot jk}\right\|_2$. By Condition $(A_3)$, we have
			\begin{equation}\label{ineq:a-update-bound-3}
				\begin{split}
					& \left\|\sum_{j=1}^{p_2}\sum_{k=1}^n (b^{(t)})_j \xi^{(t)}(s_k) \varepsilon_{\cdot jk}\right\|_2  \leq \sup_{a \in \bbS^{p_1-1}, b \in \bbS^{p_2-1}} \frac{1}{n}\sum_{i=1}^{p_1}\sum_{j=1}^{p_2}\sum_{k=1}^n a_ib_j\xi^{(t)}(s_k)\varepsilon_{ijk}\\
					\overset{\eqref{ineq:A-3}}{\leq} & Cn\tau\left\{\left(\zeta_n + \sqrt{\frac{p_1+p_2}{n}}\right)\|\xi^{(t)}\|_{\cL^2} + \left(\zeta_n^2 + \frac{(p_1+p_2)\log n}{n}\right)\|\xi^{(t)}\|_{\cH}\right\} \\
					\leq & Cn\tau\left(\zeta_n + \sqrt{\frac{p_1+p_2}{n}}\right).
				\end{split}
			\end{equation}
		\end{itemize}
		Combining \eqref{ineq:a-bar-tilde}, \eqref{ineq:a-update-bound-1}, \eqref{ineq:a-update-bound-2} and \eqref{ineq:a-update-bound-3}, we obtain
		\begin{equation}
			\left\|\tilde a^{(t+1)} - \bar a^{(t+1)}\right\|_2 \leq Cn\left\{\sqrt{r-1}\lambda_{max}(\mu+\delta)(\mu+\delta+\zeta_n) + \cE + \tau\left(\zeta_n +\sqrt{\frac{p_1+p_2}{n}}\right) \right\}.
		\end{equation}
		Also note that 
		\begin{equation*}
			\left\|\bar a_1^{(t+1)}\right\|_2  = \lambda_1 \left|\langle b_1, b^{(t)}_1 \rangle\right| \cdot \left|\langle (\xi_1)_n,(\xi_1^{(t)})_n \rangle\right| \overset{\eqref{ineq:xi-N-incoherence}}{\geq} \frac{n\lambda_{min}}{4}.
		\end{equation*}
		Therefore, by Lemma \ref{lm:hilbert-rescale}, we obtain
		\begin{equation*}
			\dist(a_1, a_1^{(t+1)}) \leq C\sqrt{r-1}\kappa(\mu+\delta)(\mu+\delta+\zeta_n) + C\left\{\frac{\cE}{\lambda_{min}}  + \frac{\tau}{\lambda_{min}}\left(\zeta_n + \sqrt{\frac{p_1+p_2}{n}}\right)\right\}. 
		\end{equation*}
	\end{proof}

	\subsection{Additional Proofs}\label{sec:additional-proof}
	\begin{proof}[Proof of Theorem \ref{thm:initialization}]
		We only need to prove the bound of $\dist(a,a^{(0)})$ as the proof for the other tabular mode essentially follows. Denote $\widetilde \bY:= \cM_1(\widetilde\bcY)$ and $\bZ := \cM_1(\bcZ)$. Then we can write
		\begin{equation*}
			\widetilde \bY = \lambda a (b\otimes \xi_n)^\top + \bZ + \bZ_{\varepsilon},
		\end{equation*}	
		where $\bZ_{\varepsilon} \in \bbR^{p_1 \times p_2n}$ and has i.i.d. $N(0,\tau^2)$ entries. 
		
		Note that $a^{(0)}$ is the leading eigenvector of $\widetilde\bY\widetilde\bY^\top$, which can be decomposed as:
		\begin{equation*}
			\begin{split}
				\widetilde\bY\widetilde\bY^\top &= \lambda^2\cdot \|\xi_n\|_2^2 \cdot aa^\top +  \bZ\bZ^\top + \bZ_\varepsilon \bZ_\varepsilon^\top + \bZ\bZ_\varepsilon^\top + \bZ_\varepsilon \bZ^\top  \\
				& \qquad  + \lambda a(b \otimes \xi_n)^\top (\bZ + \bZ_\varepsilon)^\top +  \lambda  (\bZ + \bZ_\varepsilon)(b \otimes \xi_n)a^\top.
			\end{split}
		\end{equation*}
		Then,
		\begin{equation}\label{ineq:gram-decompose}
			\begin{split}
				& \left\|\left(\widetilde \bY\widetilde \bY^\top - \bbE\left[\bZ\bZ^\top\bigg|\{s_k\}_{k=1}^n\right] - \bbE \bZ_\varepsilon\bZ_\varepsilon^\top\right) - \lambda^2\|\xi_n\|^2 \cdot aa^\top\right\| \\
				\leq  &   \left\|\bZ_\varepsilon\bZ_\varepsilon^\top - \bbE\bZ_\varepsilon\bZ_\varepsilon^\top\right\| +	2\lambda \left\|a(b\otimes\xi_n)\bZ_{\varepsilon}^\top\right\| \\
				& + \left\|\bZ\bZ^\top - \bbE\left[\bZ\bZ^\top\bigg|\{s_k\}_{k=1}^n\right]\right\|  + 2\lambda \left\|a(b\otimes\xi_n)\bZ^\top\right\| + 2\left\|\bZ\bZ_\varepsilon^\top\right\|.
			\end{split}
		\end{equation}
		Before we proceed, we first provide the probabilistic upper and lower bounds for $\|\xi\|_n^2 = n^{-1}\|\xi_n\|_2^2$. Recall that $\|\xi\|_{\cL^2} = 1$ and $\|\xi\|_\cH \leq C_\xi$. By Lemma \ref{lm:MC-square-int}, with probability at least $1-n^{-9}$,
		\begin{equation}\label{ineq:xi-n-bound}
			\begin{split}
				\|\xi\|_n^2 & \leq 1 + C\zeta_n \leq \frac{3}{2}, \quad \|\xi\|_n^2 \geq 1 - C\zeta_n \geq \frac{1}{2}. \\
			\end{split}
		\end{equation}
		Now we bound the five terms in \eqref{ineq:gram-decompose} separately. First, by the tail bound of Wishart-type random matrix~(Lemma \ref{lm:hetero-wishart}), we have 
		\begin{equation*}
			\bbP\left(\left\|\bZ_\varepsilon\bZ_\varepsilon^\top - \bbE \bZ_\varepsilon\bZ_\varepsilon^\top\right\| \geq C\tau^2\left((\sqrt{p_1} + \sqrt{p_2n}+x)^2 - p_2n\right)\right) \leq e^{-x}.
		\end{equation*}
		Let $x = \sqrt{p_1}$. Then we have with probability at least $1-e^{-\sqrt{p_1}}$,
		\begin{equation}\label{ineq:Gram-bound-1}
			\left\|\bZ_\varepsilon\bZ_\varepsilon^\top - \bbE \bZ_\varepsilon\bZ_\varepsilon^\top\right\| \leq C\tau^2\left(p_1 + \sqrt{p_1p_2n}\right).
		\end{equation}
		Note that conditioning on $\{s_k\}_{k=1}^n$, $\left\|a(b\otimes\xi_n)^\top \bZ_{\varepsilon}^\top\right\| = \|(b\otimes \xi_n)^\top \bZ_{\varepsilon}^\top\|_2$ is the $l_2$ norm for a $p_1$-dimensional random vector with i.i.d. $N(0,\tau^2\|\xi_n\|^2)$ entries. Therefore, by Gaussian concentration, we have with probability at least $1-e^{-p_1}$
		\begin{equation}\label{ineq:Gram-bound-2}
			\left\|a(b\otimes \xi_n)^\top \bZ_{\varepsilon}^\top\right\| \leq C\tau\|\xi_n\|_2 \sqrt{p_1}  \overset{\eqref{ineq:xi-n-bound}}{\leq} C\tau \sqrt{np_1}.
		\end{equation}
		
		The analysis for $ \left\|\bZ\bZ^\top - \bbE\left[\bZ\bZ^\top\bigg|\{s_k\}_{k=1}^n\right]\right\|$ and $\|a(b\otimes\xi_n)\bZ^\top\|$ are more involved as the entries of $\bZ$ are dependent. We conduct our analysis conditioning on fixed values of $\{s_k\}_{k=1}^n$. Note that for any fixed pair $(i,j) \in [p_1] \times [p_2]$, 
		\begin{equation*}
			\left(\bcZ_{ijs_1}, \bcZ_{ijs_2}, \ldots, \bcZ_{ijs_n} \right)^\top \overset{i.i.d.}{\sim} N(0, \bSigma),
		\end{equation*}
		where $(\bSigma)_{ll} = \Var(\bcZ_{ijs_l}) \leq \sigma^2$ and it follows that
		\begin{equation}\label{ineq:discrete-cov-spectral-norm}
			\left\|\bSigma\right\| \leq \left\|\bSigma\right\|_* \leq n\sigma^2,
		\end{equation}
		where $\|\cdot\|$ and $\|\cdot\|_*$ are the matrix spectral and nuclear norms respectively.
		Let $\bSigma = \bU\bLambda\bU^\top$ be the eigenvalue decomposition of $\bSigma$ with $\bU \in \bbO_{n}$ and $\bLambda$ being a diagonal matrix. Let $\bW$ be a $p_1$-by-$(np_2)$ random matrix with i.i.d. $N(0,1)$ entries.  We can then equivalently represent 
		\begin{equation}\label{eq:Z-W-represent}
			\bZ = \bW \cdot \begin{bmatrix}
				\bLambda^{1/2}\bU^\top &  & & \\
				& \bLambda^{1/2}\bU^\top & & \\
				& & \vdots & \\
				& & & \bLambda^{1/2}\bU^\top
			\end{bmatrix},
		\end{equation}
		and it follows that
		\begin{equation}
			\bZ\bZ^\top = \tilde\bW\tilde\bW^\top
		\end{equation}
		for $\tilde\bW \in \bbR^{p_1 \times (np_2)}$ with independent mean-zero Gaussian random variables and $\Var(\bW_{ij}) = \bLambda_{k+1,k+1}$ iff $(j \text{ mod } n) = k$. Now applying the Concentration of
		Heteroskedastic Wishart-type Matrix~(Lemma \ref{lm:hetero-wishart}), we have
		\begin{equation*}
			\bbP\left(\left\|\tilde \bW\tilde \bW^\top - \bbE \tilde \bW\tilde \bW^\top\right\| \geq C\left((\sigma_R  + \sigma_C + \sigma_* \log p_1 + \sigma_* x)^2 - \sigma_R^2\right)\right) \leq e^{-x},
		\end{equation*}
		where $\sigma_R^2 = p_2 \sum_{k=1}^n \bLambda_{kk} = p_2 \|\bSigma\|_*^2$, $\sigma_C^2 = p_1\|\bSigma\|$ and $\sigma_* = \|\bSigma\|$. By specifying $x = \sqrt{p_1}$ and use the bounds \eqref{ineq:discrete-cov-spectral-norm}, we obtain that with probability at least $1-n^{-9}-e^{-\sqrt{p_1}}$,
		\begin{equation}\label{ineq:Gram-bound-3}
			\begin{split}
				\left\|\bZ\bZ^\top - \bbE\left[\bZ\bZ^\top\bigg|\{s_k\}_{k=1}^n\right]\right\|  &= \left\|\tilde\bW\tilde\bW^\top - \bbE \tilde\bW\tilde\bW^\top  \right\| \\
				& \leq p_1\|\bSigma\| + \sqrt{p_1p_2\|\bSigma\|\cdot\|\bSigma\|_*} \\
				& \leq 	n\sigma^2 \left(p_1 + \sqrt{p_1p_2}\right).
			\end{split}
		\end{equation}
		
		To bound $\left\|a(b\otimes \xi_n)\bZ^\top\right\|$, we also utilize the representation \eqref{eq:Z-W-represent} and it follows by the similar argument as \eqref{ineq:Gram-bound-2} that
		\begin{equation}\label{ineq:Gram-bound-4}
			\begin{split}
				\left\|a(b \otimes \xi_n)^\top\bZ^\top\right\| & \leq C\left\|(b \otimes \xi_n)^\top\diag(\bLambda^{1/2}\bU^\top,\bLambda^{1/2}\bU^\top,\ldots,\bLambda^{1/2}\bU^\top)\right\|_2\sqrt{p_1} \\
				& \leq C\|\xi_n\|_2 \cdot \|\Sigma\|^{1/2}  \cdot \sqrt{p_1} \overset{\eqref{ineq:xi-n-bound},\eqref{ineq:discrete-cov-spectral-norm}}{\leq} n\sqrt{p_1}.
			\end{split}
		\end{equation}
		Finally, note that $\bZ^\top \bZ_\varepsilon$ has the same distribution with $\bW\diag(\bLambda^{1/2},\ldots,\bLambda^{1/2})\bZ_\varepsilon^\top$, by Lemma \ref{lm:hetero-wishart}, we know that with probability at least $1-\exp(-cp_1)$, 
		\begin{equation}\label{ineq:Gram-bound-5}
			\left\|\bZ^\top \bZ_\varepsilon^\top\right\| \leq C\sigma\tau\left(\sqrt{p_1p_2\|\Sigma\|_*} + p_1\|\Sigma^{1/2}\|\right) \leq C\sigma\tau\left(\sqrt{np_1p_2} + \sqrt{n}p_1\right).
		\end{equation}
		Combining \eqref{ineq:gram-decompose} with \eqref{ineq:Gram-bound-1}, \eqref{ineq:Gram-bound-2}, \eqref{ineq:Gram-bound-3}, \eqref{ineq:Gram-bound-4} and \eqref{ineq:Gram-bound-5}, we obtain
		\begin{equation*}
			\begin{split}
				& \left\|\left(\hat \bY\hat \bY^\top - \bbE\left[\bZ\bZ^\top\bigg|\{s_k\}_{k=1}^n\right] - \bbE \bZ_\varepsilon\bZ_\varepsilon^\top\right) - \lambda^2\|\xi_n\|^2 \cdot aa^\top\right\| \\
				\lesssim & \tau^2(p_1+\sqrt{p_1p_2n}) + \lambda(\sqrt{n}\tau+n\sigma)  \sqrt{p_1} + (n\sigma^2 + \sqrt{n}\sigma\tau)\left(p_1 + \sqrt{p_1p_2}  \right). \\
			\end{split}
		\end{equation*}
		Since both $\bbE\left[\bZ\bZ^\top\bigg|\{s_k\}_{k=1}^n\right]$ and $\bbE \bZ_\varepsilon\bZ_\varepsilon^\top$ are some multipliers of the identity matrix, the leading eigenvector of $\left(\hat \bY\hat \bY^\top - \bbE\left[\bZ\bZ^\top\bigg|\{s_k\}_{k=1}^n\right] - \bbE \bZ_\varepsilon\bZ_\varepsilon^\top\right)$ is the same as that of $\hat\bY \hat\bY^\top$, i.e., $a^{(0)}$. 
		Then it follows by Wedin's Sin Theta's theorem~\citep{wedin1972perturbation} that
		\begin{equation*}
			\dist(a, a^{(0)}) \lesssim \frac{\tau^2(p_1+\sqrt{p_1p_2n}) + \lambda(\sqrt{n}\tau+n\sigma) \sqrt{p_1} + (n\sigma^2 + \sqrt{n}\sigma\tau)\left(p_1 + \sqrt{p_1p_2}  \right)}{\lambda^2\|\xi_n\|^2}.
		\end{equation*}
		Now the proof is completed by noticing $\|\xi_n\|^2 \geq n/2 $ and plugging in the signal-to-noise conditions
		\begin{equation*}
			\begin{split}
				\frac{\lambda}{\sigma} & \geq C\left(\sqrt{p_1} + \left(p_1p_2 \right)^{1/4}\right), \\
				\frac{\lambda}{\tau} & \geq C\left(\sqrt{\frac{p_1}{n}} + \left(\frac{p_1p_2}{n} \right)^{1/4}\right).
			\end{split}
		\end{equation*}
		
	\end{proof}

	\begin{proof}[Proof of Proposition \ref{prop:sampling-incoherence}]
		By the eigendecomposition of the covariance function $\bbG(s,t)$, we can equivalently write the Gaussian function $f_i(t) = \sum_{k=1}^\infty \mu_kX_{ik} \phi_k(t)$, where $X_{ik}\overset{i.i.d.}{\sim}N(0,1)$. Given the fact that $\sum_{k=1}^\infty \mu_k < \infty$, one can simply verify that 
		\begin{equation*}
			\left\|f_i\right\|_\cH = \sqrt{\sum_{k=1}^\infty\frac{\mu_k^2X_{ik}^2}{\mu_k}} = \sqrt{\sum_{k=1}^\infty \mu_kX_{ik}^2} < \infty,~~a.s.
		\end{equation*}
		Next we obtain an upper bound for $\langle f_i, f_j \rangle_{\cL^2}$ and a lower bound for $\|f_i\|_{\cL^2}$. Note that
		\begin{equation*}
			\langle f_1, f_2 \rangle_{\cL^2} = \sum_{k=1}^\infty \mu_k^2 X_{1k}X_{2k} := W.
		\end{equation*}
		Let $W_n := \sum_{k=1}^n \mu_k^2 X_{1k}X_{2k}$ be a sequence of random variables. Since $$\bbE\left|W_n - W\right| \leq \sum_{k=1}^n \mu_k^2 \bbE\left|X_{1k}\right|\cdot\bbE\left|X_{2k}\right| \rightarrow 0,\qquad \text{as } n\rightarrow \infty,$$
		we know that $W_n$ converges to $W$ in probability by Markov inequality. Then it follows by Hanson-wright inequality~ (see, e.g., Theorem 1.1 in \cite{rudelson2013hanson}),
		\begin{equation*}
			\bbP\left(\left|W_n\right| > t\right)\leq 2\exp\left(-c\min\left(\frac{t^2}{\sum_{k=1}^n\mu_k^4}, \frac{t}{\mu_1^2}\right)\right).
		\end{equation*}
		Setting $t = C\log r\cdot \sqrt{\sum_{k=1}^\infty\mu_k^4}$, we have
		\begin{equation*}
			\begin{split}
				\bbP\left(|W| > C\log r\cdot \sqrt{\sum_{k=1}^\infty\mu_k^4}\right)& = \lim_{n\rightarrow \infty}\bbP\left(\left|W_n\right| > C\log r\cdot \sqrt{\sum_{k=1}^\infty\mu_k^4}\right) \\
				& \leq 2\exp\left(-C\min\left(\log^2 r, \log r \sqrt{\sum_{k=1}^\infty \frac{\mu_k^4}{\mu_1^4}}\right)\right) \leq r^{-10}.
			\end{split}
		\end{equation*}
		Since $\langle f_i, f_j \rangle_{\cL^2} \overset{d}{=} W$ for any $i\neq j$, it follows by union bounds that
		\begin{equation}\label{ineq:fij-inner-upper}
			\max_{i \neq j} \langle f_i, f_j \rangle_{\cL^2} \leq C\log r\cdot \sqrt{\sum_{k=1}^\infty \mu_k^4}
		\end{equation}
		holds with probability at least $1-r(r-1)/2\cdot r^{-10} \geq 1-r^{-8}$.
		
		Then we provide the lower bound for $\|f_i\|_{\cL^2}^2$, which has the same distribution as $\sum_{k=1}^\infty \mu_k^2 X_{k}^2$ for $X_k \overset{i.i,d}{\sim} N(0,1)$. Let $Y := \sum_{k=1}^\infty \mu_k^2 (X_{k}^2-1)$ and consider the sequence $Y_n := \sum_{k=1}^n \mu_k^2 (X_{k}^2 - 1)$, one can see that $Y_n \overset{p}{\rightarrow} Y$ as $\sum_{k=1}^\infty \mu_k^2 < \infty$. Note that $Y_n$ is a weighted summation of independent centralized Chi-square distribution, by \cite[Theorem 6]{zhang2018non}, we have
		\begin{equation*}
			\bbP\left(Y_n \leq -t\right) \leq \exp\left(-\frac{t^2}{4\sum_{k=1}^\infty \mu_k^4}\right), \qquad \forall 0\leq t \leq c_0\frac{\sum_{k=1}^\infty \mu_k^4}{\mu_1^2},
		\end{equation*}
		where $c$ is some universal constant. Set $t = c_0\frac{\sum_{k=1}^\infty \mu_k^4}{\mu_1^2}$ and we get
		\begin{equation*}
			\begin{split}
				\bbP\left(\|f_i\|_{\cL^2}^2 \leq (1-c_0)\sum_{k=1}^\infty \mu_k^2\right) & =  \bbP\left(Y \leq -c\sum_{k=1}^\infty \mu_k^2\right) \\ 
				& \leq \bbP\left(Y \leq -c_0\frac{\sum_{k=1}^\infty \mu_k^4}{\mu_1^2}\right) \\
				& = \lim_{n\rightarrow \infty}  \bbP\left(Y_n \leq -c_0\frac{\sum_{k=1}^\infty \mu_k^4}{\mu_1^2}\right) \\
				& \leq \exp\left(-c\frac{\sum_{k=1}^\infty\mu_k^4}{\mu_1^4}\right) \leq \exp(-C\log r) \leq r^{-10}. 
			\end{split}
		\end{equation*} 
		Here the last inequality comes from the assumption $\log r \leq c\sum_{k=1}^\infty \mu_k^4/\mu_1^4$, by union bound, we get that with probability at least $1 - r^{-9}$ that 
		\begin{equation}\label{ineq:fi-norm-lower}
			\min_{i} \|f_i\|_\cL^2 \geq (1-c_0)\sum_{k=1}^\infty\mu_k^2.
		\end{equation}
		Combining \eqref{ineq:fij-inner-upper} and \eqref{ineq:fi-norm-lower}, we can see for any $i \neq j \in [r]$
		\begin{equation*}
			\frac{\langle f_i,f_j\rangle_{\cL^2}}{\left\|f_i\right\|_{\cL^2} \cdot \left\|f_j\right\|_{\cL^2}} \leq \frac{C\log r\cdot \sqrt{\sum_{k=1}^\infty \mu_k^4}}{(1-c_0)\sum_{k=1}^\infty \mu_k^2} \leq \frac{C\log r}{(1-c_0)\sqrt{\sum_{k=1}^\infty\mu_k^2/\mu_1^2}}
		\end{equation*}
		and the proof is complete.
	\end{proof}

	\begin{proof}[Proof of Proposition \ref{prop:Gaussian-process-xi}]
		To bound $\cE$, we apply an $\varepsilon$-net argument on the tabular modes (i.e., $a, b$) and then apply the Borell-TIS Inequality~\citep{adler2009random}.
		
		We first construct an $\varepsilon$-net $\{a^{(1)},\ldots,a^{(N_a)}\}$ of $\mathbb S^{p_1-1}$ such that 
		\begin{equation*}
			\sup_{a \in \mathbb S^{p_1-1}}\min_{l \in [N_a]}\|a - a^{(l)}\|_2 \leq \varepsilon
		\end{equation*}
		with $N_a \leq \left(3/\varepsilon\right)^{p_1}$ \citep[Corollary 4.2.13]{vershynin2018high}. An $\varepsilon$-net $\{b^{(1)},\ldots,b^{(N_b)}\}$ for $\mathbb S^{p_2-1}$ can be constructed similarly. We take $\varepsilon=1/4$ for the following calculation. For each fixed pair $(l_1,l_2) \in [N_a] \times [N_b]$, define 
		\begin{equation*}
			T_s^{(l_1,l_2)} = \sum_{i=1}^{p_1}\sum_{j=1}^{p_2} a^{(l_1)}_ib^{(l_2)}_j\bcZ_{ijs}.
		\end{equation*}
		By Assumption \ref{asmp:Z-Gaussian-H}, for each pair $(i,j) \in [N_a] \times [N_b]$, $\bcZ_{ij\cdot}$ are i.i.d. mean-zero Gaussian process. Therefore, $T_s^{(l_1,l_2)}$ has the same distribution as $\bcZ_{ij\cdot}$ and $\mathbb{E}\sup_s|T_s^{(l_1,l_2)}|= \mathbb{E}\sup_s|\bcZ_{ijs}| \leq m$, $\sup_s \mathbb{E}(T_s^{(l_1,l_2)})^2=\sup_s \mathbb{E}\bcZ_{ijs}^2\leq \sigma^2$. Applying Borell-TIS Inequality (Lemma \ref{lm:Borell-TIS}) yields 
		\begin{equation*}
			\bbP\left(\sup_s\left|T_s^{(l_1,l_2)}\right| \geq m + x\sigma\right) \leq 2\exp(-x^2/2).
		\end{equation*}
		Setting $x = 3\sqrt{p_1+p_2}$ and applying union bounds, we obtain
		\begin{equation}\label{ineq:borell-union}
			\begin{split}
				\bbP\left(\max_{l_1,l_2}\sup_s\left|T_s^{(l_1,l_2)}\right| \geq m + 3\sigma\sqrt{p_1+p_2}\right) & \leq 2N_aN_b\exp(-9(p_1+p_2)/2) \\
				& \leq 2\exp\left((\log(3/\varepsilon)-8)(p_1+p_2)\right) \\
				& \leq 2\exp(-(p_1+p_2)).
			\end{split}
		\end{equation}
		Note that by definition,
		\begin{equation*}
			\begin{split}
				\cE = \sup_s \|\bcZ_{\cdot\cdot s}\| = \sup_s \sup_{\substack{a \in \bbS^{p_1-1} \\ b \in \bbS^{p_2-1} }} \left|\sum_{i=1}^{p_1}\sum_{j=1}^{p_2} a_ib_j\bcZ_{ijs}\right|.
			\end{split}
		\end{equation*}
		Let $a^* \in \bbS^{p_1-1}, b^* \in \bbS^{p_2-1}$ such that
		\begin{equation*}
			\cE = \sup_s  \left|\sum_{i=1}^{p_1}\sum_{j=1}^{p_2} a_i^*b_j^*\bcZ_{ijs}\right|.
		\end{equation*}
		Then one can find some $a^{(l_1^*)}, b^{(l_2^*)}$ such that
		\begin{equation*}
			\|a^* - a^{(l_1^*)}\|_2 \leq \varepsilon,\qquad \|b^* - b^{(l_2^*)}\|_2 \leq \varepsilon.
		\end{equation*}
		Then it follows that
		\begin{equation*}
			\begin{split}
				\cE & = \sup_s  \left|\sum_{i=1}^{p_1}\sum_{j=1}^{p_2} \left(a_i^{(l_1^*)} + a_i^* - a_i^{(l_1^*)}\right)\left(b_j^{(l_2^*)} + b_j^* - b_j^{(l_2^*)}\right)\bcZ_{ijs}\right| \\
				& \leq \sup_s T_s^{(l_1^*,l_2^*)} + 2\varepsilon \cE \\
				& \leq \max_{l_1,l_2}\sup_s\left|T_s^{(l_1,l_2)}\right| + \frac{1}{2}\cE.
			\end{split}
		\end{equation*}
		Combining this with \eqref{ineq:borell-union}, we obtain that with probability at least $1-C\exp(-c(p_1+p_2))$,
		\begin{equation*}
			\cE \leq 2m + 6\sigma\sqrt{p_1+p_2}.
		\end{equation*}
	\end{proof}
	
	\begin{proof}[Proof of Proposition \ref{prop:optimization-solution}]
	    We assume the kernel matrix $\bK$ is invertible while the same results could be easily extended to non-invertible kernels.  According to the Representor Theorem~\citep{kimeldorf1971some}, the solution for \eqref{ineq:alg-update-xi} admits the form $\xi (\cdot) = \sum_{k=1}^n \beta_k \bbK(\cdot, s_k)$ for some vector $\beta \in \bbR^n$. Therefore, the loss function defined in \eqref{ineq:alg-update-xi} can be written as
	    \begin{equation*}
	        \begin{split}
	            & \left\|\cM_3\left(\widetilde \bcY\right) - \bK \beta \left(a^{(t)} \otimes b^{(t)}\right)^\top\right\|_\tF^2 \\
	            & \qquad = \tr\left((a^{(t)}\otimes b^{(t)})\beta^\top \bK^2 \beta \left(a^{(t)} \otimes b^{(t)}\right)^\top\right) - 2\tr\left(\left(a^{(t)} \otimes b^{(t)}\right) \beta^\top \bK \cM_3(\widetilde\bcY)\right) + \left\|\cM_3\left(\widetilde \bcY\right)\right\|_\tF^2 \\
	            & \qquad = \beta^\top \bK^2 \beta - 2\beta^\top \bK \tilde y^{(t)} + \left\|\cM_3\left(\widetilde \bcY\right)\right\|_\tF^2.
	        \end{split}
	    \end{equation*}
	    Also note that 
	    \begin{equation*}
			\left\|\xi\right\|_\cH^2 = \sum_{k_1,k_2 \in [n]} \beta_{k_1}\beta_{k_2} \left\langle \bbK(\cdot,s_{k_1}),\bbK(\cdot,s_{k_2}) \right\rangle_\cH = \sum_{k_1,k_2 \in [n]} \beta_{k_1}\beta_{k_2} \bbK(s_{k_1},s_{k_2}) = \beta^{\top} \bK \beta.
		\end{equation*}
		In conclusion, the original optimization \eqref{ineq:alg-update-xi} can be equivalently formalized as
		\begin{equation}\label{eq:beta-optimize}
		    \argmin_{\beta: \beta^\top \bK \beta \leq \alpha^2} \beta^\top \bK^2 \beta - 2\beta^\top \bK \tilde y^{(t)}.
		\end{equation}
		Here $\alpha = C_\xi \lambda_{max}$. Consider the corresponding Lagrangian:
		\begin{equation*}
		    L(\beta, \mu) = \beta^\top \bK^\top \beta - 2\beta^\top \bK \tilde y^{(t)} + \mu \left(\beta^\top \bK \beta - \alpha^2\right).
		\end{equation*}
		Since \eqref{eq:beta-optimize} is a convex optimization, there exist at least one local minima and all the local minimas must be global minima. We denote it as $\beta^*$. Then Karush–Kuhn–Tucker conditions implies that there exists some $\mu^*$ such that
		\begin{equation}\label{eq:KKTs}
		    \frac{\partial L}{\partial \beta}(\beta^*, \mu^*) = 0,\qquad \beta^{*\top} \bK \beta^* \leq \alpha^2, \qquad \mu^* \geq 0, \qquad \mu^*\left(\beta^{*\top} \bK \beta^* - \alpha\right) = 0.
		\end{equation}
	    By the first condition, we have
	    \begin{equation*}
	        2\bK^2 \beta^* - 2 \bK \tilde y^{(t)} + 2\mu^* \bK \beta = 0,
	    \end{equation*}
	    which implies
	    \begin{equation*}
	        \beta^* = \left(\bK + \mu^*\bI\right)^{-1} \tilde y^{(t)}.
	    \end{equation*}
	    Now consider the following two scenarios:
	    \begin{itemize}
	        \item $\tilde y^{(t)} \bK^{-1} \tilde y^{(t)} \leq \alpha^2$. In this case, by setting $\mu^* = 0$, conditions \eqref{eq:KKTs} are satisfied and we have $\beta^* = \left(\bK + \mu^*\bI\right)^{-1} \tilde y^{(t)}$.
	        \item $\tilde y^{(t)} \bK^{-1} \tilde y^{(t)} > \alpha^2$. One can see that $\mu^* = 0$ no longer meet the conditions \eqref{eq:KKTs}. Therefore, we must have $\mu^* > 0$ and $\beta^{*\top} \bK \beta^{*} - \alpha^2 = 0$. In other words, we have $\mu^*$ satisfying
	        \begin{equation*}
	            \tilde y^{(t)\top} \left(\bK + \mu \bI\right)^{-1} \bK \left(\bK + \mu \bI\right)^{-1} y^{(t)} = \alpha^2
	        \end{equation*}
	    \end{itemize}
	    Combining these two scenarios, we proved the proposition.
	    
	\end{proof}

	\section{Properties of \texorpdfstring{$Q_n(\delta)$}{Lg} and \texorpdfstring{$\zeta_n$}{Lg}}\label{sec:prop-zeta}
	Recall the definitions of $Q_n(\delta)$ and $\zeta_n$:
	\begin{equation*}
		\begin{split}
			Q_n(\delta) &= \frac{1}{\sqrt{n}}\left(\sum_{k=1}^\infty \mu_k \wedge \delta^2\right)^{1/2},\\
			\zeta_n &= \inf\left\{\zeta \geq \sqrt{\frac{\log n}{n}}: Q_n(\delta) \leq \zeta \delta + \zeta^2,\quad \forall \delta \in (0,1]\right\}.
		\end{split}
	\end{equation*}
	We discuss some important properties of $Q_n(\delta)$ and $\zeta_n$ in this section, which are used to establish several technical results.
	
	The function $Q_n(\delta)$ is firstly introduced by \cite{mendelson2002geometric}, which characterizes the geometry for the general statistical learning problem in RKHS $\cH$ via the corresponding eigenvalues of the kernel operators. It is closely related to both local Rademacher complexity and Gaussian complexity, as stated in the following proposition.
	\begin{Proposition}\label{prop:Qn-complexity}
		Let $\{\epsilon_k\}_{k=1}^n$ be i.i.d. Rademacher variables with $\bbP(\epsilon_k = 1) = \bbP(\epsilon_k = -1) = 1/2$; $\{\omega_k\}_{k=1}^n$ be i.i.d. Gaussian random variables; and $\{s_k\}_{k=1}^n$ be i.i.d. uniformly distributed random variables on $[0,1]$. Let $\cF(\delta) = \{f \in \bbR^{[0,1]}: \|f\|_\cH \leq 1, \|f\|_{\cL^2} \leq \delta\}$. Define the local Rademacher complexity and Gaussian complexity as:
		\begin{equation*}
			\begin{split}
				\mathcal R_n (\delta) & = \frac{1}{n}\bbE\sup_{f \in \cF(\delta)} \left|\sum_{k=1}^n \epsilon_k f(s_k)\right|, \\
				\mathcal G_n (\delta) & = \frac{1}{n}\bbE\sup_{f \in \cF(\delta)} \left|\sum_{k=1}^n \omega_k f(s_k)\right|.
			\end{split}
		\end{equation*}
		Then, there exist absolute constants $c_\cR,c_\cG$ and $C_\cR,C_\cG$ such that for every $\delta \geq 1/\sqrt{n}$, 
		\begin{equation*}
			\begin{split}
				c_\cR Q_n(\delta) & \leq \cR_n(\delta) \leq C_\cR Q_n(\delta), \\
				c_\cG Q_n(\delta) & \leq \cG_n(\delta) \leq C_\cG Q_n(\delta).
			\end{split}
		\end{equation*}
	\end{Proposition}
	\begin{proof}
		The first inequality comes from a direct implication of \cite[Theorem 41]{mendelson2002geometric} by taking the probability measure $\mu$ as $\text{Unif}(0,1)$. The second inequality can be similarly established and we present it here for completeness. Recall $\{(\mu_k,\phi_k)\}_{k=1}^\infty$ are the eigenvalue-eigenvector pairs of $\bbK$. For any $f \in \cF(\delta)$, we can write
		\begin{equation*}
			f = \sum_{k=1}^\infty \beta_k \sqrt{\mu_k} \phi_k
		\end{equation*} 
		with coefficients $\{\beta_k\}_{k=1}^\infty$ satisfying
		\begin{equation*}
			\|f\|_{\cH}^2 = \sum_{k=1}^\infty \beta_k^2 \leq 1, \quad \|f\|_{\cL^2}^2 = \sum_{k=1}^\infty \mu_k\beta_k^2 \leq \delta^2,
		\end{equation*}
		which implies that
		\begin{equation}\label{ineq:beta-series-bound}
			\sum_{k=1}^\infty \left(1 \vee \frac{\mu_k}{\delta^2} \right)\beta_k^2 \leq 2.
		\end{equation}
		Now we define the function set
		\begin{equation*}
			\cF'(\delta) = \left\{f = \sum_{k=1}^\infty \beta_k \sqrt{\mu_k} \phi_k,\quad \sum_{k=1}^\infty \left(1 \vee \frac{\mu_k}{\delta^2} \right)\beta_k^2 \leq 2\right\}.
		\end{equation*}
		By the reasoning above we know that $\cF(\delta) \subset \cF'(\delta)$, and it follows that
		\begin{equation*}
			\begin{split}
				\bbE \sup_{f \in \cF(\delta)}\left|\sum_{k=1}^n\omega_k f(s_k)\right|^2 & \leq \bbE\sup_{f \in \cF'(\delta)}\left|\sum_{k=1}^n\omega_k f(s_k)\right|^2 \\
				& = \bbE\sup_{\beta \text{ satisfies } \eqref{ineq:beta-series-bound} }\left|\sum_{i=1}^\infty\sum_{k=1}^n\omega_k \beta_i \sqrt{\mu_i}\phi_i(s_k)\right|^2 \\
				& = \bbE\sup_{\beta \text{ satisfies } \eqref{ineq:beta-series-bound} }\left|\sum_{i=1}^\infty\left(\beta_i\sqrt{1\vee \frac{\mu_i}{\delta^2}}\right)\left(\sqrt{\mu_i\wedge \delta^2}\sum_{k=1}^n\omega_k \phi_i(s_k)\right)\right|^2 \\
				& \leq 2\bbE \sum_{i=1}^\infty(\mu_i \wedge \delta^2) \left(\sum_{k=1}^n \omega_k \phi_i(s_k)\right)^2 \\ 
				& = 2 \sum_{i=1}^\infty (\mu_i\wedge\delta^2)\bbE_s\sum_{k=1}^n \phi_i^2(s_k) = 2n^2Q_n^2(\delta).
			\end{split}
		\end{equation*}
		Then, by Jensen's inequality, we have
		\begin{equation*}
			\cG_n(\delta) \leq \frac{1}{n}\left(\bbE \sup_{f \in \cF(\delta)}\left|\sum_{k=1}^n\omega_k f(s_k)\right|^2\right)^{1/2} \leq \sqrt{2}Q_n(\delta),
		\end{equation*}
		which proves one side of the second inequality in the statement;
		the other side can be obtained by utilizing the natural inequality that $\cR_n(\delta) \leq \sqrt{\pi/2} \cG_n(\delta)$~\citep[Lemma 4]{bartlett2002rademacher} and the inequality between $\cR_n(\delta)$ and $Q_n(\delta)$.
		
	\end{proof}
	
	Next, we provide an explicit form of $\zeta_n$ for the two special RKHSs studied in Section \ref{sec:theory-init}.
	\begin{Proposition}\label{prop:zeta-rate}
		Suppose $\cH$ is finite-dimensional such that $\sum_{k=1}^\infty 1_{\{\mu_k > 0\}} = d$ for some $\log n \leq d \leq n$,  then we have $\zeta_n = \sqrt{d/n}$; suppose $\cH = W^{\alpha,2}$ for $\alpha > 1/2$, then we have $\zeta_n \leq Cn^{-\alpha/(2\alpha+1)}$.
	\end{Proposition}
	\begin{proof}
		Note that similar results can be found in \citep{koltchinskii2010sparsity,raskutti2012minimax}, and we present it here for completeness.
		If $\cH$ is finite-dimensional, we have
		\begin{equation*}
			Q_n(\delta) = \frac{1}{\sqrt{n}}\left(\sum_{k=1}^d \mu_k \vee \delta^2 \right)^{1/2} \leq \sqrt{\frac{d}{n}}\delta.
		\end{equation*}
		Specifying $\zeta_n = \sqrt{d/n}$ naturally yields $Q_n(\delta) \leq \zeta_n\delta + \zeta_n^2$ for any $\delta \in (0,1]$.
		
		Now suppose $\cH = W^{\alpha,2}$ and we have $\mu_k \leq Ck^{-2\alpha}$. Then one can bound
		\begin{equation*}
			\begin{split}
				Q_n(\delta) & \leq \frac{C}{\sqrt{n}} \left(\sum_{k=1}^{\lceil n^{1/(2\alpha+1)}\rceil}(k^{-2\alpha} \wedge \delta^2) + \sum_{k=\lceil n^{1/(2\alpha+1)}\rceil+1}^\infty (k^{-2\alpha} \wedge \delta^2) \right)^{1/2} \\
				& \leq C\left(\sqrt{\frac{\lceil n^{1/(2\alpha+1)}\rceil}{\sqrt{n}}} \delta + \sqrt{\frac{\sum_{k=\lceil n^{1/(2\alpha+1)}\rceil+1}^\infty k^{-2\alpha}}{n}}\right) \\
				& \leq C\left(n^{-\frac{\alpha}{2\alpha+1}}\delta + \sqrt{\frac{\left(n^{1/(2\alpha+1)}\right)^{1-2\alpha}}{n}}\right) \\
				& = C\left(n^{-\frac{\alpha}{2\alpha+1}}\delta + n^{-\frac{2\alpha}{2\alpha+1}}\right).
			\end{split}
		\end{equation*}
		Therefore, we have $\zeta_n \leq Cn^{-\frac{\alpha}{2\alpha+1}}$.
	\end{proof}

	\section{Technical Lemmas}\label{sec:lemma}
	\begin{Lemma}\label{lm:functional-Rademacher}
		Let $\mathcal F = \{f \in [0,1]: \|f\|_\infty \leq b\}$ and $\mathcal F_\alpha = \{f \in \mathcal F: \|f\|_{\cL^2} \leq \alpha\}$ for some $b,\alpha>0$. Let $s_1,\ldots,s_n \overset{i.i.d.}{\sim} \text{Unif}(0,1)$. Then for any $x>0$, 
		\begin{equation*}
			\begin{split}
				& \bbP\left(\sup_{f \in \mathcal F_\alpha} \left\{\int_0^1 f(t) dt - \frac{1}{n}\sum_{s=1}^n f(s_k)\right\} \leq 4\mathcal R_n \mathcal F_\alpha + \alpha\sqrt{\frac{2 x}{n}} + \frac{2bx}{n}\right) \geq 1 - e^{-x},
			\end{split}
		\end{equation*}
		where $\mathcal R_n \mathcal F_\alpha = \bbE\sup_{f \in \mathcal F_\alpha} \frac{1}{n}\sum_{k=1}^n\epsilon_k f(s_k)$ with $\{\epsilon_k\}_{k=1}^n$ being i.i.d. Rademacher variables.
	\end{Lemma}
	\begin{proof}
		The inequality comes from a direct application of \cite[Theorem 2.1]{bartlett2005local} by setting $a=0$, $r = \alpha^2$. 
	\end{proof}
	
	\bigskip
	\begin{Lemma}\label{lm:MC-int}
		Let $\{s_k\}_{k=1}^n$ be a collection of i.i.d. uniform random variables in $[0,1]$ and assume $\zeta_n < 1$. Then,
		\begin{equation*}
			\mathbb P\left(\left|\int_0^1 f(t)dt - \frac{1}{n}\sum_{k=1}^n f(s_k)\right|\leq C\left(\zeta_n\|f\|_{\cL^2} + \zeta_n^2 \|f\|_{\mathcal H}\right) ,\quad \forall f \in \mathcal H\right) \geq 1-Cn^{-9}.
		\end{equation*}
	\end{Lemma}
	\begin{proof}
		When $\|f\|_\cH = 0$, we have $f = 0$ and the proof is trivial; when $\|f\|_\cH > 0$, it suffices to show the inequality holds for all $f \in \cH$ such that $\|f\|_\cH = 1$. By the definition of $\zeta_n$, we know that $\sqrt{\log n/n} \leq \zeta_n < 1$. Therefore, for any $n \geq 2$, we can find some $J \in \mathbb N$ such that $2^{-J} \leq \zeta_n \leq 2^{-J+1}$, and $J \leq 1 + \frac{1}{2}\log_2(n/\log n) \leq 3\log n$.
		
		Now we apply Lemma \ref{lm:functional-Rademacher} sequentially. For any $j \in \{1,\ldots,J\}$, let $\cF_j:=\{f: \|f\|_\cH = 1, 2^{-j} \leq \|f\|_{\cL^2} \leq 2^{-j+1}\}$. Then we have with probability at least $1-e^{-x}$ that,
		\begin{equation}\label{ineq:Rn-peeling-j}
			\begin{split}
				\sup_{\substack{f \in \cF_j}}\left|\int_{0}^1 f(t)dt - \frac{1}{n}\sum_{k=1}^n f(s_k)\right| & \leq C\left(\mathcal R_n \cF_j + 2^{-j+1}\sqrt{\frac{x}{n}} + \frac{x}{n}\right) \\
				& \leq C\left(\zeta_n 2^{-j+1} + \zeta_n^2  + 2^{-j+1}\sqrt{\frac{x}{n}} + \frac{x}{n}\right) \\
				& \leq C\left(2\zeta_n \|f\|_{\cL^2} + \zeta_n^2  + 2\|f\|_{\cL^2}\sqrt{\frac{x}{n}} + \frac{x}{n}\right).
			\end{split}
		\end{equation}
		Here the first inequality comes from the facts that $\|f\|_{\cL^2} \leq 2^{-j+1}$ and $\|f\|_\infty \leq  \|f\|_\cH \leq 1$ for any $f \in \cF_j$, while the second inequality comes from the relationship between $\cR_n \cF_j$ and $\zeta_n$ (Proposition \ref{prop:Qn-complexity}) and the definition of $\zeta_n$: 
		\begin{equation*}
			\bbE \cR_n \cF_j \leq C Q_n(2^{-j+1}) \leq C\left(\zeta_n 2^{-j+1} + \zeta_n^2\right).
		\end{equation*}
		Setting $x = 10\log n$ (so that $\sqrt{x/n}\leq \sqrt{10}\zeta_n$), we have with probability at least $1-n^{-10}$ that
		\begin{equation}\label{ineq:Rn-peeling-term}
			\begin{split}
				\sup_{\substack{f \in \cF_j}}\left|\int_{0}^1 f(t)dt - \frac{1}{n}\sum_{k=1}^n f(s_k)\right| \leq C\left(\zeta_n \|f\|_{\cL^2} + \zeta_n^2  \right).
			\end{split}	
		\end{equation}
		Then by a union bound argument, 
		\begin{equation*}
			\bbP\left(\sup_{\substack{f: \|f\|_\cH=1, \|f\|_{\cL^2}\geq 2^{-J}}}\left|\int_{0}^1 f(t)dt - \frac{1}{n}\sum_{k=1}^n f(s_k)\right| \geq C\left(\zeta_n \|f\|_{\cL^2} + \zeta_n^2  \right)\right) \leq J n^{-10}.
		\end{equation*}
		For $f$ such that $\|f\|_{\cL^2} \leq 2^{-J} \leq \zeta_n$, we directly apply Lemma \ref{lm:functional-Rademacher}. By the similar argument as \eqref{ineq:Rn-peeling-j} and \eqref{ineq:Rn-peeling-term}, we know that with probability at least $1-n^{-10}$,
		\begin{equation}\label{ineq:Rn-small-term}
			\sup_{\substack{f: \|f\|_\cH = 1, \|f\|_{\cL^2}\leq 2^{-J}}}\left|\int_{0}^1 f(t)dt - \frac{1}{n}\sum_{k=1}^n f(s_k)\right| \leq C\left(2^{-J}\zeta_n + \zeta_n^2\right) \leq C\zeta_n^2.
		\end{equation}
		Combining \eqref{ineq:Rn-peeling-term} and \eqref{ineq:Rn-small-term}, we conclude that
		\begin{equation*}
			\bbP\left(\sup_{\substack{f: \|f\|_\cH=1}}\left|\int_{0}^1 f(t)dt - \frac{1}{n}\sum_{k=1}^n f(s_k)\right| \leq C\left(\zeta_n \|f\|_{\cL^2} + \zeta_n^2  \right)\right) \leq (J+1) n^{-10} \leq Cn^{-9},
		\end{equation*}
		and the proof is complete.
	\end{proof}
	
	\bigskip
	\begin{Lemma}\label{lm:MC-square-int}
		Let $\{s_k\}_{k=1}^n$ be a collection of i.i.d. uniform random variables on $[0,1]$ and $\zeta_n<1$. Then,
		\begin{equation*}
			\mathbb P\left( \left|\|f\|_{\cL^2}^2 - \frac{1}{n}\sum_{k=1}^n f^2(s_k)\right| \leq C\|f\|_\cH\cdot \left(\zeta_n\|f\|_{\cL^2} + \zeta_n^2\|f\|_\cH\right),\quad \forall f \in \cH \right) \geq 1-n^{-9}.
		\end{equation*}
	\end{Lemma}
	\begin{proof}
		Let $g := f^2$. Then $\|g\|_{\cL^2} = \|f^2\|_{\cL^2} \leq \|f\|_\infty \|f\|_{\cL^2} \leq \|f\|_\cH \|f\|_{\cL^2}$ and $\|g\|_\cH \leq C_\cH \|f\|_\cH^2$. Applying Lemma \ref{lm:MC-int} on $g$, we conclude the proof.
	\end{proof}
	
	\bigskip
	\begin{Lemma}\label{lm:Gaussian-complexity}
		Suppose $\varepsilon_{ijk} \overset{i.i.d.}{\sim} N(0,1)$, and $\{s_k\}_{k=1}^n$ are i.i.d. uniform random variables on $[0,1]$. Assume $\zeta_n \leq c_1$ for some sufficiently small constant $c_1$. Then, with probability at least $1-Cn^{-9}-C\log n \exp(-c(p_1+p_2))$, for any $f \in \cH$,
		\begin{equation*}
			\sup_{a \in \mathbb S^{p_1-1}, b \in \mathbb {S}^{p_2-1}} \frac{1}{n}\sum_{i=1}^{p_1}\sum_{j=1}^{p_2}\sum_{k=1}^n a_ib_jf(s_k)\varepsilon_{ijk} \lesssim \left(\zeta_n^2 + \frac{(p_1+p_2)\log n}{n}\right)\|f\|_\cH + \left(\zeta_n + \sqrt{\frac{p_1+p_2}{n}}\right)\|f\|_{\cL^2}.
		\end{equation*}
	\end{Lemma}
	\begin{proof}
		We assume $\|f\|_\cH = 1$ without loss of generality. Let $\{a^{(1)},\ldots, a^{(N_a)}\} \subset \bbS^{p_1-1}$ and $\{b^{(1)},\ldots, b^{(N_b)}\} \subset \bbS^{p_2-1}$ be the $\varepsilon$-nets of $\bbS^{p_1-1}$ and $\bbS^{p_2-1}$ respectively with $N_a \leq (3/\varepsilon)^{p_1}$ and $N_b \leq (3/\varepsilon)^{p_2}$ (they can be constructed based on \cite[Corollary 4.2.13]{vershynin2018high}). Then for any $i_a \in [N_a]$ and $i_b \in [N_b]$, we define the following quantities: 
		\begin{equation*}
			\begin{split}
				\hat Z_{i_a,i_b}(\varepsilon; \delta) & = \sup_{\substack{\|g\|_\cH \leq 1 \\ \|g\|_{\cL^2} \leq \delta}}\left|\frac{1}{n}\sum_{i=1}^{p_1}\sum_{j=1}^{p_2}\sum_{k=1}^n a_i^{(i_a)}b_j^{(i_b)}g(s_k)\varepsilon_{ijk}\right|. 
			\end{split}
		\end{equation*}	
		
		We aim to show the following concentration inequality of $ \hat Z_{i_a,i_b}(\varepsilon; \delta)$ for each pair of fixed $(i_a,i_b)$:
		\begin{equation}\label{ineq:concen-hat-Z}
			\begin{split}
				&\mathbb P\left(\hat Z_{i_a,i_b}(\varepsilon; \delta) > C\left(\zeta_n^2 + \frac{(p_1+p_2)\log n}{n} + \delta\left(\zeta_n + \sqrt{\frac{p_1+p_2}{n}}\right)\right),~~\exists \delta \in (0,1] \right) \\ & \qquad \leq C\log n\cdot \exp(-C(p_1+p_2)).
			\end{split}
		\end{equation}
		We first complete the proof given \eqref{ineq:concen-hat-Z}. By applying union bounds, we obtain that with probability at least $1-C\log n\exp\left(-c(p_1 + p_2)\right)$,
		\begin{equation}\label{ineq:conecn-hat-Z-max}
			\max_{i_a \in [N_a], i_b \in [N_b]} \hat Z_{i_a,i_b}(\varepsilon; \delta) \leq C\left(\zeta_n^2 + \frac{(p_1+p_2)\log n}{n} + \delta\left(\zeta_n + \sqrt{\frac{p_1+p_2}{n}}\right)\right), \qquad \forall \delta \in (0,1].
		\end{equation}
		Now for any $f$, let $a_f \in \bbS^{p_1-1}, b_f \in \bbS^{p_2-1}$ such that
		\begin{equation*}
			T := \sup_{a \in \mathbb S^{n-1}, b \in \mathbb {S}^{p-1}} \frac{1}{n}\sum_{i=1}^{p_1}\sum_{j=1}^{p_2}\sum_{k=1}^n a_ib_jf(s_k)\varepsilon_{ijk} = \frac{1}{n}\sum_{i=1}^{p_1}\sum_{j=1}^{p_2}\sum_{k=1}^n (a_f)_i(b_f)_jf(s_k)\varepsilon_{ijk}.
		\end{equation*}
		Let $l_a \in [N_a], l_b \in [N_b]$ such that $\|a_f - a^{(l_a)}\| \leq \varepsilon$, $\|b_f - b^{(l_b)}\| \leq \varepsilon$. Then we have
		\begin{equation*}
			\begin{split}
				T & = \frac{1}{n}\sum_{i=1}^{p_1}\sum_{j=1}^{p_2}\sum_{k=1}^n \left(a_i^{(l_a)} + (a_f)_i - a_i^{(l_a)}\right)\left(b_j^{(l_b)} + (b_f)_j - b_j^{(l_b)}\right)f(s_k)\varepsilon_{ijk} \\
				& \leq \hat Z_{i_a,i_b}(\varepsilon; \|f\|_{\cL^2}) + (\varepsilon+\varepsilon^2) T.
			\end{split}
		\end{equation*}
		Setting $\varepsilon = 1/2$, we obtain 
		\begin{equation*}
			T \leq 4\max_{i_a \in [N_a], i_b \in [N_b]} \hat Z_{i_a,i_b}(\varepsilon; \|f\|_{\cL^2}),
		\end{equation*}
		and the conclusion can be obtained by applying \eqref{ineq:conecn-hat-Z-max}.
		
		Now we focus on the proof for \eqref{ineq:concen-hat-Z}. First note that we can rewrite
		\begin{equation*} 
			\hat Z_{i_a,i_b}(\varepsilon, \delta) = \sup_{\substack{\|g\|_\cH \leq 1 \\ \|g\|_{\cL^2} \leq \delta}}\left|\frac{1}{n}\sum_{k=1}^n\left(\sum_{i=1}^{p_1}\sum_{j=1}^{p_2} a_ib_j \varepsilon_{ijk}\right)  g(s_k)\right|.
		\end{equation*}
		Since both $a^{(i_a)},b^{(i_b)}$ are unit vectors, $\left\{\sum_{i=1}^{p_1}\sum_{j=1}^{p_2} a_ib_j \varepsilon_{ijk}\right\}_{k=1}^n$ are always i.i.d. standard normal random variables and the distribution of $\hat Z_{i_a,i_b}(\varepsilon; \delta)$ is independent of the index $(i_a,i_b)$. Therefore, we only need to analyze the following random variable:
		\begin{equation*}
			\hat Z(\omega; \delta) := \sup_{\substack{\|g\|_\cH \leq 1 \\ \|g\|_{\cL^2} \leq \delta}}\left|\frac{1}{n}\sum_{k=1}^n w_k g(s_k)\right|,
		\end{equation*}
		where $w_k \overset{i.i.d.}{\sim} N(0,1)$. Denote $\cG_n(\delta) = \bbE \hat Z(\omega;\delta)$, which is the local Gaussian complexity. By Proposition \ref{prop:Qn-complexity}, we have 
		\begin{equation}\label{ineq:CQ-Q}
			\cG_n(\delta) \leq CQ_n(\delta) \leq C\left(\zeta_n\delta + \zeta_n^2\right),\qquad \forall \delta \in (0,1].
		\end{equation}
		Here $C$ is some universal constant. To bound the difference between $\hat Z(\omega;\delta)$ and $\cG_n(\delta)$, we apply Adamczak inequality~\citep{adamczak2008tail}. Let $f(\omega_k,s_k) = n^{-1}\omega_k g(s_k)$ and consider the class $\cF(\delta) = \{f: \|g\|_{\cH}\leq 1, \|g\|_{\cL^2} \leq \delta\}$. The weak variance of $\cF$ can be calculated:
		\begin{equation*}
			\begin{split}
				\sigma^2 & := \sup_{f \in \cF(\delta)} \sum_{k=1}^n \bbE f_k(\omega_k, s_k)^2 = \sup_{\substack{\|g\|_\cH \leq 1 \\ \|g\|_{\cL^2} \leq \delta}} n^{-2} \sum_{k=1}^n \bbE \omega_k^2 \bbE g(s_k)^2 \leq \delta^2/n.
			\end{split}
		\end{equation*}
		On the other hand, recall the $\psi_1$-Orlicz norm of a random variable $X$ is defined as $\|X\|_{\psi_1} = \inf\{\lambda > 0:\bbE\exp(|X|/\lambda)\leq 2\}$. We have
		\begin{equation*}
			\begin{split}
				\left\|\max_{k\in [n]} \sup_{f \in \cF(\delta)} |f(\omega_k,s_k)|\right\|_{\psi_1} & \leq C\log n\max_{k \in [n]} \left\|\sup_{f \in \cF(\delta)} |f(\omega_k,s_k)|\right\|_{\psi_1} \\
				& \leq \frac{C\log n}{n} \left\|\sup_{\substack{\|g\|_\cH \leq 1 \\ \|g\|_{\cL^2} \leq \delta}} g(s_1)|\omega_1|\right\|_{\psi_1} \\
				& \leq \frac{C\log n}{n} \cdot  \sup_{\substack{\|g\|_\cH \leq 1 \\ \|g\|_{\cL^2} \leq \delta}}\|g\|_\infty \cdot \left\||\omega_1|\right\|_{\psi_1} \leq \frac{C\log n}{n}.
			\end{split}
		\end{equation*}
		Here the first inequality comes from \cite[Theorem 4]{pisier1983some}. Then by applying Adamczak inequality~\cite[Theorem 4]{adamczak2008tail}, we obtain
		\begin{equation}\label{ineq:Adam-bound}
			\bbP\left(\hat Z(\omega;\delta) \geq 2\cG_n(\delta) + t\right) \leq \exp\left(-\frac{cnt^2}{\delta^2}\right) + 3\exp\left(-\frac{cnt}{\log n}\right).
		\end{equation}
		Set $t = C\left(\delta\sqrt{\frac{p_1+p_2}{n}} + \frac{(p_1+p_2)\log n}{n}\right)$ in \eqref{ineq:Adam-bound} and combine it with \eqref{ineq:CQ-Q}, then we proved for any fixed $\nu \in (0,1]$, 
		\begin{equation}\label{ineq:1-slice-peeling}
			\bbP\left(\hat Z(\omega;\delta) \geq C\left(\zeta_n^2 + \frac{(p_1+p_2)\log n}{n} + \delta\left(\zeta_n + \sqrt{\frac{p_1+p_2}{n}}\right)\right)\right) \leq C\exp(-5(p_1+p_2)).
		\end{equation}
		
		Now we extend above argument from fixed $\delta$ to uniform $\delta \in (0,1]$ via peeling technique. Let $J \in \mathbb N$ such that $2^{-J} \leq \zeta_n \leq 2^{-J+1}$. Since $\zeta_n \geq \sqrt{\log n /n}$, we have $J \leq C\log n$. By setting $\delta = \zeta_n, 2^{-J},2^{-J+1},\cdots,1/2,1$ in \eqref{ineq:1-slice-peeling}, we know that with probability at least $1-C(J+2)\exp(-5(p_1+p_2))$,
		\begin{equation}\label{ineq:peeling-slice-1}
			\begin{split}
				\hat Z(\omega,\zeta_n) & \leq C\left(\zeta_n^2 + \frac{(p_1+p_2)\log n}{n}\right),
			\end{split}
		\end{equation}
		\begin{equation}\label{ineq:peeling-slice-rest}
			\begin{split}
				\hat Z(\omega,2^{-j+1}) & \leq  C\left(\zeta_n^2 + \frac{(p_1+p_2)\log n}{n} + 2^{-j+1}\left(\zeta_n + \sqrt{\frac{p_1+p_2}{n}}\right)\right) ,\qquad 1 \leq j \leq J.     
			\end{split}
		\end{equation}
		
		Now for any $\delta \leq \zeta_n$, \eqref{ineq:peeling-slice-1} implies that
		\begin{equation}\label{ineq:leq-zeta}
			\hat Z(\omega, \delta) \leq \hat Z(\omega,\zeta_n) \leq C\left(\zeta_n^2 + \frac{(p_1+p_2)\log n}{n}\right).
		\end{equation}
		For $\delta \in (\zeta_n, 1]$, we can find some $j \in [J]$ such that $2^{-j} \leq \delta \leq 2^{-j+1}$. Then it follows by \eqref{ineq:peeling-slice-rest} that
		\begin{equation}\label{ineq:geq-zeta}
			\begin{split}
				\hat Z(\omega,\delta) & \leq \hat Z(\omega, 2^{-j+1}) \\
				& \leq 2C\left(\zeta_n^2 + \frac{(p_1+p_2)\log n}{n} + 2^{-j}\left(\zeta_n + \sqrt{\frac{p_1+p_2}{n}}\right)\right) \\
				& \leq 2C\left(\zeta_n^2 + \frac{(p_1+p_2)\log n}{n} + \delta\left(\zeta_n + \sqrt{\frac{p_1+p_2}{n}}\right)\right).
			\end{split}
		\end{equation}
		Therefore, \eqref{ineq:concen-hat-Z} is proved by combining \eqref{ineq:leq-zeta} and \eqref{ineq:geq-zeta}.
	\end{proof}
	
	\bigskip
	\begin{Lemma}\label{lm:hilbert-rescale}
		Let $\cH_0$ be a Hilbert space with inner products $\langle \cdot, \cdot \rangle_{\cH_0}$ and induced norm $\|\cdot\|_{\cH_0}$. Let $x_1, x_2$ be any two non-zero elements in $\cH_0$. Then,
		\begin{equation*}
			\sqrt{1 - \left\langle \frac{x_1}{\|x_1\|_{\cH_0}}, \frac{x_2}{\|x_2\|_{\cH_0}} \right\rangle_{\cH_0}^2} \leq \frac{\|x_1-x_2\|_{\cH_0}}{\|x_1\|_{\cH_0} \vee \|x_2\|_{\cH_0}}.
		\end{equation*}
	\end{Lemma}
	\begin{proof}
		By Arithmetic-Geometric mean inequality, we have
		\begin{equation*}
			\frac{2\langle x_1, x_2 \rangle_{\cH_0}}{\|x_1\|_{\cH_0}^2} \leq \frac{\langle x_1, x_2 \rangle^2_{\cH_0}}{\|x_1\|_{\cH_0}^2\|x_2\|_{\cH_0}^2} + \frac{\|x_2\|_{\cH_0}^2}{\|x_1\|_{\cH_0}^2}.
		\end{equation*}
		Rearrange this inequality and one obtains
		\begin{equation*}
			\begin{split}
				1 - \frac{\langle x_1, x_2 \rangle^2_{\cH_0}}{\|x_1\|_{\cH_0}^2\|x_2\|_{\cH_0}^2} & \leq 1 + \frac{\|x_2\|_{\cH_0}^2}{\|x_1\|_{\cH_0}^2} - \frac{2\langle x_1, x_2 \rangle_{\cH_0}}{\|x_1\|_{\cH_0}^2} \\
				& = \frac{\|x_1\|_{\cH_0}^2 + \|x_2\|_{\cH_0}^2 - 2\langle x_1, x_2 \rangle_{\cH_0}}{\|x_1\|_{\cH_0}^2} \\
				& = \frac{\|x_1-x_2\|_{\cH_0}^2}{\|x_1\|_{\cH_0}^2}.
			\end{split}
		\end{equation*}
		Changing the positions of $x_1$ and $x_2$, one can similarly obtain that
		\begin{equation*}
			\begin{split}
				1 - \frac{\langle x_1, x_2 \rangle^2_{\cH_0}}{\|x_1\|_{\cH_0}^2\|x_2\|_{\cH_0}^2} \leq \frac{\|x_1-x_2\|_{\cH_0}^2}{\|x_2\|_{\cH_0}^2}.
			\end{split}
		\end{equation*}
		The proof is completed by combining the above two inequalities.
	\end{proof}
	
	\bigskip
	\begin{Lemma}[Borell-TIS Inequality]\label{lm:Borell-TIS}
		Suppose $f_t$ is a centered Gaussian process, almost surely bounded on a compact set $\cT$. Let $\sigma^2 = \sup_{t \in \cT}\bbE f_t^2$. Then, for any $x\geq 0$,
		\begin{equation*}
			\bbP\left(\sup_t |f_t| - \bbE\sup_t |f_t| \geq x\sigma\right) \leq e^{-x^2/2}.
		\end{equation*}
	\end{Lemma}
	\begin{proof}
		See \cite[Theorem 2.1.1]{adler2009random}.
	\end{proof}
	
	\bigskip
	\begin{Lemma}[Heteroskedastic Wishart-type Concentration]\label{lm:hetero-wishart}
		Let $\bZ$ be a $p_1$-by-$p_2$ random matrix with independent mean-zero Gaussian entries and $\Var(\bZ_{ij}) = \sigma_{ij}^2$. Denote 
		\begin{equation*}
			\sigma_C^2 = \max_j \sum_{i=1}^{p_1}\sigma_{ij}^2,\qquad \sigma_R^2 = \max_i \sum_{j=1}^{p_2} \sigma_{ij}^2,\qquad \sigma_*^2 = \max_{i,j} \sigma_{ij}^2.
		\end{equation*}
		Then there exists some constant $C$, such that
		\begin{equation*}
			\bbP\left(\left\|\bZ\bZ^\top - \bbE\bZ\bZ^\top\right\| \geq C\left\{\left(\sigma_R + \sigma_C + \sigma_*\sqrt{\log(p_1 \wedge p_2)} + \sigma_* x\right)^2 - \sigma_R^2\right\}\right) \leq e^{-x}.
		\end{equation*}
	\end{Lemma}
	\begin{proof}
		See \cite[Theorem 5]{cai2020non-asymptotic}.
	\end{proof}
	
	\bigskip
	\begin{Lemma}\label{lm:independent-matrix-product}
		Let $\bZ_1 \in \bbR^{m_1 \times n}, \bZ_2 \in \bbR^{m_2 \times n}$ be two independent random matrices with i.i.d. $N(0,1)$ entries and $\bLambda = \diag(\lambda_1,\ldots, \lambda_n)$ with $\lambda_1 \geq \ldots \geq \lambda_n$. Then for any $t>0$,
		\begin{equation*}
			\bbP\left(\left\|\bZ_1 \bLambda \bZ_2^\top\right\| > C\left(\sqrt{(m_1+m_2+t)\sum_{k=1}^n \lambda_k^2} + (m_1+m_2+t)\lambda_1\right)\right) \leq \exp(-c(m_1+m_2+t)).
		\end{equation*}
	\end{Lemma}
	\begin{proof}
		Note that 
		\begin{equation*}
			\left\|\bZ_1 \bLambda \bZ_2^\top\right\| = \sup_{\substack{u \in \bbS^{m_1-1} \\ v \in \bbS^{m_2-1}}} u^\top \bZ_1\bLambda_1 \bZ_2^\top v.
		\end{equation*}
		We bound the above random variable via $\varepsilon$-net argument. Let $\{u^{(i)}\}_{i=1}^{N_1}$ and $\{v^{(j)}\}_{j=1}^{N_2}$ be the $\varepsilon$-net of $\bbS^{m_1-1}$ and $\bbS^{m_2-1}$ respectively such that
		\begin{equation*}
			N_1 \leq (\varepsilon/3)^{m_1},\qquad N_2 \leq (\varepsilon/3)^{m_2}.
		\end{equation*}
		For each fixed $(i,j) \in [N_1] \times [N_2]$, $\bZ_1^\top u^{(i)}$ and $\bZ_2^\top v^{(j)}$ are independent $n$-dimensional random vectors with distribution $N(0, \bI)$. Therefore,
		\begin{equation}
			u^{(i)\top} \bZ_1\bLambda \bZ_2^\top v^{(j)} \overset{d}{=} \sum_{k=1}^n \lambda_k x_ky_k,
		\end{equation}
		where $x_1,\ldots,x_n,y_1,\ldots,y_n$ are independent $N(0,1)$ random variables and $\overset{d}{=}$ denotes equal in distribution. By the Hanson-Wright's inequality~\citep{rudelson2013hanson}, we have
		\begin{equation}
			\bbP\left(u^{(i)\top} \bZ_1\bLambda \bZ_2^\top v^{(j)} > t\right) = \bbP\left(\sum_{k=1}^n \lambda_k x_k y_k > t\right) \leq 2\exp\left(-c\min\left\{\frac{t^2}{\sum_{k=1}^n \lambda_k^2}, \frac{t}{\lambda_1}\right\}\right).
		\end{equation}
		Replace $t$ with $C\sqrt{(m_1+m_2 + t)\sum_{k=1}^n \lambda_k^2} + (m_1+m_2 + t)\lambda_1 $, apply union bound on all the pairs $(i,j)$, and we obtain:
		\begin{equation}\label{ineq:lm-matrix-product-1}
			\begin{split}
				&\bbP\left(\max_{i,j} u^{(i)\top} \bZ_1\bLambda \bZ_2^\top v^{(j)} > \sqrt{(m_1+m_2+t)\sum_{k=1}^n \lambda_k^2} + (m_1+m_2+t)\lambda_1\right) \\
				& \qquad\qquad \leq 2\exp\left(-c(m_1+m_2+t)\right).
			\end{split}
		\end{equation}
		Now let $(u^*, v^*) = \argmax_{u,v} u^\top \bZ_1\bLambda_1 \bZ_2^\top v$, then there exists $(i^*,j^*) \in [N_1] \times [N_2]$ such that 
		\begin{equation*}
			\|u^* - u^{(i^*)}\| \leq \varepsilon, \qquad \|v^* - v^{(j^*)}\| \leq \varepsilon.
		\end{equation*}
		Then it follows that
		\begin{equation*}
			\begin{split}
				\left\|\bZ_1\bLambda\bZ_2^\top\right\| &= u^{*\top} \bZ_1\bLambda_1 \bZ_2^\top v^{*} \\
				& = u^{(i^*)\top} \bZ_1\bLambda_1 \bZ_2^\top v^{(j^*)} + \left(u^{*} - u^{(i^*)}\right)^\top \bZ_1\bLambda_1 \bZ_2^\top v^{(j^*)} + u^{*\top}\bZ_1\bLambda_1 \bZ_2^\top \left(v^* - v^{(j^*)}\right) \\
				& \leq u^{(i^*)\top} \bZ_1\bLambda_1 \bZ_2^\top v^{(j^*)} + 2\varepsilon \cdot \left\|\bZ_1\bLambda\bZ_2^\top\right\|.
			\end{split} 
		\end{equation*}
		Take $\varepsilon = 1/4$ and we obtain
		\begin{equation*}
			\left\|\bZ_1 \bLambda \bZ_2^\top\right\| \leq 2u^{(i^*)\top} \bZ_1\bLambda_1 \bZ_2^\top v^{(j^*)}.
		\end{equation*}
		Now the proof is completed by utilizing \eqref{ineq:lm-matrix-product-1}.
	\end{proof}

    \begin{Lemma}\label{lm:projection}
        Suppose $u,v$ are two vectors of the same dimension. For any $\lambda\in \mathbb{R}$, we have
        $$\dist(u, v) \leq \frac{\|u - \lambda v\|_2}{\|u\|_2}.$$
        Suppose $\xi, \eta \in \mathcal{L}_2$. For any $\lambda\in \mathbb{R}$, we have
        $$\dist(\xi, \eta) \leq \frac{\|\xi - \lambda \eta\|_{\mathcal{L}_2}}{\|\xi\|_{\mathcal{L}_2}}.$$
    \end{Lemma}
    \begin{proof}
        \begin{equation*}
            \begin{split}
                & \dist(u, v) = \sqrt{1 - \frac{\langle u, v\rangle^2}{\|u\|_2^2 \|v\|_2^2}} \leq \frac{\|u-\lambda v\|_2}{\|u\|_2}\\
                \Leftarrow  \quad & 1 - \frac{\langle u, v\rangle^2}{\|u\|_2^2\|v\|_2^2} \leq \frac{\|u\|_2^2 - 2\lambda\langle u, v\rangle + \lambda^2 \|v\|^2_2}{\|u\|_2^2}\\
                \Leftarrow \quad & 0 \leq \langle u, v\rangle^2 - 2\lambda\|v\|_2^2 \langle u, v\rangle + \lambda^2\|v\|_2^4 \\
                \Leftarrow \quad & 0 \leq \left(\langle u,v\rangle - \lambda\|v\|_2^2\right)^2.
            \end{split}
        \end{equation*}
        The proof for $\dist(\xi, \eta) \leq \frac{\|\xi - \lambda \eta\|_{\mathcal{L}_2}}{\|\xi\|_{\mathcal{L}_2}}$ follows similarly.
    \end{proof}
 
\end{document}